%% file: arxiv-ver2.tex
\documentclass[11pt]{article}
\usepackage[letterpaper, margin=1.0in]{geometry}
\usepackage[utf8]{inputenc}
\usepackage{tabularx}
\usepackage{microtype}
\usepackage{hyperref}
\usepackage{amsmath,amssymb,amsfonts,amsthm}
\usepackage{mathtools}
\usepackage{enumerate}
\usepackage{color}
\usepackage{xspace}
\usepackage{soul}
\usepackage{tikz}
\usepackage{thm-restate}
\usetikzlibrary{shapes,shapes.callouts,shadows,arrows,backgrounds,%
matrix,patterns,arrows,decorations.pathmorphing,decorations.pathreplacing,%
positioning,fit,calc,decorations.text,arrows.meta,quotes%
}

\usepackage{cases}
\usepackage{multirow}
\usepackage{caption,subcaption}
\usepackage{todonotes}

\theoremstyle{plain}
  \newtheorem{theorem}{Theorem}
  \newtheorem{corollary}[theorem]{Corollary}
  \newtheorem{observation}[theorem]{Observation}
  \newtheorem{lemma}[theorem]{Lemma}
  \newtheorem{proposition}[theorem]{Proposition}
  \newtheorem{claim}{Claim}[theorem]

  \newtheorem*{theorem*}{Theorem}
  \newtheorem*{corollary*}{Corollary}
  \newtheorem*{lemma*}{Lemma}
  \newtheorem*{proposition*}{Proposition}
  \newtheorem*{claim*}{Claim}
  
\theoremstyle{definition}
  \newtheorem{definition}[theorem]{Definition}
  \newtheorem{example}[theorem]{Example}

  \newtheorem*{definition*}{Definition}
  \newtheorem*{example*}{Example}
  \newtheorem*{question*}{Question}
  \newtheorem*{conjecture*}{Conjecture}
  \newtheorem*{remark*}{Remark}

\newenvironment{claimproof}
{\noindent {\em Proof of claim:} }
{\hfill $\diamond$ \smallskip}

\DeclarePairedDelimiter\norm{\lVert}{\rVert}

\mathchardef\mhyphen="2D

\newcommand*{\eps}{\varepsilon}

\newcommand*{\cA}{\mathcal{A}}

\newcommand*{\cC}{\mathcal{C}}

\newcommand*{\cF}{\mathcal{F}}

\newcommand*{\cI}{\mathcal{I}}

\newcommand*{\cP}{\mathcal{P}}

\newcommand*{\cT}{\mathcal{T}}

\newcommand*{\cY}{\mathcal{Y}}

\newcommand*{\FF}{\mathbb{F}}
\newcommand*{\GG}{\mathbb{G}}

\newcommand*{\NN}{\mathbb{N}}

\newcommand*{\QQ}{\mathbb{Q}}
\newcommand*{\RR}{\mathbb{R}}

\newcommand*{\ZZ}{\mathbb{Z}}

\newcommand*{\rx}{\mathsf{x}}
\newcommand*{\ry}{\mathsf{y}}
\newcommand*{\rz}{\mathsf{z}}

\newcommand{\bigoh}{\mathcal{O}}
\newcommand{\bigohs}{\mathcal{O}^*}
\newcommand{\cc}[1]{{\mbox{\textnormal{\textsf{#1}}}}\xspace}  %
\newcommand{\SB}{\{\,}
\newcommand{\SM}{\;{|}\;}
\newcommand{\SE}{\,\}}

\renewcommand{\P}{\cc{P}}
\newcommand{\NP}{\cc{NP}}
\newcommand{\FPT}{\cc{FPT}}

\newcommand{\Weft}{{\cc{W}}}
\newcommand{\W}[1]{{\Weft}{{[#1]}}}

\newcommand*{\ETH}{\textsc{ETH}\xspace}

\DeclareMathOperator{\ord}{ord}
\DeclareMathOperator{\lsu}{lsu}
\DeclareMathOperator{\nxt}{next}

\newcommand*{\GAP}[1]{\ensuremath{\textsc{Gap}_{#1}}}

\newcommand*{\rep}[1]{\ensuremath{#1^{\sim}}}
\newcommand*{\res}[1]{\ensuremath{#1^+}}

\DeclareMathOperator{\cost}{cost}
\DeclareMathOperator{\mincost}{mincost}

\usepackage{boxedminipage}
\newcommand{\pbDef}[3]{%
  \noindent
  \begin{center}
  \begin{boxedminipage}{0.98 \columnwidth}
  {\sc #1}\\[5pt]
  \begin{tabular}{l p{0.70 \columnwidth}}
  {\sc Instance}: & #2\\
  {\sc Question}: & #3
  \end{tabular}
  \end{boxedminipage}
  \end{center}
}

\newcommand{\pbDefP}[4]{%
  \noindent
  \begin{center}
  \begin{boxedminipage}{0.98 \columnwidth}
  {\sc #1}\\[5pt]
  \begin{tabular}{l p{0.70 \columnwidth}}
  {\sc Instance}: & #2\\
  {\sc Parameter:} & #3\\
  {\sc Question}: & #4
  \end{tabular}
  \end{boxedminipage}
  \end{center}
}

\newcommand{\pbDefGap}[5]{%
  \noindent
  \begin{center}
  \begin{boxedminipage}{0.98 \columnwidth}
  {\sc #1}\\[5pt]
  \begin{tabular}{l p{0.80 \columnwidth}}
    {\sc Instance}: & #2\\
    {\sc Parameter:} & #3\\
    {\sc Goal}: & Distinguish between the following cases: \\
    {\quad (YES)} & #4 \\
    {\quad (NO)} & #5
  \end{tabular}
  \end{boxedminipage}
  \end{center}
}

\newcommand*{\csp}[1]{\ensuremath{\textsc{CSP}(#1)}}
\newcommand*{\mincsp}[1]{\ensuremath{\textsc{MinCSP}(#1)}}
\newcommand*{\lin}[2]{\textsc{\ensuremath{#1}-Lin(\ensuremath{#2})}}
\newcommand*{\linideal}[3]{\textsc{\ensuremath{#1}-Lin(\ensuremath{#2})-over-\ensuremath{#3}}}
\newcommand*{\minlin}[2]{\textsc{Min-\ensuremath{#1}-Lin(\ensuremath{#2})}}

\newcommand*{\mintwolin}{\textsc{Min-2-Lin}}
\newcommand*{\minthreelin}{\textsc{Min-3-Lin}}

\newcommand*{\minlinideal}[3]{\textsc{Min-\ensuremath{#1}-Lin(\ensuremath{#2})-over-\ensuremath{#3}}}

\DeclareMathOperator{\Ann}{Ann}

\newcommand*{\ringelement}[1]{\mathsf{#1}}

\def \rx {\ringelement{x}\xspace}
\def \ry {\ringelement{y}\xspace}

\mathchardef\mhyphen="2D
\newcommand*{\reals}{\mathbb{R}}
\newcommand*{\integers}{\mathbb{Z}}
\newcommand*{\rationals}{\mathbb{Q}}
\newcommand*{\naturals}{\mathbb{N}}
\newcommand*{\range}[2]{\ensuremath{\{#1,\dots,#2\}}}

\newcommand*{\GAPMLD}{\ensuremath{\textsc{Gap}_{\gamma}\mhyphen\textsc{MLD}_{p}}}

\newcommand{\BA}[1]{#1_G}
\newcommand{\con}{\textsf{eqn}}
\newcommand{\sep}{\textsf{sep}}
\newcommand{\ed}{\textsf{ed}}
\newcommand{\comp}[1]{\overline{#1}}
\newcommand{\EQ}{\Gamma}
\newcommand{\CCC}{\mathcal{C}}
\newcommand{\compI}[2]{#1_{#2}}

\newcommand{\clasn}[1]{\tau_{#1}}

\usepackage[noend]{algpseudocode}
\usepackage{algorithm,algorithmicx}

\algnewcommand\algorithmicinput{\textbf{Input:}}
\algnewcommand\INPUT{\item[\algorithmicinput]}

\algnewcommand\algorithmicoutput{\textbf{Output:}}
\algnewcommand\OUTPUT{\item[\algorithmicoutput]}

\title{Towards a Parameterized Approximation Dichotomy of MinCSP for Linear Equations over Finite Commutative Rings}

 \author{
   Konrad K. Dabrowski\thanks{Newcastle University, UK, \texttt{konrad.dabrowski@newcastle.ac.uk}} \and
   Peter Jonsson\thanks{Link{\"o}ping University, Sweden, \texttt{peter.jonsson@liu.se}} \and
   Sebastian Ordyniak\thanks{University of Leeds, UK, \texttt{sordyniak@gmail.com}} \and
   George Osipov\thanks{Link{\"o}ping University, Sweden, \texttt{george.osipov@pm.me}} \and
   Magnus Wahlstr{\"o}m\thanks{Royal Holloway, University of London, UK, \texttt{Magnus.Wahlstrom@rhul.ac.uk}}
 }

\date{}

\begin{document}

\maketitle


\begin{abstract}
We consider the $\minlin{r}{R}$ problem: for a system $S$ of 
length-$r$ linear equations
over a ring $R$, find 
$Z \subseteq S$ of minimum cardinality
such that $S-Z$ is satisfiable. The problem is \NP-hard and
UGC-hard to approximate within any constant even when
$r = |R| = 2$, so we focus on parameterized approximability with solution size as the parameter.
For a large class of infinite rings $R$ called {\em Euclidean domains}, 
Dabrowski, Jonsson, Ordyniak, Osipov, and Wahlström [SODA-2023] obtained
an FPT-algorithm for $\minlin{2}{R}$ based on Wahlström's [SODA-2017] LP-based approach.
In this paper, we consider $\minlin{r}{R}$ for finite commutative rings $R$,
initiating a line of research with the goal of
separating problems that are FPT-approximable within a constant from those that are not.
A major motivation is that our project is a promising step
(and a necessary consequence) for even more ambitious 
classification projects concerning finite-domain MinCSP and VCSP.

Dabrowski et al.'s algorithm is limited to domains, i.e. rings without zero divisors
($a \cdot b = 0$
  implies $a = 0$ or $b = 0$),
which are only fields among finite commutative rings.
  Handling zero divisors
  is a severe obstacle for the LP-based approach. 
In response, we introduce a class of finite commutative
rings ({\em Bergen rings}) and develop a constant-factor FPT-approximation algorithm for them, thus proving approximability for chain rings, principal ideal rings, and $\ZZ_m$ for all $m \geq 2$.
  We present a novel domain abstraction method that iteratively solves tighter and tighter relaxations
  by abstracting away parts of the domain 
  using carefully chosen ring element equivalences.
  In each level, equations are reformulated as
  graphs where solutions can be 
  identified with a particular collection of cuts.
  The abstractions may hide obstructions 
  so that they only become visible at later levels and this
  makes it difficult to find an optimal solution.
  To deal with this, we use a strategy based on shadow removal [Marx \& Razgon, STOC-2011] to compute 
  solutions that (1) are at most a constant factor
  away from the optimum and (2) allow us to reduce the domain for all variables simultaneously.
  We complement the algorithmic result with powerful lower bounds.
  For $r \geq 3$, we show that the problem 
  is not FPT-approximable within any constant (unless $\FPT =  \W{1}$) by
  exploiting connections to coding theory.
  We identify the class of
  \emph{non-Helly} rings
  for which $\minlin{2}{R}$ is not FPT-approximable because
  the binary equations over such rings can encode
  long equations over finite fields. 
  Under \ETH, we also rule out $(2-\eps)$-approximation 
  for every $\eps > 0$ for \emph{non-lineal} rings. 
  This class includes e.g. rings $\ZZ_{pq}$
  where $p$ and $q$ are distinct primes,
  and is tight for such rings.    
  Our work lays the foundation for a geometric approach based on the convex discrete analysis framework by Murota and others. This approach proved to be crucial for relating
various ring-theoretic concepts and constructing informative
examples when charting the approximability landscape.
\end{abstract}

%


\newpage

\vspace*{-1.5cm}
\tableofcontents

\section{Introduction}

We begin this introductory section by providing a ring-centric background and then broaden the perspective into CSP territory.
We continue by presenting our contributions, providing a technical overview of the results, and finally outline the structure of the paper.

\paragraph{Ring background.}
Systems of linear equations 
are ubiquitous~\cite{grcar2011ordinary} and
methods such as Gaussian elimination can solve such systems 
over various algebraic structures.
There is a plethora of applications for
classical infinite-domain rings such as the rationals ${\mathbb Q}$
and the integers ${\mathbb Z}$. 
Finite rings also have many important applications: coding theory is a striking example where finite fields (and other rings) are fundamental.
We do not attempt to survey this research field but recommend the textbook~\cite{Bini:Flamini:FCR}.
Even more exotic rings are important in coding theory
(see 
the textbook~\cite{Shi:etal:CaR}):
{\em chain rings} 
\cite{Dougherty:etal:ijict2010,Honold:Landjev:ejc2000} and {\em principal ideal rings}~\cite{Dougherty:etal:dcc2009,Kalachi:Kamche:amc2023} (including
${\mathbb Z}_m$ when $m$ is not a prime power~\cite{Blake:ic72}) are concrete examples.
Such rings have additionally been used in work on, for instance, 
Gröbner bases~\cite{Eder:Hofmann:jsc2021,Norton:Salagean:endm2001} and
other kinds of algebraic computation~\cite{Georgieva:sjc2016,Mikhailov:Nechaev:dma2004}.
All of the ring families mentioned above will be encountered later
in this paper.
Linear equation systems over finite rings are solvable
in polynomial time~\cite[Corollary 5.4.2]{Jeavons:etal:jacm97} (this result
was later rediscovered using a different approach~\cite{Arvind:Vijayaraghavan:stacs2005}) but these methods are not well suited for handling inconsistent systems.
We let
$\minlin{r}{R}$ denote the problem where the input is
a system of linear equations (with at most $r$ variables per equation) over a ring $R$ together
with an integer $k$,
and the question is if the system can be made consistent by removing at most $k$ equations.

Let us now take a look at $\minlin{r}{R}$.
The problem $\minlin{2}{R}$ is \NP-hard for all non-trivial
commutative finite rings by the VCSP dichotomy result~\cite{Kolmogorov:etal:sicomp2017}.
$\minlin{2}{R}$ seems hard to polynomial-time approximate, too:
it is conjectured that $\mintwolin$ for finite fields
is not polynomial-time approximable within any constant under
the Unique Games Conjecture (UGC); 
see Definition~3 in~\cite{khot2016candidate} and the following discussion.
Concerning the parameterized complexity of
$\minlin{r}{R}$, an early result shows
that
$\minlin{2}{\ZZ_2}$ is in \FPT\
while $\minlin{3}{\ZZ_2}$ is \W{1}-hard~\cite{crowston2013parameterized}.
This is generalized in~\cite{Dabrowski:etal:soda2023}:
$\minlin{3}{R}$ is \W{1}-hard for {\em all} non-trivial
rings and $\minlin{2}{R}$ is in \FPT\ when $R$ is a {\em Euclidean domain}.
While covering many important examples
(including $\rationals$ and $\integers)$,
some of the most well-known rings (such as $\ZZ_m$ with $m$ non-prime) are not Euclidean domains---a {\em domain} is a ring not
containing
{\em zero divisors}, i.e. 
$ab = 0$ implies $a = 0$ or $b = 0$.
Wedderburn's Little Theorem  (see~\cite{Herstein:amm61}) states that a finite commutative ring is either
a field or it contains zero divisors,
so domains are rare among finite commutative
rings.
Zero divisors often lead to \W{1}-hardness: It is shown in~\cite{Dabrowski:etal:soda2023} that
whenever a commutative ring $R$ is the direct sum of
nontrivial rings, then
 \minlin{2}{R} is \W{1}-hard. This implies (via the Chinese Remainder Theorem) that 
$\minlin{2}{\ZZ_m}$ is \W{1}-hard if $m$ contains two distinct prime factors.

\paragraph{CSP background.} 
The {\em constraint satisfaction problem} (CSP) is a framework for
combinatorial problems. A CSP$(\Gamma)$ has a domain $D$ and
a set of relations $\Gamma$ (the {\em constraint language}) 
that specifies allowed constraints.
The input is a constraint set
and the question is if there is an assignment of values from $D$ that satisfies them. 
Examples include fundamental problems such as $k$-SAT and $k$-{\sc Colouring}.
This framework is suited for complexity characterizations for all problems in a class of constraint languages---so-called {\em dichotomy theorems}. One example is the celebrated result by (independently) Bulatov~\cite{Bulatov:focs2017} and Zhuk~\cite{Zhuk:jacm2020}:
if the domain of $\Gamma$ is finite, then CSP$(\Gamma)$ is either in \P or \NP-complete.
Dichotomy results are important since they provide fundamental information about
the problem class at hand and 
push the algorithmic toolbox to its limit.

When a CSP$(\Gamma)$ instance is not satisfiable, one may look for an
assignment that violates as few constraints as possible---we let $\mincsp{\Gamma}$
denote this generalization of $\minlin{r}{R}$.
 The computational complexity of finite-domain  $\mincsp{\Gamma}$ is fully understood due to the \P/\NP-hardness dichotomy~\cite{Kolmogorov:etal:sicomp2017}.
Unfortunately, it is typically \NP-hard, so it is an obvious target for applying parameterized complexity. Indeed, taking the number of unsatisifed constraints $k$
as the parameter, many cases of $\mincsp{\Gamma}$ are in \FPT: a classical example is the \textsc{Almost 2-SAT} problem~\cite{Razgon:O'Sullivan:jcss2009}.
While there is an FPT/W[1]-hardness dichotomy for the two-element domain~\cite{Kim:etal:soda2023}, the general picture is unclear.
This motivates studying $\minlin{r}{R}$, since a complete dichotomy for 
$\mincsp{\Gamma}$ is impossible without a
dichotomy theorem for $\minlin{r}{R}$.

Providing a dichotomy theorem for $\minlin{r}{R}$ may be very difficult. 
The FPT dichotomy for  \textsc{MinCSP} in the Boolean case~\cite{Kim:etal:soda2023}, which is normally relatively easy to handle compared to other settings, arguably took a decade to finish, as it involved settling the FPT status of \textsc{$\ell$-Chain SAT} and \textsc{Coupled Min-Cut} by developing \emph{flow augmentation}; see discussion in~\cite{KimKPW21flow,Kim:etal:soda2023}. This result also relies on the structure of Boolean languages ({\em Post's lattice}), which appear tamer than the structure of languages based on equations over finite commutative rings. 
Constant-factor FPT-approximable Boolean MinCSPs were separated
from those that are not almost ten years earlier in ~\cite{Bonnet:etal:esa2016}.
This success story suggests that analysing approximability
of MinCSP may be advantageous.
Let $c \geq 1$ be a constant. A factor-$c$ {\em FPT-approximation
algorithm} for $\mincsp{\Gamma}$
takes an instance
$(I,k)$ and
either returns that there is no solution of size at most $k$ or returns that there is 
a solution of size at most $c \cdot k$. The running time of the algorithm is bounded by 
$f(k) \cdot \norm{I}^{O(1)}$ where $f: \naturals \rightarrow \naturals$
is some computable function. Thus, there is more
time to compute the solution (compared to polynomial-time
approximation) but the algorithm may output an oversized solution (unlike an exact FPT algorithm).
This combination has received a rapidly increasing interest.
Well-known algorithmic contributions include 
methods for designing FPT approximation algorithms
for problems that are in \FPT but with faster running 
times~\cite{Lokshtanov:etal:soda2021},
and FPT approximation algorithms for graph
parameters, see~\cite{Korhonen:focs2021}.
Notable hardness results are the ones showing that {\sc Clique}, {\sc Biclique} and {\sc Dominating Set} (and many related problems) admit no FPT approximation algorithm (within various factors)~\cite{Chalermsook:etal:sicomp2020,Lin:jacm2018}
and PCP-like approaches to FPT inapproximability~\cite{Guruswami:etal:stoc2024}.
More on FPT approximability can be found in
the surveys~\cite{Feldmann:etal:algorithms2020} and~\cite{Marx:tcj2008}. 

If $\mincsp{\Gamma}$ is not FPT-approximable within
a constant, then $\mincsp{\Gamma}$ is not in \FPT, and hardness results for FPT-approximability are stronger than hardness results for 
exact solvability in FPT time.
However, it may be easier to
study algorithms that achieve constant-factor approximations in FPT time
compared to exact FPT algorithms.
One reason is pointed out in~\cite{Bonnet:etal:esa2016}: 
the constant-factor FPT approximability of \textsc{MinCSP} is closed under slightly restricted {\em pp-definitions} and is, at least
in principle, amenable to powerful algebraic methods~\cite{Barto:etal:polymorphisms}. For exact solvability of \textsc{MinCSP}s, we have
a less powerful toolbox, i.e. constructions such as {\em proportional implementations}~\cite{Khanna:etal:sicomp2000,Kim:etal:soda2023} and an algebraic
theory based on {\em fractional polymorphisms}~\cite{Kolmogorov:etal:sicomp2017}.
Another reason is that constant-factor FPT-approximability
is preserved under certain decompositions of rings.
Suppose $R_i$, $i \in \{1,2\}$, are rings such that $\minlin{2}{R_i}$ is constant-factor FPT-approximable.
If $R=R_1 \oplus R_2$ is the direct sum of rings $R_1,R_2$, then
$\minlin{2}{R}$ is also constant-factor FPT-approximable, by Proposition~\ref{prop:sumapproximation}.
A similar result does not hold if we consider exact solvability in FPT time: MinCSP$(\ZZ_2 \oplus \ZZ_3)$
is \W{1}-hard ~\cite[Section~6.2]{Dabrowski:etal:soda2023} while both MinCSP$({\mathbb Z}_2)$ and
MinCSP$({\mathbb Z}_3)$ are in \FPT. Decomposition of rings is one
of our main themes; for instance, our FPT approximability algorithm  is based on decomposing a ring along a chain of ideals.

\paragraph{Summary of results.}
We initiate a project for proving dichotomy theorems
for $\minlin{r}{R}$ that separate (1)
problems that are in \FPT from those that are not, and
(2)
problems that are FPT-approximable within a constant from those that are not.
The tracks are not independent: hardness results from track (2) are directly applicable within track (1), and algorithmic results from track (2) may be refined into exact
algorithms. A prime example is the 
{\sc Multicut} problem where a factor-2 FPT-approximability
algorithm by Marx and Razgon~\cite{marx2009constant}
was later refined (by the same authors) into an exact algorithm~\cite{marx2014fixed}. A CSP-oriented example is
the {\sc Coupled Min Cut} problem---a notorious problem known as
a barrier to the study of \textsc{MinCSP}. It was early noted to have
a simple factor-2 FPT approximation algorithm (basically by combining \cite{Bonnet:etal:esa2016} and \cite{marx2014fixed})
but an exact algorithm was elusive until
the introduction of flow augmentation many years later~\cite{KimKPW21flow}.

We begin with some terminology. 
Let $(R,+,\cdot)$ be a finite commutative ring.
An {\em ideal} in $R$ is a subset $I \subseteq R$ such that
(1) $(I,+)$ is a subgroup of $(R,+)$ and 
  (2) for every $d \in R$ and every $x\in I$, 
  the product $dx$ is in $I$. Given elements $r_1,\dots,r_m \in R$, we let $(r_1,\dots,r_m)$ denote
the ideal generated by these elements.
The ring $R$ is a field if and only if $R$ only has the two trivial ideals $\{0\}$ and $R$, while
$R$ is a {\em chain ring} if its ideals are totally ordered under set inclusion.

For arbitrary $A \subseteq R$,
the {\em annihilator} of $A$ is the set $\Ann(A)=\{b \in R \; | \; ab = 0 \text{ for all } a \in A\}$, and this set
is always an ideal.
A ring is {\em lineal} if its annihilators are totally ordered under set inclusion~\cite{Marks:Mazurek:ijm2016}.
An annihilator of the form $\Ann(\{a\})$ is a {\em one-element} annihilator.
A \emph{coset} of an ideal $I$
is a set $a + I = \{ a + i \mid i \in I \}$
with $a \in R$.
A triple $C_1, C_2, C_3$ of one-element annihilator cosets of $R$
forms a \emph{tangle} if
$C_i \cap C_j \neq \emptyset$
for all $i,j \in \{1,2,3\}$ with $i \neq j$,
and $C_1 \cap C_2 \cap C_3 = \emptyset$.
We say that $R$ is \emph{Helly}
if it does not admit a tangle.
The term Helly has been used with varying meanings when studying different
algebraic structures: for instance,
Domokos and Szabó~\cite{Domokos:Szabo:jlms2011} use the full set of cosets in Abelian groups when defining Helly properties. We stress that we only consider the Helly property
on cosets of one-element annihilators in this paper.

A commutative ring is {\em local} if it has a unique
maximal ideal (an ideal $I \subsetneq R$ is maximal if there is
no ideal $J \subseteq R$ such that $I \subsetneq J \subsetneq R$).
Local rings are the building blocks of more complex rings. The {\em direct sum} of two rings 
 $R_1=(D_1;+_1,\cdot_1)$ and $R_2=(D_2;+_2,\cdot_2)$ 
 is denoted $R_1 \oplus R_2 = (R; +, \cdot)$.
 Its domain $R$ consists of the ordered pairs 
 $\{(d_1,d_2) \mid d_1 \in D_1, d_2 \in D_2\}$ 
 and the operations are defined coordinate-wise:
 $(d_1,d_2)+(d'_1,d'_2)=(d_1+_1 d'_1,d_2+_2 d'_2)$ and
 $(d_1,d_2) \cdot (d'_1,d'_2)=(d_1 \cdot_1 d'_1,d_2 \cdot_2 d'_2)$.

\begin{theorem}[Theorem 3.1.4 in \cite{Bini:Flamini:FCR}] \label{thm:sumoflocalrings}
Every finite commutative ring $R$ is isomorphic to a
direct sum $\bigoplus _{i=1}^{n} R_{i}$, where each $R_i$ is a commutative
local ring.
\end{theorem}

Let us now summarize our results.
We introduce {\em Bergen rings} (named after the city where they were invented), which form the basis for our algorithmic results.
These rings are defined via
a particular decomposition along a chain of ideals
such that levels are related in a structure-preserving way.
The definition alone provides limited guidance
for constructing MinCSP algorithms;
our algorithmic results can be viewed as an investigation
into exploiting structural decompositions for MinCSP algorithms.
The formal definition
is somewhat involved so we defer it to Section~\ref{sec:bergenrings};
an informal description will be given shortly.
We use the following chain of ring inclusions when presenting
our results.
\[{\rm field} \subsetneq {\rm chain} \subsetneq {\rm Bergen} \subseteq {\rm lineal} \subsetneq {\rm Helly}.\]
The fields are a strict subset of the chain rings since fields only contain
two ideals.
The other inclusions follow from Lemma~\ref{lem:chain-is-bergen}, Proposition~\ref{prop:bergenislineal} and 
Proposition~\ref{prop:lineal-is-helly}, respectively.
We have concrete rings that separate the classes of chain and Bergen rings
($\ZZ_2[\rx,\ry]/(\rx^2,\rx\ry,\ry^2)$ from Section~\ref{sec:bergenrings}) and the classes of lineal and Helly
rings ($\ZZ_2[\rx,\ry] / (\rx^3,\rx\ry,\ry^3)$ from Section~\ref{sec:annihilators}), but it is currently open whether
every lineal ring is Bergen or not.
There are large interesting classes of rings that are Bergen but not chain rings: one example is given in Lemma~\ref{lem:simplebergenexample} and more elaborate examples are presented in Section~\ref{sec:geometry}.
Finally,
lineal rings are local (Proposition~\ref{prop:lineal-is-local}) and local rings need not be 
 Helly
(as witnessed by $\ZZ_2[\rx,\ry] / (\rx^2,\ry^2)$ from Section~\ref{sec:annihilators}).

Arbitrarily pick a finite, commutative, non-trivial ring $R$.
We prove the following results. 

\medskip

\noindent
\fbox{
\parbox{0.95\textwidth}{

\begin{itemize}
\item
$\minlin{r}{R}$ is not FPT-approximable within any constant when $r \geq 3$ unless \FPT =  \W{1}
(Corollary~\ref{cor:gap-min3lin(ring)-is-hard}).

\item
If $R$ is not local,
then $R$ is the direct sum of local rings
$R_1,\dots,R_n$ (Theorem~\ref{thm:sumoflocalrings}) and
$\minlin{2}{R}$ is FPT-approximable within a constant if and only if $\minlin{2}{R_1},\dots,\allowbreak\minlin{2}{R_n}$ are FPT-approximable within a constant (Proposition~\ref{prop:sumapproximation}).

\item
If $R$ is local and not Helly, then
$\minlin{2}{R}$ is not FPT-approximable within any constant unless \FPT = \W{1} (Theorem~\ref{thm:non-helly-hard}).

\item
If
$R$ is local, Helly and not lineal, then
$\minlin{2}{R}$ 
is not FPT-approximable within $2-\eps$ for any $\eps > 0$
unless the {\em Exponential-Time Hypothesis} (ETH) is false~(Theorem~\ref{thm:incomparable-annihilators}).

\item
If $R$ is local, Helly and lineal, then
$\minlin{2}{R}$ is FPT-approximable within some constant if $R$ is Bergen (Theorem~\ref{thm:bergen-algo}).\\
\end{itemize}
\vspace*{-5mm}
}
}

\medskip
\noindent
The FPT-inapproximability for non-Helly rings is particularly
surprising, since one might naively have expected the language of
binary linear equations over $R$ to have a \emph{majority
polymorphism},
(see~\cite{Barto:etal:polymorphisms}), 
which would exclude all
constructions we are aware of 
for
FPT-inapproximability beyond constant factors;
see Theorem~48 and Conjecture~50 of~\cite{Barto:etal:polymorphisms}.

Every ring $\ZZ_{p^q}$ with $p$ prime is a chain ring and
the Chinese Remainder Theorem immediately implies that $\minlin{2}{{\mathbb Z}_m}$ is
FPT-approximable
within a constant for arbitrary $m \in {\mathbb N}$.
A {\em principal ideal ring} is a ring
where every ideal is generated by a single element.
Our results imply that $\minlin{2}{R}$ is FPT-approximable
within a constant whenever $R$ is a finite and commutative principal ideal ring
since such rings are direct sums of chain rings (see \cite{Hungerford:pjm68} or
\cite[Proposition 2.7]{Dougherty:etal:dcc2009}).

We continue by describing the technical contributions behind the results
outlined above.

\paragraph{FPT approximation algorithm.}
Let $R$ be a finite commutative ring.
By iterative compression
and branching, 
we can reduce $\minlin{2}{R}$ to
a version with \emph{simple} instances:
all binary equations are of the form
$x = r \cdot y$ with $r \in R$ and
all unary equations are crisp (i.e. undeletable) and of the form
$x = r$ for some $r \in R$.
Observe that the binary equations are homogeneous,
so the all-zero assignment satisfies them,
with the only obstacle being the unary equations.

We start by describing a constant-approximation
algorithm for $\minlin{2}{R}$ when $R$ is a field because 
it nicely illustrates our approach.\footnote{We remark that $\minlin{2}{\FF}$ for finite fields $\FF$
can be solved in FPT time~\cite{ChitnisCHPP16contract,Dabrowski:etal:soda2023,iwata2016half}.}
Let $(S, k)$ be a simple instance and
construct a graph $G = G(S)$ with vertices $x_i$ 
for every $x \in V(G)$ and $i \in R \setminus \{0\}$
and special vertices $s$ and $t$.
For unary equations in $S$,
add undeletable edges $sx_r$ if the equation
is the form $x = r$ for some $r \neq 0$,
and undeletable edges $x_rt$ for every $r \in R \setminus \{0\}$
if the equation is of the form $x = 0$ 
(to make an edge undeletable, 
it is sufficient to add many parallel copies).
For every binary equation $e$ of the form $x = r \cdot y$,
construct a bundle of edges $B_e = \{x_{ri} y_{i} : i \in R, ri \neq 0\}$
and add these edges to $G$.
This completes the construction.
There is a useful correspondence between certain
$st$-cuts in $G(S)$ and assignments to $S$.
Formally, for a set of vertices $X \subseteq V(G)$,
let $\delta(X)$ be the set of edges with exactly one
endpoint in $X$, i.e. $\delta(X)$ is the cut
separating $X$ and $\overline{X}$.
If $U \subseteq V(G)$ is such that
$s \in U$, $t \notin U$ and there is at most one 
vertex $x_i \in U$ for any $x \in V(S)$,
we say that $\delta(U)$ is a \emph{conformal $st$-cut}.
Conformal cuts in $G$ correspond to assignments to $S$:
to construct an assignment for a conformal cut $U$,
set $x$ to $i$ if $x_i \in U$ for some $i \in R \setminus \{0\}$,
and set $x$ to $0$ otherwise;
the reverse direction of the correspondence is obvious.

Let $b(U)$ be the number of edge bundles $B_e$
intersected by $\delta(U)$.
Then $R$ being a field implies that
$b(U)$ is exactly the cost of the assignment corresponding to $U$.
More specifically, consider an equation $e = (x = r \cdot y)$,
suppose $x_i, y_j \in U$ and $B_e \cap \delta(U) = \emptyset$.
Then $i = rj$ because the edge $x_{rj} y_{j}$ is uncut
and both its endpoints are reachable from $s$, 
so the assignment corresponding to $U$ satisfies $e$.
This guarantee is represented by $B_e$ being a matching:
think of an edge $x_i y_j$ as 
an encoding of the constraint 
``$x = i$ if and only if $y = j$''.
To complete the algorithm for fields, it suffices
to compute a conformal cut $U$ with $b(U) \leq k$.
Observe that $b(U) \leq k$ implies $\delta(U) \leq (|R|-1)k$
because the bundles $B_e$ are of size $|R|-1$ and disjoint.
One obtains an $(|R|-1)$-factor 
approximation in single-exponential FPT time
using a branching algorithm in the style of
\textsc{Digraph Pair Cut}~\cite{kratsch2020representative}.
If $R$ is not a field, then
the bundles $B_e$ stop being matchings:
simply consider 
an equation $e = (x = 2y)$ over $\ZZ_8$
(second graph from the left in Figure~\ref{fig:z8-matchings}). 
Note that both $y_2$ and $y_6$ are adjacent to $x_4$ in $B_e$.
Moreover, if $x=4$ then we require either $y=2$ or $y=6$, hence the 
dependencies cannot be captured by binary edges (even if one were to 
use directed graphs).
Thus, we lose the connection between
the number of bundles intersected by a conformal cut
and the cost of the corresponding assignment.

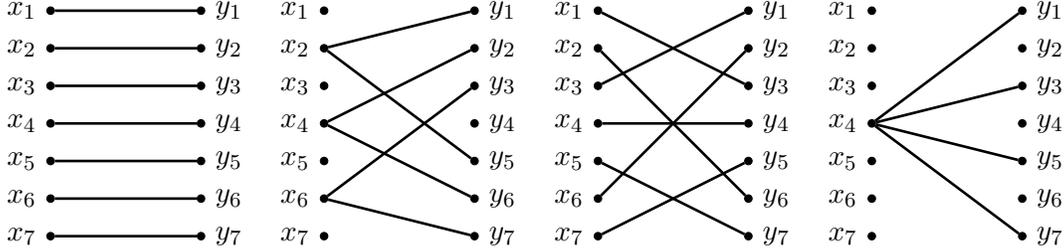
\begin{figure}
 \centering
 \input{figures/z8-1}
 \input{figures/z8-2}
 \input{figures/z8-3}
 \input{figures/z8-4}
 \caption{Graphs $B_e$ corresponding to equations
 $x = 1 \cdot y$, $x = 2 \cdot y$, $x = 3 \cdot y$ and
 $x = 4 \cdot y$.}
 \label{fig:z8-matchings}
\end{figure}

One idea for solving $\minlin{2}{R}$
for non-fields $R$ is to retain the 
``if and only if'' semantics of an edge in the associated graph
by matching \emph{sets} of values rather than individual values.
More specifically, let us partition $R \setminus \{0\}$ into
classes $C_1, \dots, C_\ell$ and build $G(S)$
with vertices $s$, $t$ and $x_{C_1}, \dots, x_{C_\ell}$ for all $x \in V(S)$.
One can also see this partition as an equivalence relation $\equiv$.
To keep the matching structure for an equation $e$,
we want an edge $x_{C_i} y_{C_j}$ in $B_e$ 
to mean ``$x \in C_i$ if and only if $y \in C_j$''.
For fields, we used the most refined partition
(every nonzero element is a class of its own).
For other rings, 
a coarser partition is needed:
e.g. if $R = \ZZ_8$, then 
$2$ and $6$ have to be in the same class (think of
$x = 2y$).
However, simply taking a coarser partition is not sufficient:
indeed, the coarsest partition 
(putting all non-zero elements in the same class)
has the required structure,
but it only distinguishes between zero and nonzero values, 
and is not very useful algorithmically.
Thus, we need to strike a balance
by choosing a partition
coarse enough to guarantee
the matching structure and
refined enough to preserve
useful information about the equations.

The quest for the right definition of 
a useful partition leads to one of our main algorithmic ideas.
First, we generalize $\minlin{2}{R}$ into $\minlinideal{2}{R}{I}$ where
$I \subseteq R$ is an ideal and we only allow variables to take values from $I$.
We illustrate by $\ZZ_8$.
Consider the chain of ideals
$(1) \supsetneq (2) \supsetneq (4) \supsetneq (0)$ 
in $\ZZ_8$ where $(1) = \ZZ_8$ and $(0) = \{0\}$.
We can partition $(1) \setminus \{0\}$, 
$(2) \setminus \{0\}$ and $(4) \setminus \{0\}$
in a way that allows us to perform a chain of reductions
from $\minlinideal{2}{R}{(1)}$ to
$\minlinideal{2}{R}{(2)}$ to
$\minlinideal{2}{R}{(4)}$ down to
$\minlinideal{2}{R}{(0)}$.
The catch is that in each reduction step
we increase the cost, but only by a constant factor.
Thus, if we start with a set of equations $S$ 
that admits an assignment $V(S) \to (1)$ of cost at most $k$,
then we end up with a set of equations $S'$ 
that admits an assignment $V(S') \to (0)$
of cost at most $c \cdot k$.
Observe that checking the latter condition
only requires polynomial time because $0$ is 
the only available value.
The steps never reduce the cost,
so if every assignment $V(S) \to (1)$
has cost greater than $c \cdot k$,
then every assignment $V(S') \to (0)$
has cost greater than $c \cdot k$.
Thus, our algorithm is a $c$-approximation.

We now take a closer look at the domain-shrinking step,
i.e. the reduction from $\minlin{2}{R}$ over ideal $I$ to
the problem over ideal $I' \subsetneq I$.
We require a partition of $I \setminus \{0\}$
to respect the matching property and 
an additional absorption property:
if values $i \equiv j$, then $i - j \in I'$.
Such a partition for ideal $(1)$ of $\ZZ_8$
can have classes $\{1,3,5,7\}$, $\{2,6\}$ and $\{4\}$.
For intuition, it is instructive 
to view values in base 2,
i.e. $\{001_2, 011_2, 101_2, 111_2\}$,
$\{010_2, 110_2\}$ and $\{100_2\}$---two values are in the same class
if and only if they have the same number
of trailing zeros.\footnote{A partition can be defined
in a similar way for more general rings
$\ZZ_{p^q}$ where $p > 2$, but with an extra
condition that two elements in the same class 
also have the same least significant nonzero digit.}
Thus, the step from $(1)$ to $(2)$ can be thought of
as aligning the number of trailing zeros of the values
assigned to the variables.
Having an assignment $\tau$ of classes to the variables,
we can concentrate on
assignments that agree with $\tau$,
and use the absorption property to rewrite
equations over $(1)$ into equations over $(2)$.
For example, if we have an equation 
$x = 2y$ over $(1)$ and the class assignment is
$\tau(x) = \{2,6\}$ and $\tau(y) = \{1,3,5,7\}$,
then we rewrite it as $x'+2 = 2(y'+3)$ over $(2)$,
where the $2$ on the left hand side and 
the $3$ on the right hand side are (arbitrarily chosen) 
representatives from the classes 
$\{2,6\}$ and $\{1,3,5,7\}$, respectively.
Note that this represents a reversible transformation of the solution space
(by mapping $x'=x-2$ and $y'=y-3)$, so we do not gain or lose any solutions.
Moreover, for any solution to $x=2y$ that respects $\tau$
we have $x \equiv 2$ and $y \equiv 3$, hence
$x', y' \in (2)$ by the absorption property.
Thus, $x=2y$ is translated to a new (non-homogeneous) equation
in $x'$ and $y'$ over $(2)$.

Based on this,
we call a finite commutative ring $R$
\emph{Bergen} if it contains a chain of ideals 
$R = I_0 \supset I_1 \supset \dots \supset I_\ell = \{0\}$
such that for every $0 \leq i < \ell$,
$I_i \setminus \{0\}$ admits a partition
with the matching property and 
absorption property into $I_{i+1}$.
Unfortunately, choosing the chain for 
a given ring $R$ is nontrivial.
For chain rings, the chain of ideals works (Lemma~\ref{lem:chain-is-bergen})
but there are Bergen rings that are not chain rings (Lemma~\ref{lem:simplebergenexample}).
Bergen rings are lineal 
(Proposition~\ref{prop:bergenislineal}),
which might suggest using the chain of
annihilators, but examples
in~\ref{ssec:knt-347} show that 
this does not always work.

It remains to show how to use a partition of $I \setminus \{0\}$
that has the matching property and
the absorption property into $I' \subsetneq I$.
Recall that our goal is to reduce from $\minlinideal{2}{R}{I}$
to $\minlinideal{2}{R}{I'}$.
Let $\EQ^{\neq 0}$ be the set of classes in the partition.
Let $(S, k)$ be a simple instance and $G = G(S)$
be the corresponding graph.
Now we have a one-to-one correspondence
between class assignments $V(S) \to \EQ^{\neq 0}$
and conformal cuts in $G$.
And again, if $(S, k)$ is a yes-instance,
then there is a conformal cut $U$ such that $b(U) \leq k$.
However, now we cannot use a simple branching algorithm
in the style of \textsc{Digraph Pair Cut} since some conformal cuts of low cost
correspond to class assignments of high cost.
For an extreme example over $\ZZ_8$,
consider a system of equations
\begin{equation}
  \label{eq:z8-example}
  \tag{$\triangle$}
  S = \{x = 4, 2a = x, 3a = b, 3b = c, 3c = a\}.
\end{equation}
An equation of the form $x = 2y$ over $\ZZ_8$
implies that a matching partition of $(1) \setminus \{0\}$ has to put
$2$ and $6$ in one class $A$ and match it
with the class $B$ containing $4$. 
By construction of $G(S)$,
the vertex $t$ is isolated and 
the connected component $U$ of $s$ contains
vertices $s$, $x_B$, $a_A$, $b_A$ and $c_A$.
Hence, $\delta(U)$ is empty and conformal,
but the cost of $S$ is at least one because
the system is inconsistent (follows from
considering both possible values for $a$, which are $2$ and $6$).
In fact, the cost is exactly one because
it is sufficient to delete $2a = x$.
By making vertex-disjoint copies of $S$,
we can obtain instances of arbitrarily high cost
while the corresponding graphs
admit conformal cuts of zero size.

We also mention that the problem is not amenable
to LP-branching based on biased graphs~\cite{wahlstrom2017lp}
or the more general important balanced subgraphs
which was the main tool used in~\cite{Dabrowski:etal:soda2023}.
These methods can deal with cycle obstructions
in graphs that have the so-called {\em theta property}:
if there is a chord in a cyclic obstruction,
then at least one of the newly created cycles
must also be an obstruction.
However, our obstructions do not have this property.
Consider an instance 
\[ T = \{a = 1, a = 3b, b = 3c, c = 3d, d = 3e, e = 3a, 2a = x, x = 2c\} \]
of $\minlinideal{2}{\ZZ_8}{(1)}$. 
Note that all units in $\ZZ_8$ are in the same class, and call this class $A$.
Let $B$ be the class containing $2$.
Then $G(T)$ contains a cycle on vertices $a_A$, $b_A$, $c_A$, $d_A$, $e_A$
which is an obstruction since it corresponds to 
an unsatisfiable set of equations
$\{a = 1, a = 3b, b = 3c, c = 3d, d = 3e \}$.
Moreover, there is a chord connecting $a_A$ and $c_A$ through $x_B$,
but the subinstances
$\{a = 1, a = 3b, b = 3c, 2a = x, x = 2c\}$ and
$\{a = 1, c = 3d, d = 3e, e = 3a, 2a = x, x = 2c\}$
corresponding to the newly created cycles are satisfiable.
For similar reasons, the powerful method of flow augmentation 
does not appear to directly solve the problem. 

Instead, we use a sophisticated greedy method based on shadow removal.
Shadow removal, introduced by Marx and Razgon~\cite{marx2014fixed}
and refined in~\cite{chitnis2015directed}, provides a way
of exploring small transversals to (often implicit)
families of connected subgraphs.
Let $G$ be a graph and $\cF$ be a family of
connected subgraphs in $G$ that contain vertex $s$.
We are interested in transversals of $\cF$ of size at most $k$,
i.e. sets of at most $k$ edges that intersect every subgraph in $\cF$. 
For instance, if we are interested in $st$-cuts, 
then $\cF$ could be the set of all $st$-walks.
In our case, the conformal cut $\delta_{\sf opt} = \delta(U_{\sf opt})$ 
corresponding to an optimal class assignment is 
a transversal for all $st$-walks and all walks that
contains $s$, $x_A$ and $x_B$ for the same vertex $x$
and different classes $A$ and $B$.
More interestingly, $\delta_{\sf opt}$ is also a transversal
for subgraphs induced by $U \subseteq V(G)$ 
such that $\delta(U)$ is a conformal cut,
but the class assignment corresponding to $U$
cannot be extended into a satisfying assignment.
The system of equations in~\eqref{eq:z8-example}
and the connected component of $s$ described below it
provide an example of such a conformal cut.
Intuitively, the optimal conformal cut $\delta_{\sf opt}$ 
takes into account unsatisfiability occurring ``down the line'',
i.e. the obstructions that are not visible over the current ideal,
but will become apparent over the next ideal.
Formally, shadow removal with additional branching steps 
allows us to compute a transversal to the same 
family of subgraphs as those intersected by $\delta_{\sf opt}$,
which are precisely all $sv$-walks with $v \notin U_{\sf opt}$.
However, since we dismantle the bundles of edges
corresponding to an equation and treat them individually,
the size of the transversal we compute is 
$|\delta_{\sf opt}| \leq b \cdot k$,
where $b$ is the maximum bundle size.
When translating back into equations, this means that 
we may have to delete up to $b \cdot k$ equations, 
while an optimal assignment can delete $k$.
This is the reason the cost increases in each step of our algorithm.

Finally, we note that, as with the example of equations over finite
fields earlier, if we wanted to present an algorithm purely for
integer rings $\ZZ_m$, then many steps could be simplified.
As noted, we can focus on the case where $m=p^q$ for $p$ prime
and $q \geq 2$.
We
use the equivalence relation $\equiv$ described above, where $a \equiv b$ if $a$ and $b$,
written in base $p$, have the same number of trailing zeroes and the
same least significant non-zero digit. Then, instead of progressing to
the problem \textsc{Min-2-Lin$(\ZZ_{p^q})$-over-$(p)$}, thanks to the
regularity of the ring structure, we can recurse directly into the
problem \textsc{Min-2-Lin$(\ZZ_{p^{q-1}})$}. This may generalize
slightly, perhaps to \textsc{Min-2-Lin} over chain rings, 
but for more complex lineal rings the strategy will fail.
Similarly, consider the ring $R=\FF[\rx_1,\ldots,\rx_n]/((\rx_1,\ldots,\rx_n)^{d+1})$
of bounded-degree multivariate polynomials over a finite field $\FF$.
Define an equivalence relation where $p \equiv q$ for $p, q \in R$
if $p$ and $q$ have the same coefficients over their minimum-degree monomials.
Then we can reuse the same partition $\equiv$ of $R\setminus \{0\}$ 
over all ideals $I_i$ in the algorithm, which again simplifies some steps.
Unfortunately, for sufficiently complex Bergen rings (specifically, when
$R$ is not \emph{fully convex}; see Section~\ref{sec:geometry}),
this strategy will again fall short. Thus,
Bergen rings reflect the structures arising in more complex lineal
rings. 


\paragraph{Lower bounds.} We complement the approximation algorithm
for Bergen rings with a number of hardness results.
We prove the following for finite, commutative, non-trivial
rings $R$:

\begin{enumerate}
\itemsep0em
\item
$\minlin{r}{R}$ is \W{1}-hard to FPT-approximate within any constant when $r \geq 3$,

\item
$\minlin{2}{R}$ is \W{1}-hard to FPT-approximate within any constant if
$R$ is not Helly, and

\item
$\minlin{2}{R}$ is \ETH-hard to FPT-approximate within $2-\eps$
$(\eps > 0)$ if $R$ is not lineal. 

\end{enumerate}

The result for $\minlin{r}{R}$, $r \geq 3$, is based on
the additive structure of $R$ only and is thus
applicable to $\minlin{r}{G}$ where $G$ is an non-trivial Abelian group.
On the other hand, results for $\minlin{2}{R}$ must exploit
the multiplicative structure since
$\minlin{2}{G}$ is in \FPT whenever
$G$ is a finite Abelian group (by
reduction to the {\sc Group Feedback Edge Set} problem~\cite{guillemot2011FPT}).

Consider $\minlin{3}{G}$ when $r \geq 3$
and $G$ is a finite non-trivial Abelian group. Hardness for
this problem clearly carries over to finite commutative rings and
longer
equations.
Hardness results for 
$\minlin{3}{G}$ include the following:
Theorem~6.1 in \cite{Dabrowski:etal:soda2023} show that
$\minlin{3}{G}$ is \W{1}-hard to solve exactly and
it is known that $\minlin{3}{G}$ is not constant-factor approximable in polynomial time if
 P $\neq$ \NP (by combining results in~\cite{haastad2001some}
and~\cite{dalmau2013robust}).
Our generalization
shows that $\minlin{3}{G}$ is not FPT-approximable within
any constant if \FPT $\neq$ \W{1}.
We reduce from a fundamental problem in coding
theory: the
maximum likelihood decoding problem over $\ZZ_p$ with $p$ prime (\GAPMLD). 
We are given a matrix $A \in \ZZ_p^{n \times m}$, a vector $y \in \ZZ_p^n$, and
a $k \in {\mathbb N}$, and the goal is to distinguish the case when there
is a nonzero vector $x \in \ZZ_p^m$
with Hamming weight at most $k$ such that $Ax = y$ from the
case when for all $x \in \ZZ_p^m$
with Hamming weight at most $\gamma \cdot k$, $Ax \neq y$. 
This problem is \W{1}-hard for every 
$\gamma \geq 1$ and prime $p$~\cite[Theorem~5.1]{bhattacharyya2021parameterized}.
Consider an instance $(A, y, k)$ of $\GAPMLD$. The linearity of the equations
in $Ax=y$ gives us a reduction to
$\GAP{\gamma}$-$\minlin{*}{G}$ where $G$ is the additive group over $\ZZ_p$:
make an instance with $x_1, \dots, x_n$ as variables,
add row equations as crisp equations,
and add soft equations $x_j = 0$ for all $j \in [n]$.
Then an assignment of cost at most $\gamma k$
can only choose vectors $x$ with at most $\gamma k$ nonzero entries.
To show that $\GAP{\gamma}$-$\minlin{3}{G}$ is hard to 
FPT-approximate, we algebraically rewrite
long row equations into group equations of length 3,
and transfer the result to all finite groups using the structure
theorem of finite Abelian groups (Theorem~\ref{thm:fundamental-abelian-group}).

\begin{sloppypar}
Next, we prove that
$\minlin{2}{R}$ is \W{1}-hard to FPT-approximate within any constant when
$R$ is not Helly by a reduction from
$\minlin{3}{\FF}$ to $\minlin{2}{R}$, where $\FF$
is a suitably chosen finite field.
We illustrate using the non-Helly ring
$R = \ZZ_2[\rx,\ry]/(\rx^2,\ry^2)$.
The hardness of $\minlin{2}{R}$
is based on
``simulating'' certain ternary
equations over $\FF$ using
binary equations over $R$.
In this case,
we may choose $\FF=\ZZ_2$.
An element $r \in R$ is a sum $r_{\sf unit} + r_x \rx + r_y \ry + r_{xy} \rx\ry$,
where $r_{\sf unit}, r_x, r_y, r_{xy} \in \ZZ_2$.
Consider an equation $a + b + c = 0$ over $\ZZ_2$
and the following three equations over $R$
with variables $a'$, $b'$, $c'$ and $v'$: 
(1) $\rx \cdot v' = \rx\ry \cdot b'$, 
(2) $\ry \cdot v' = \rx\ry \cdot a'$, and 
(3) $(\rx+\ry) \cdot v' = - \rx\ry \cdot c'$.
Summing up the first two equations,
we get $(\rx + \ry) \cdot v' = \rx\ry \cdot (a' + b')$.
Together with the third equation,
this implies $\rx\ry \cdot (a' + b' + c') = 0$.
On the one hand, if an assignment 
$\phi : \{a,b,c\} \to \{0,1\}$ satisfies $a + b + c = 0$,
then setting $a' \mapsto \phi(a)$,
$b' \mapsto \phi(b)$, $c' \mapsto \phi(c)$
and $v' \mapsto \phi(a) \rx + \phi(b) \ry$
satisfies the equations over $R$.
On the other hand,
if $\phi' : \{a',b',c'\} \to R$
satisfies $\rx\ry \cdot (a' + b' + c') = 0$,
then $\phi'(a')_{\sf unit} + \phi'(b)_{\sf unit} + \phi'(c)_{\sf unit} = 0$,
so setting $a \mapsto \phi'(a')_{\sf unit}$,
$b \mapsto \phi'(b')_{\sf unit}$ and
$c \mapsto \phi'(c')_{\sf unit}$
satisfies $a + b + c = 0$.
We lift this idea to arbitrary non-Helly rings $R$. First,
assume that $R$ is local; this can be done
without loss of generality by Proposition~\ref{prop:sumapproximation}(1).
We then exploit connections between $R$ and its 
{\em residue field}, i.e.
the quotient $\FF=R/M$ where $M$ is the unique maximal ideal of $R$.
We start the reduction from $\minlin{3}{\FF}$
and use a tangle in $R$ of a particularly simple kind
for converting 
equations over $\FF$ into (at most) binary equations over $R$.
In the example above, the target tangle is $\Ann(\rx)$, $\Ann(\ry)$ and
$\Ann(\rx+\ry)+\rx$.
Correctness follows from
$\FF$ being the residue field of $R$: we have a projection operator that enables us to turn a solution over $R$
into a solution over $\FF$.
\end{sloppypar}

We finally consider $\minlin{2}{R}$ when $R$ is
not lineal.
We prove that $\minlin{2}{R}$ is \ETH-hard to FPT-approximate 
within $2-\eps$, $\eps > 0$, by a 
reduction from 2CSP parameterized by the number of variables.
2CSP is a variant of CSP with arbitrary {\em binary} constraints.
The \emph{Parameterized Inapproximability Hypothesis (PIH)}~\cite{lokshtanov2020parameterized} postulates
that the gap version of 2CSP is not in FPT, i.e.
no FPT algorithm can distinguish a satisfiable instance
from an instance where every assignment
satisfies at most $\eps$-fraction of constraints
for \emph{any} $0 < \eps < 1$.
In a recent breakthrough~\cite{Guruswami:etal:stoc2024}, it is shown that
the ETH implies the PIH, and also that
the gap version of 2CSP is not in FPT
parameterized by the number of constraints~\cite{guruswami2024almost}.
We reduce 2CSP to \textsc{Paired Min Cut}: the input is
a graph $G$ with vertices $s$ and $t$
such that the $st$-maxflow in $G$ is $2k$,
and a set of pairwise disjoint edge pairs $\cP$;
the goal is to find an $st$-cut in $G$
that intersects as few pairs in $\cP$ as possible.
\textsc{Paired Min Cut} is
\W{1}-hard~\cite{marx2009constant} and
we strengthen the result to FPT-inapproximability.
To illustrate the idea, 
consider a 2CSP instance with domain $\{1,\dots,n\}$ and $k$
constraints.
Let $v$ be a variable in the 2CSP.
A simple choice gadget for $v$ is an $st$-path of length $n$.
We modify it as follows:
replace every edge $d$
with a set $\cP_d(v)$ of internally disjoint paths;
each path in $\cP_d(v)$ has length $n$ and corresponds to a constraint of the 2CSP that involves $v$.
A minimal $st$-cut has to choose $d \in \{1,\dots,n\}$
for every variable $v$ and intersect every path in $\cP_d(v)$ in exactly one edge out of $n$.
We exploit this to encode constraints $R(x,y)$ of the 2CSP:
for every $(d,d') \in R$, consider the path in $\cP_d(x)$ and the path in $\cP_{d'}(y)$
corresponding to the constraint $R(x, y)$, and
pair up the edge $d'$ on the first path with edge $d$ on the second path.
Intuitively, if the cut respects the constraint $R(x,y)$,
it can cut two edges at cost one, otherwise it has to pay two.
Moreover, the $st$-maxflow in the resulting graph is $2k$,
so every $st$-mincut contains $2k$ edges.
If there is an $st$-mincut that intersects at most $(2 - \eps)k$ pairs,
then it contains at least $\eps k$ pairs of edges
and respects at least $\eps k$ constraints of the 2CSP.
We conclude that, under \ETH,
\textsc{Paired MinCut} is not 
$(2-\eps)$-approximable in FPT time.
The hardness result for non-lineal rings
follows by a reduction from 
\textsc{Paired Min Cut} using constructions
similar to the \W{1}-hardness proofs in
\cite{dabrowski2023parameterized,Dabrowski:etal:soda2023}.
Note that for direct sums of two finite fields,
e.g. $\ZZ_6 = \ZZ_2 \oplus \ZZ_3$, we have a $2$-approximation
FPT algorithm by Proposition~\ref{prop:sumapproximation}.
Hence, our hardness results imply optimality under \ETH in this case.

\paragraph{Geometric approach.}
In order to get a better view of lineal and Helly rings,
we study them for the special case of \emph{monomial} rings.
Consider a polynomial ring $R=\FF[\rx_1,\ldots,\rx_n]$
over some field $\FF$ with indeterminates $\rx_1$, \ldots, $\rx_n$.
A \emph{monomial} of $R$ is a formal product $\rx_1^{a_1} \cdots \rx_n^{a_n}$ 
with exponents $a_i \in \naturals$.
A \emph{monomial ideal} in $R$ is an ideal generated by
monomials. Finally, a \emph{monomial ring} is a ring $R$
that can be represented as
$R=\FF[\rx_1,\ldots,\rx_n]/I$ where $I$ is a monomial ideal.
We always assume that the field $\FF$ is finite. 
Monomial rings have a significantly easier structure than
arbitrary rings. In particular, they can be captured geometrically.
Each monomial $m=\rx_1^{a_1} \cdots \rx_n^{a_n}$ corresponds to a
point $(a_1,\ldots,a_n) \in \naturals^n$. For a point $\alpha \in \naturals^n$,
let $\rx^\alpha=\rx_1^{\alpha_1} \cdots \rx_n^{\alpha_n}$,
let $Z$ contain all points $\alpha \in \naturals^n$
such that $\rx^\alpha \in I$ in $R$ and let $N=\naturals^n \setminus Z$.
Then the bipartition $(N,Z)$ of $\naturals^n$ represents the split of
nonzero and zero monomials in $R$, and $R$ can be
analyzed in geometric terms.

We use this viewpoint for connecting monomial rings to various properties from \emph{discrete convex analysis}. This
field studies functions $f \colon \integers^n \to \RR$
that behave like discrete analogues to convex functions
(see Murota's book~\cite{Murota:DCA}). Our connection goes most directly to
Kashimura, Numata, and Takemura~\cite{KashimuraNT13} who study discrete
analogues of the separation of convex shapes by a hyperplane.
One says that a shape $S \subseteq \integers^n$ is
\emph{hole-free} if every integer point in its convex hull is
contained in $S$. It is \emph{1-convex} if for any $p, q \in S$,
all integer points on the line from $p$ to $q$ are in $S$.
Note that these are a stronger and a weaker discrete analogue of
the notion of a convex continuous shape. Moreover, 
for $(N,Z)$ as above we say that $(N,Z)$ are \emph{separated by a hyperplane}
if there is an affine hyperplane $H$ disjoint from $N \cup Z$
that separates $N$ from $Z$, and adopting a term from~\cite{KashimuraNT13}
we say that $(N,Z)$ satisfies \emph{Condition P} if there are no
points $p_1, p_2 \in N$ and $q_1, q_2 \in Z$ such that $p_1+p_2=q_1+q_2$.
Let $R=\FF[\rx_1,\ldots,\rx_n]/I$ be a monomial
ring, let $N$ be the set of exponents of non-zero monomials in $R$
and $Z$ the set of exponents of zero monomials in $R$. 
We find the following (see also Figure~\ref{fig:convex-hulls-intro}).
\begin{enumerate}
\itemsep0em
\item $R$ is lineal if and only if $(N,Z)$ satisfies Condition~P
\item The case that $(N,Z)$ are separated by a hyperplane is
  a strictly stronger condition that we refer to as $R$ being
  \emph{fully convex}
\item If $Z$ is hole-free then $R$ is Helly, and
  if $R$ is Helly then $Z$ is 1-convex. 
\end{enumerate}
The language of discrete geometry seems well suited
to characterize the potentially tractable classes of monomial lineal and Helly rings. 
We also note that more naive interpretations of convexity fail to capture
these classes. On the one hand, even if both $N$ and $Z$ are hole-free,
this does not guarantee  that $R$ is lineal. On the other hand, requiring
that the convex hulls of $N$ and $Z$ are disjoint (which is equivalent
to the fully convex case) is too restrictive, and not met by all lineal rings.

\begin{figure}
\centering
\footnotesize

\begin{subfigure}{.45\textwidth}
  \centering
\begin{tikzpicture}[scale = 0.78]
    \fill[gray,opacity=0.5] (0,0) -- (5,0) -- (5,4) -- (0,4) -- cycle;

    \draw[thick,->] (0,0) -- (5,0);
    \draw[thick,->] (0,0) -- (0,4);
    
    \node at (-0.2, -0.2) {0};
    \node at (4,    -0.2) {4};
    \node at (1,    -0.2) {1};
    \node at (2,    -0.2) {2};
    \node at (3,    -0.2) {3};
    \node at (-0.2,    1) {1};
    \node at (-0.2,    2) {2};
    \node at (-0.2,    3) {3};

    \draw[gray, fill=white] (0,0) rectangle (1,1);
    \draw[gray, fill=white] (1,0) rectangle (2,1);
    \draw[gray, fill=white] (2,0) rectangle (3,1);
    
   \draw[gray, fill=white] (0,1) rectangle (1,2);
    
    \draw[gray, fill=white] (0,2) rectangle (1,3);

    \draw[thick] (0,0) -- (0,3) -- (1,3) -- (1,1) -- (3,1) -- (3,0) -- cycle;

    \draw[dotted] (0,3) -- (1,1) -- (3,0);
    
\end{tikzpicture}
  \caption{$R=\FF[\rx,\ry] / (\rx^3,\rx\ry,\ry^3)$. The shaded region is hole-free so $R$ is Helly. However, $R$ is not lineal since $(2,0)+(0,2)=(1,1)+(1,1)$ contradicts Condition~P.\\ \phantom{.}}
  \label{fig:sub1}
\end{subfigure} \hspace{1em}%
\begin{subfigure}{.45\textwidth}
  \centering
 \begin{tikzpicture}[scale = 0.78]
    \fill[gray,opacity=0.5] (0,0) -- (5,0) -- (5,4) -- (0,4) -- cycle;

    \draw[thick,->] (0,0) -- (5,0);
    \draw[thick,->] (0,0) -- (0,4);
    
    \node at (-0.2, -0.2) {0};
    \node at (4,    -0.2) {4};
    \node at (1,    -0.2) {1};
    \node at (2,    -0.2) {2};
    \node at (3,    -0.2) {3};
    \node at (-0.2,    1) {1};
    \node at (-0.2,    2) {2};
    \node at (-0.2,    3) {3};

    \draw[gray, fill=white] (0,0) rectangle (1,1);
    \draw[gray, fill=white] (1,0) rectangle (2,1);
    \draw[gray, fill=white] (2,0) rectangle (3,1);
    \draw[gray, fill=white] (3,0) rectangle (4,1);
    
    \draw[gray, fill=white] (0,1) rectangle (1,2);
    \draw[gray, fill=white] (1,1) rectangle (2,2);
    \draw[gray, fill=white] (2,1) rectangle (3,2);
    \draw[gray, fill=white] (3,1) rectangle (4,2);
    
    \draw[gray, fill=white] (0,2) rectangle (1,3);

    \draw[dotted] (0,3) -- (1,2) -- (1,3) -- cycle;
    \draw[dotted] (1,2) -- (4,2) -- (4,0) -- cycle;

    \draw[thick] (0,0) -- (0,3) -- (1,3) -- (1,2) -- (4,2) -- (4,0) -- cycle;

 \draw[thick,->] (6.14,3.44) -- (3,1);

\node[draw] at (6.3,3.7) {{\bf !}};
   
\end{tikzpicture}
  \caption{$R=\FF[\rx,\ry] / (\rx^4,\rx^2\ry,\ry^3)$. The shaded region is not hole-free due to the marked point. Furthermore, it is on the line from $(2,2)$ to $(4,0)$, hence $R$ is not Helly and \minlin{2}{R} is not FPT-approximable.}
  \label{fig:sub2}
\end{subfigure}
\caption{Illustration of the geometric approach for two monomial rings $R=\FF[\rx,\ry]/I$. The shaded region contains exponents $(a,b)$ such that $\rx^a\ry^b=0$ in $R$. The dotted line marks its convex hull.}
\label{fig:convex-hulls-intro}
\end{figure}
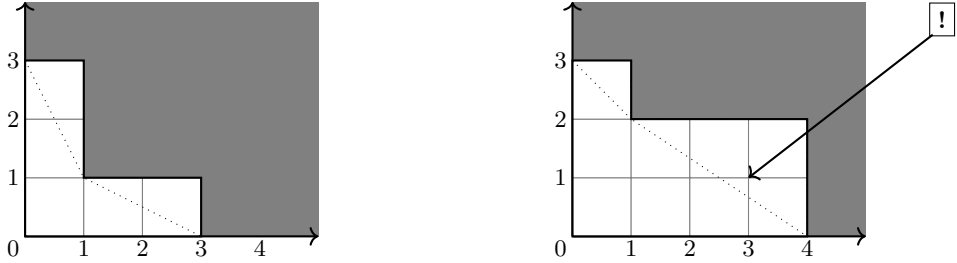

Let us exemplify. First, let $d \in \naturals$ and 
consider the ring $R=\FF[\rx_1,\ldots,\rx_n]/I$ over some finite field $\FF$ 
where $I$ contains all monomials of degree more than $d$. 
Then, $R$ is fully convex and Bergen, and our algorithm for approximating \minlin{2}{R}
would proceed in order of degree, starting by fixing the unit terms,
then the linear terms, and so on. A more general fully convex ring can
be defined by a notion akin to ``weighted degree'', i.e.\ a generic fully convex
ring can be defined as $R=\FF[\rx_1,\ldots,\rx_n]/I$ by weights $w_1, \ldots, w_n \in \QQ_{\geq 0}$ and a threshold $T$
where a monomial $\rx^\alpha$ is in $I$ if and only if $\sum_i w_i \alpha_i \geq T$.
However, using the geometric approach we can identify monomial rings that are lineal
but not fully convex.
Let $R_{\rm KNT}=\integers_2[\rx,\ry,\rz]/I$ where
$I$ contains all monomials with exponents above the half-plane spanned by
$(5,0,0)$, $(0,4,0)$ and $(0,0,3)$ and in addition $\rx^2\ry\rz \in I$. 
Then $R_{\rm KNT}$ is clearly not fully convex, but we show that it is Bergen.
In the other direction, let $R_{347}=\integers_2[\rx,\ry,\rz]/I$ where $I$
contains all monomials with exponents on or above the half-plane spanned by
$(3,0,0)$, $(0,4,0)$, and $(0,0,7)$ except $I$ does not contain $\rx\ry\rz^3$.
With the above notation, $R_{\rm KNT}$ has an $N$-set that is not hole-free,
and $R_{347}$ is not Z-hole-free, yet we show that it is Bergen.
In fact, their structure goes a long way towards justifying the definition of Bergen rings
and the structure of the algorithm.
With this in mind, we consider settling the complexity of \minlin{2}{R} for monomial rings $R$ 
to be an important step towards a full dichotomy, since
the geometric approach and monomial rings have been the basis
for our developments (such as the Bergen rings) and forms the best perspective
through which to understand the various classes of rings we consider.

\medskip

\noindent
{\bf Roadmap.}
The remainder of the paper is structured as follows.
We begin by introducing some necessary preliminaries in Section~\ref{sec:prelims}
and continue by taking a closer look at ideals and annihilators in Section~\ref{sec:ideals}.
Our algorithmic result can be found in Section~\ref{sec:approximation-algorithm}, the hardness
results in Section~\ref{sec:lowerbounds}, and we introduce the geometric approach in
Section~\ref{sec:geometry}.
We conclude the paper with a brief discussion of our results and future
research directions in Section~\ref{sec:discussion}.
Some of the examples in this paper have been discovered using
straightforward computer-assisted search. We stress that the properties of these
examples can, in principle, be verified by hand.

\section{Preliminaries}
\label{sec:prelims}
We assume familiarity with the basics of graph theory, linear and abstract algebra, and combinatorial optimization throughout the article.
The necessary material can be found in, for instance, the textbooks by Diestel~\cite{Diestel:GT}, Artin~\cite{Artin:A}, and Schrijver~\cite{Schrijver:PE}, respectively.
We use the following graph-theoretic terminology in what follows.
Let $G$ be an undirected graph.
We write $V(G)$ and $E(G)$ to denote the vertices and edges of $G$, respectively.
If $U \subseteq V(G)$, then the {\em subgraph of $G$ induced
by $U$} is the graph $G'$ with
$V(G')=U$ and $E(G')=\{ \{v, w\} \mid v,w \in U \; {\rm and} \; \{v, w\} \in E(G)\}$. 
We denote this graph by $G[U]$.
If $Z$ is a subset of edges in $G$, we write $G-Z$ to denote the graph~$G'$ with $V(G') = V(G)$ and $E(G') = E(G) \setminus Z$.
For $X,Y \subseteq V(G)$, 
an \emph{$(X,Y)$-cut} is a subset of edges $Z$
such that $G - Z$ does not contain a path with 
one endpoint in $X$ and another in $Y$.
When $X,Y$ are singleton sets $X=\{x\}$ and
$Y=\{y\}$, we simplify the notation and write $xy$-cut instead
of $(X,Y)$-cut. We say that $Z$ is an $(X,Y)$-cut \emph{closest} to
$X$ if there is no $(X,Y)$-cut $Z'$ with $|Z'|\leq |Z|$ such that the
set of vertices reachable from $X$ in $G-Z'$ is a strict subset of the
set of vertices reachable from $X$ in $G-Z$.

\subsection{Linear Equations over a Ring}

A {\em ring} is an Abelian group (whose operation is called {\em addition}), 
with a second binary operation called {\em multiplication} that is associative, 
is distributive over the addition operation, and has an identity element.
We will exclusively consider {\em commutative rings}, 
where multiplication is a commutative operator. 
We assume henceforth that rings are non-trivial 
in the sense that they contain at least two distinct elements.
Let $R=(R,+,\cdot)$ denote such a ring.
Examples are the integers $\ZZ$, 
finite fields $\FF_{p^n}$, and rings $\ZZ_m$
with integer domain $\{0,\dots,m-1\}$ and addition and multiplication modulo $m$.
We let $0$ denote the additive identity element and 
$1$ the multiplicative identity element.
An element $d \in R$ is a {\em zero divisor} if $d \neq 0$ and
there exists an element $0 \neq d' \in R$ such that $d \cdot d'=0$.
We let $R^+$ denote the Abelian additive group $(R,+)$ and $R^\times$
the commutative multiplicative monoid $(R,\cdot)$.
A {\em unit} in $R$ is an invertible element for the multiplication of the ring, i.e. an element $d \in R$ is a unit if there exists $d' \in R$ such that
$d \cdot d'=1$. The units of $R$ under multiplication form a group within the monoid $R^\times$.

An expression $c_1 \cdot x_1+\dots+c_r \cdot x_r=c$ is a {\em (linear) equation over $R$}  
if $c_1,\dots,c_r,c \in R$ and $x_1,\dots,x_r$ are variables with domain $R$.
This equation is {\em homogeneous} if $c=0$.
Let $S$ denote a set (or equivalently a system) of equations over $R$.
We let $V(S)$ denote the variables appearing in $S$, and we
say that $S$ is {\em consistent} if there is an assignment
$\varphi : V(S) \rightarrow R$
that satisfies all equations in $S$. 
An instance of the computational problem $\lin{r}{R}$
is a system $S$ of equations in at most $r$ variables
over $R$, and the question is whether $S$ is consistent.
We allow some equations in an instance to be
soft (i.e. deletable at unit cost) and
crisp (i.e. undeletable).
We study the following computational problem.

\pbDefP{$\minlin{r}{R}$}
{A (multi)set $S$ of equations over $R$ with at most $r$
variables per equation,
a subset $S^\infty \subseteq S$ of crisp equations
and an integer $k$.}
{$k$.}
{Is there a set $Z \subseteq S \setminus S^\infty$ 
such that $S - Z$ is consistent and 
$|Z| \leq k$?}

We use crisp equations merely for convenience 
since they can be modelled by $k+1$ copies of 
the same soft equation.
For an assignment $\alpha : V(S) \to R$,
let $\cost_S(\alpha)$ be $\infty$ if $\alpha$
does not satisfy a crisp equation and
the number of unsatisfied soft equations otherwise.
We drop the subscript $S$ when it is clear from context.
We also write $\mincost(S)$ to denote the minimum cost
of an assignment to $S$.

\subsection{Parameterized Complexity}

In parameterized 
algorithmics~\cite{book/DowneyF99,book/FlumG06,book/Niedermeier06}
the runtime of an algorithm is studied with respect to
the input size~$n$ and a parameter $p \in \NN$.
The basic idea is to find a parameter that describes the structure of
the instance such that the combinatorial explosion can be confined to
this parameter.
In this respect, the most favourable complexity class is \FPT
(\emph{fixed-parameter tractable}),
which contains all problems that can be decided by an algorithm
running in $f(p)\cdot n^{O(1)}$ time, where $f$ is a computable
function.
Problems that can be solved within such a time bound are said to be 
\emph{fixed-parameter tractable} (FPT).

We will prove that certain problems are not in $\FPT$
and this requires some extra machinery.
A {\em parameterized problem} is, formally speaking, a subset of $\Sigma^* \times {\mathbb N}$
where $\Sigma$ is the input alphabet. Reductions between parameterized problems need to take
the parameter into account. To this end, we will use {\em parameterized reductions} (or FPT-reductions).
Let $L_1$ and $L_2$ denote parameterized problems with $L_1 \subseteq \Sigma_1^* \times {\mathbb N}$
and $L_2 \subseteq \Sigma_2^* \times {\mathbb N}$. 
A parameterized reduction from $L_1$ to $L_2$ is a
mapping $P: \Sigma_1^* \times {\mathbb N} \rightarrow \Sigma_2^* \times {\mathbb N}$
such that
(1) $(x, k) \in  L_1$ if and only if $P((x, k)) \in L_2$, (2) the mapping can be computed
by an FPT-algorithm with respect to the parameter $k$, and (3) there is a computable function $g : {\mathbb N} \rightarrow {\mathbb N}$ 
such that for all $(x,k) \in L_1$ if $(x', k') = P((x, k))$, then $k' \leq g(k)$.

The class \W{1} contains all problems that are FPT-reducible to \textsc{Independent Set} when parameterized
by the size of the solution, i.e. the number of vertices in the independent set.
Showing \W{1}-hardness (by an FPT-reduction) for a problem rules out the existence of a fixed-parameter
algorithm under the standard assumption $\FPT \neq \W{1}$.

\subsection{Approximation} \label{sec:approx}

We refer the reader to the textbook~\cite{ausiello2012complexity} 
for more information about polynomial-time approximability, and 
to the surveys~\cite{Marx:tcj2008} and~\cite{Feldmann:etal:algorithms2020}
for an introduction to parameterized approximability.
We will exclusively consider minimisation problems in what follows.
Formally speaking, a {\em minimisation problem} is a 3-tuple $(X, sol, cost)$ where

\begin{enumerate}
\item
$X$ is the set of instances. 

\item
For an instance $x \in X$, $sol(x)$ is the set of feasible solutions
for $x$, the length of each $y \in sol(x)$ is polynomially bounded in $\norm{x}$, and it is decidable in
polynomial time (in $\norm{x}$) whether $y \in sol(x)$ holds for given $x$ and $y$. 

\item
For an instance
$x \in X$ with feasible solution $y$, $cost(x, y)$ is a polynomial-time computable positive integer. 
\end{enumerate}

The objective of the minimisation problem $(X, sol, cost)$ is to find an solution 
$z$ of minimum cost for a given instance $x \in X$ 
i.e. a solution $z$ with $cost(x, z) = opt(x)$ where $opt(x) = \min \{cost(x, y) \mid y \in sol(x)\}$.
If $y$ is a solution for the instance $x$, then the performance ratio of $y$ is defined as 
$cost(x, y)/opt(x)$. For $c \geq 1$, we say that an algorithm is a factor-$c$ 
{\em approximation} algorithm if it always computes a solution with performance ratio at most $c$
i.e. the cost of solutions returned by the algorithm is at most $c$ times greater than the minimum cost.
The constant $c$ is sometimes replaced by a function of the instance $x$,
but this is not needed in this paper.

We continue with
FPT-approximability of minimization problems. 
For a minimization problem $(X, sol, cost)$ parameterized by the solution
size $k$, a factor-$c$ {\em FPT-approximation algorithm} 
is an algorithm that 
\begin{enumerate}
\item
takes a tuple $(x,k)$ as input where $x \in X$ is an instance and $k \in \NN$,

\item
either returns that there is no solution of size at most $k$ or returns a solution of size at
most $c \cdot k$, and

\item
runs in time $f(k) \cdot \norm{x}^{O(1)}$ where $f: \NN \rightarrow \NN$ 
is some computable function.
\end{enumerate}

Thus, an FPT approximation algorithm is given more time to compute the solution (compared to a polynomial-time
approximation algorithm) and it may output a slightly oversized solution (unlike an exact FPT algorithm). There is a very direct connection between a finite direct sum of rings and 
its FPT-approximability.

\begin{proposition} \label{prop:sumapproximation}
Suppose that the ring $R$ is 
  isomorphic to a
direct sum $\bigoplus _{i=1}^{n} R_{i}$.
\begin{enumerate}
\item
If $\minlin{2}{R_i}$ is not constant-factor approximable for some $i \in [n]$, then
$R$ is not constant-factor approximable.

\item
If $\minlin{2}{R_i}$ is factor-$c_i$ FPT-approximable for $i \in [n]$, then $\minlin{2}{R}$ admits an
  FPT-approximation algorithm with approximation factor $\sum_{i=1}^n c_i$.
\end{enumerate}
\end{proposition}
\begin{proof}
{\em Case 1.}
Assume without loss of generality that $i=1$.
There is an obvious approximability-preserving polynomial-time reduction from
$\minlin{2}{R_1}$ to $\minlin{2}{R}$ by replacing elements $r_1 \in R_1$ with
$(r_1,0,\dots,0) \in R$.

\medskip

\noindent
{\em Case 2.}
Since $R \cong \bigoplus _{i=1}^{n} R_{i}$, there exists a
ring isomorphism $r \mapsto (r_1,\dots,r_n)$, 
where $r_i \in R_i$ for all $i \in \{1,\ldots,n\}$.
If $e$ is the equation $ax+by=c$ over $R$, then we
let $e_i$ denote the equation $a_i x + b_i y = c_i$ over $R_i$.
By this isomorphism, 
the equation $e$ has a solution if and only if $e_i$ has a solution for
every $i \in \{1,\ldots,n\}$.

Let $I=(S,k)$ denote an arbitrary instance of
$\minlin{2}{R}$ with $S=\{e^1,\dots,e^m\}$. Let 
$S_i=(\{e^1_i,\dots,e^m_i\},k)$ and $I_i=(S_i,k)$, $i \in \{1,\ldots,n\}$.
Compute $c_i$-approximate solutions to $I_i$, $i \in \{1,\ldots,n\}$.
If at least one of these instances have no solution, 
then $I$ has no solution since every such instance
is a relaxation of $I$.
Otherwise, let $\alpha_i : V \to R_i$ be
solutions to $I_i$, $i \in \{1,\dots,n\}$,
and define $\alpha(v) \in R$ for all $v \in V$ 
to be the preimage of
$(\alpha_1(v_1), \dots, \alpha_n(v_n))$.
Observe that $\alpha$ violates equation $e$
if and only if there is $i \in \{1,\dots,n\}$
such that $\alpha_i$ violates $e_i$.
Hence, $\alpha$ violates at most $k \cdot \sum_{i=1}^n c_i$ equations,
and we are done.
\end{proof}

\emph{Gap problems} are commonly used to formulate approximation
problems as decision problems.
Let $\gamma \geq 1$ be a constant,
$r \in \NN$ be a positive integer
and $R$ be a ring.
We study the following problem.

\pbDefGap{$\GAP{\gamma}$-$\minlin{r}{R}$}
{An instance $(S, k)$ of $\minlin{r}{R}$.}
{$k$.}
{there exists $Z \subseteq S$, $|S| \leq k$ such that $S - Z$ is consistent,}
{for all $Z \subseteq S$, $|S| \leq \gamma k$, the instance $S - Z$ is inconsistent.}

Observe that $\GAP{1}$-$\minlin{r}{R}$ is the same problem as $\minlin{r}{R}$.
To see the connection to the previous definition of approximation,
suppose $\minlin{r}{R}$ admits a factor-$\gamma$ FPT-approximation algorithm.
We claim this implies that $\GAP{\gamma}$-$\minlin{r}{R}$ is in \FPT.
If $(S, k)$ is a YES-instance of the gap problem,
then the approximation algorithm 
applied to $(S, k)$
produces a solution of size $\leq \gamma k$.
If $(S, k)$ is a NO-instance,
then the approximation algorithm 
cannot find a solution of cost $\leq \gamma k$, and
correctly reports that
no solution of size $\leq k$ exists.
The contrapositive of this claim is useful for proving
of FPT-inapproximability within any constant factor
under the standard assumption \FPT $\neq$ \W{1}.

\begin{observation}
  If $\GAP{\gamma}$-$\minlin{r}{R}$ is \W{1}-hard,
  then $\minlin{r}{R}$ does not admit $c$-factor
  FPT-approximation algorithms for any constant $c$
  unless \FPT = \W{1}.
\end{observation}

\section{Ideals in Rings}
\label{sec:ideals}

Ideals are a central concept in the theory of rings and they are
often used for analysing the structure of rings, in particular with
respect to divisibility properties.
Another important aspect of ideals is that they enable us to create new rings using natural
constructions. 
Basic properties of ideals are discussed and explained in Section~\ref{sec:ideals-lattices}
together
with some ideal-based constructions.
In Section~\ref{sec:annihilators}, we take a closer look at a particular kind
of ideals known as {\em annihilators}. Annihilators are closely
connected to the solvability of linear equations over rings and
are thus highly relevant for our purposes.

\subsection{Basics of Ideals}
\label{sec:ideals-lattices}

\begin{table}
\scriptsize
  \centering
   \begin{tabular}{|c||c|cccc|ccc|} \hline
    $+$ & $0$ & $1$ & $1+\rx$ & $1+\ry$ & $1+\rx+\ry$ & $\rx$ & $\ry$ & $\rx+\ry$\\
    \hline \hline
    $0$ & $0$ & $1$ & $1+\rx$ & $1+\ry$ & $1+\rx+\ry$ & $\rx$ & $\ry$ & $\rx+\ry$\\\hline
    $1$ & $1$ & $0$ & $\rx$ & $\ry$ & $\rx+\ry$ & $1+\rx$ & $1+\ry$ & $1+\rx+\ry$\\
    $1+\rx$ & $1+\rx$ & $\rx$ & $0$ & $\rx+\ry$ & $\ry$ & $1$ & $1+\rx+\ry$ & $1+\ry$\\
    $1+\ry$ & $1+\ry$ & $\ry$ & $\rx+\ry$ & $0$ & $\rx$ & $1+\rx+\ry$ & $1$ & $1+\rx$\\
    $1+\rx+\ry$ & $1+\rx+\ry$ & $\rx+\ry$ & $\ry$ & $\rx$ & $0$ & $1+\ry$ & $1+\rx$ & $1$\\\hline
    $\rx$ & $\rx$ & $1+\rx$ & $1$ & $1+\rx+\ry$ & $1+\ry$ & $0$ & $\rx+\ry$ & $\ry$\\
    $\ry$ & $\ry$ & $1+\ry$ & $1+\rx+\ry$ & $1$ & $1+\rx$ & $\rx+\ry$ & $0$ & $\rx$\\
    $\rx+\ry$ & $\rx+\ry$ & $1+\rx+\ry$ & $1+\ry$ & $1+\rx$ & $1$ & $\ry$ & $\rx$ & $0$\\ \hline
  \end{tabular}

\bigskip
  
  \begin{tabular}{|c||c|cccc|ccc|} \hline
  $\cdot$  & $0$ & $1$ & $1+\rx$ & $1+\ry$ & $1+\rx+\ry$ & $\rx$ & $\ry$ & $\rx+\ry$\\
    \hline \hline
    $0$ & $0$ & $0$ & $0$ & $0$ & $0$ & $0$ & $0$ & $0$\\\hline
    $1$ & $0$ & $1$ & $1+\rx$ & $1+\ry$ & $1+\rx+\ry$ & $\rx$ & $\ry$ & $\rx+\ry$\\
    $1+\rx$ & $0$ & $1+\rx$ & $1$ & $1+\rx+\ry$ & $1+\ry$ & $\rx$ & $\ry$ & $\rx+\ry$\\
    $1+\ry$ & $0$ & $1+\ry$ & $1+\rx+\ry$ & $1$ & $1+\rx$ & $\rx$ & $\ry$ & $\rx+\ry$\\
    $1+\rx+\ry$ & $0$ & $1+\rx+\ry$ & $1+\ry$ & $1+\rx$ & $1$ & $\rx$ & $\ry$ & $\rx+\ry$\\\hline
    $\rx$ & $0$ & $\rx$ & $\rx$ & $\rx$ & $\rx$ & $0$ & $0$ & $0$\\
    $\ry$ & $0$ & $\ry$ & $\ry$ & $\ry$ & $\ry$ & $0$ & $0$ & $0$\\
    $\rx+\ry$ & $0$ & $\rx+\ry$ & $\rx+\ry$ & $\rx+\ry$ & $\rx+\ry$ & $0$ & $0$ & $0$\\ \hline
  \end{tabular}
  \caption{Addition and multiplication tables for $\ZZ_2[\rx,\ry]/(\rx^2,\rx\ry,\ry^2)$. The additive group is isomorphic to $\ZZ_{2} \times \ZZ_{2} \times \ZZ_{2}$.}
  \label{tb:mt-z2-li}
\end{table}

Let $(R, +, \cdot)$ be a commutative ring.
An {\em ideal} in $R$ is a subset $I \subseteq R$ such that
\begin{enumerate}[(1)]
  \item $(I,+)$ is a subgroup of $(R,+)$ and 
  \item for every $d \in R$ and every $x\in I$, 
  the product $dx$ is in $I$. 
\end{enumerate}
Given a set of elements $S = \{s_1,\dots,s_n\} \subseteq R$, 
we let both $(S)$ and $(s_1,\dots,s_n)$ denote the smallest ideal
that contains the elements in $S$---such an ideal always exists.
For an element $s \in R$, we use the notation $(s)$ and $sR$ interchangeably.
An ideal $I$ in $R$ is \emph{principal} if $I = (s)$ for some $s \in R$.
Principal ideals are directly connected to divisibility in rings.
Given elements $u,v \in R$, we say that $u$ {\em divides} $v$ if there
is an element $w \in R$ such that $uw=v$, and $u$ divides $v$
if and only if $(v) \subseteq (u)$.
Note that the ring~$R$ is itself an ideal and is called the \emph{unit ideal} since $R = 1R$; 
an ideal is \emph{proper} if it is not the unit ideal.
A proper ideal is \emph{maximal} if it is not contained in any other proper ideal, and
a ring is {\em local} if it has a unique maximal ideal.

A \emph{coset} of an ideal $I$
is the set $a + I = \{ a + i \mid i \in I \}$
defined for some element $a \in R$.
Two cosets $a + I$ and $b + I$ with $a,b \in R$ are either equal or disjoint, 
so the cosets of an
ideal $I$ decomposes $R$ into disjoint subsets of equal size.
In other words, they are equivalence classes
of the equivalence relation 
that holds for $(a,b) \in R^2$ if and only if $a - b \in I$.
These equivalence class are sometimes written as $a \mod I$ and called the 
{\em residue class} of $a$ modulo $I$.
The \emph{quotient ring} $R / I$ has
the cosets $a + I$ for all $a \in R$ as elements,
with addition and multiplication defined 
using the corresponding operations in $R$, namely
$(a + I) + (b + I) = (a +_R b) + I$ and
$(a+I) \cdot (b+I) = (a \cdot_R b)+I$.
The additive identity in this ring
is $(0+I)$ and the multiplicative identity is $(1+I)$.
It is easy to verify that $\ZZ_m$ is isomorphic to $\ZZ/(m\ZZ)$ for all $m \in \NN$.

We continue with rings based on polynomials.
Given $n$ indeterminates $\rx_1,\dots,\rx_n$,
a {\em monomial} is a formal product 
$\rx_1^{\alpha_1}\rx_2^{\alpha_2} \cdots \rx_n^{\alpha_n}$
with nonnegative exponents $\alpha_1,\dots,\alpha_n$, and it has
{\em degree} $\alpha_1+\dots+\alpha_n$.
A {\em polynomial} in these indeterminates over a ring $R$ 
is a finite linear combination of monomials
with coefficients taken from the ring $R$.
The set of polynomials is denoted $R[\rx_1,\dots,\rx_n]$ and it can be equipped 
(in the standard way) with addition and multiplication that 
turns it into a ring---we use the same notation for this ring.
If $R$ is commutative, then the ring $R[\rx_1,\dots,\rx_n]$ is also commutative.
The rings
$\ZZ_2[\rx,\ry]/(\rx^2,\rx\ry,\ry^2)$ and
$\ZZ_2[\rx,\ry]/(\rx^2,\ry^2)$
will appear as recurring examples in what follows.
The addition and multiplication tables for 
the former can be found in Table~\ref{tb:mt-z2-li}.
In the second ring, $\rx \ry \neq 0$, and it has twice as many elements as the first one:
elements of the first ring can be written as $a_0 + a_x \rx + a_y \ry$ for $a_0, a_x, a_y \in \ZZ_2$,
while elements of the second one can be written as 
$a_0 + a_x \rx + a_y \ry + a_{xy} \rx\ry$ for $a_0, a_x, a_y, a_{xy} \in \ZZ_2$.

Every finite commutative local ring is isomorphic
to a polynomial ring (with coefficients taken from a finite field)
factored by an ideal. Since we are exclusively interested in finite
commutative rings, we verify this by using a result by Ganske and McDonald (more
general results can be obtained by using
Cohen's structure theorem~\cite{Cohen:tams46}).
Recall that the characteristic of a ring $R$ is the smallest positive integer such
that $1 + 1 + \dots + 1$ ($n$ times) equals $0$, and that
the characteristic of a  finite local commutative ring is a prime 
power~\cite[p. 522]{Ganske:McDonald:rmjm73}.
A {\em ring homomorphism} is a structure-preserving function between two rings: if $R$ and $S$ are rings, then a ring homomorphism is a function $\phi : R \rightarrow S$ 
 such that (1)
$\phi(a+b)=\phi(a)+\phi(b)$ for all $a,b \in R$,
(2)
$\phi(ab)=\phi(a)\phi(b)$ for all $a,b \in R$, and (3)
$f(1_R)=1_S$. 
We note that the characteristic of $R$ can equivalently be defined
via the unique ring homomorphism $\pi : \ZZ \to R$:
the characteristic of $R$ 
is the positive integer $n \in \NN$ such that
$\pi^{-1}(0) = n \ZZ$.

\begin{theorem}[Theorem 6.1 in \cite{Ganske:McDonald:rmjm73}]
\label{thm:localstructure}
Let $R$ be a finite commutative local ring with
characteristic $p^r$ where $p$ is a prime and $r$ is a positive integer. Let the maximal ideal $M$ of $R$ have minimal generating set $\{u_1,\dots,u_n\}$.
Then, there is a ring homomorphism from ${\mathbb Z}_{p^r}[x_1,\dots,x_{n+1}]$
to $R$.
\end{theorem}

Let $\phi$ be a homomorphism from a ring $R$ to a ring $S$.
The {\em kernel} of $\phi$, defined as ${\rm ker}(\phi) = \{a \in R \; | \; \phi(a) = 0\}$, is an ideal in $R$, and every ideal in $R$ arises from some ring homomorphism in this way.

\begin{theorem}
\label{thm:ring-fundamental}
 (Fundamental Theorem of Ring Homomorphisms). Suppose that $R$ and $S$ are rings,
and that $\phi : R \rightarrow S$ is a ring homomorphism. Then $R/ {\rm ker}(\phi)$ is isomorphic to $\phi(R)$.
\end{theorem}

We arrive at the following observation by combining the previous two theorems.

\begin{corollary} \label{cor:localringstandard}
Every finite local commutative ring is isomorphic to
${\mathbb Z}_{p^r}[x_1,\dots,x_n]/I$ where $p$ is prime, $n$ and $r$ are positive integers, and $I$ is an ideal.
\end{corollary}

\subsection{Annihilators}
\label{sec:annihilators}

Let $R$ be a commutative ring.
For $a \in R$, let $\Ann_R(a)=\{x \in R \mid ax=0\}$ 
be the {\em annihilator} of $a$.
Note that annihilators are ideals of $R$,
and we say that an ideal $I$ is a 
\emph{one-element annihilator} if there exists $a \in R$
such that $I = \Ann(a)$.
We drop the subscript in $\Ann_R$ when the ring $R$ is clear from the context.
Finite commutative rings without zero divisors
are fields by 
Wedderburn's Little Theorem (see, for instance, \cite{Herstein:amm61}) 
and they have
exactly two annihilators: 
$\Ann(0)=R$ and $\Ann(a)=\{0\}$ for all $0 \neq a \in R$.
Rings with zero divisors have more annihilators: for example,
in the ring $\ZZ_2[\rx,\ry]/(\rx^2,\rx\ry,\ry^2)$ 
we have $\Ann(\rx)=\{0,\rx,\ry,\rx+\ry\}$.

Annihilators characterize the set of solutions
to two-variable linear equations over $R$.
Consider a binary linear equation $ax + by = c$,
where $a,b,c \in R$ and $x,y$ are variables.
Suppose $(r_1, r_2) \in R^2$ satisfies the equation,
i.e. $ar_1 + br_2 = c$.
Pick any $a' \in \Ann(a)$ and note that
$(r_1 + a', r_2)$ is still a solution
since $a(r_1 + a') + br_2 = ar_1 + aa' + br_2 = ar_1 + br_2 = c$ in $R$.
Furthermore, any pair $(q_1, q_2) \in R^2$
such that $q_1 \in r_1 + \Ann(a)$ and $q_2 \in r_2 + \Ann(b)$
satisfies the equation.
Given these observations, it seems likely that the annihilators of $R$
play a key role when studying $\minlin{2}{R}$
and this is indeed true, as we will see in the rest of this paper. We
now describe two classes of rings that are
defined via their annihilators: {\em lineal} and {\em Helly} rings.
A commutative ring $R$ is called {\em lineal} if its
annihilators are totally ordered under set inclusion~\cite{Marks:Mazurek:ijm2016}. 
Examples of lineal rings include $\ZZ_{p^m}$
when $p$ is prime and $\ZZ_2[\rx,\ry]/(\rx^2,\rx\ry,\ry^2)$.
A non-example is the ring $\ZZ_2[\rx,\ry]/(\rx^2,\ry^2)$ since
annihilators $\Ann(\rx)$ and $\Ann(\ry)$ are incomparable:
we have $\rx \in \Ann(\rx)$, $\ry \not\in \Ann(\rx)$, 
$\rx \not\in \Ann(\ry)$, and $\ry \in \Ann(\ry)$.

It is sometimes convenient to use an alternative
definition of lineality: $R$
is lineal if and only if 
$R$ \emph{does not} contain elements $a,b,c,d$ such that $ab=0,ad \neq 0, bc \neq 0, cd=0$. 
We refer to this as the {\em magic square property}.
The validity of the reformulated definition
follows from \cite[Proposition 2.2]{Marks:Mazurek:ijm2016}.
The magic square property implies that
the lineality of a ring is fully determined by its one-element annihilators.

\begin{proposition}[also follows from Theorem~2.1 in \cite{Marks:Mazurek:ijm2016}]
\label{prop:one-elem-lineal}
Let $R$ be a ring. The one-element annihilators of $R$ are totally ordered if and only if the annihilators of $R$ are totally ordered.
\end{proposition}
\begin{proof} 
The backward direction is obvious. For the other direction, we assume that the one-element annihilators of $R$ are totally ordered.
Arbitrarily choose elements $a,b,c,d \in R$
such that $ab=cd=0$. Assume $ad \neq 0$ and
$bc \neq 0$.
We see that $b \in \Ann(a)$ and $d \not\in \Ann(a)$.
Similarly, $b \not\in \Ann(c)$ and $d \in \Ann(c)$. It follows that
$\Ann(a)$ and $\Ann(c)$ are incomparable.
Hence, at least one of $ad$ and $bc$ equals 0
so the magic square property implies that the annihilators of $R$ are totally
ordered.
\end{proof}

The magic square property also implies that lineal rings are local.

\begin{proposition} \label{prop:lineal-is-local}
Every finite commutative lineal ring is local.
\end{proposition}
\begin{proof}
Suppose to the contrary that $R$ is a finite commutative lineal ring that is not local. Then $R$ is
isomorphic to the direct sum of commutative local rings
$R_1,\dots,R_k$ by
Theorem~\ref{thm:sumoflocalrings}. Consider the elements
$a=(1,0,\dots,0)$ and $b=(0,1,0,\dots,0)$ in $R$.
We see that $ab=ba=(0,\dots,0)$. The
magic square property implies that at least
one of $a^2$ and $b^2$ equal $(0,\dots,0)$ and
this is impossible since $1^2=1$.
\end{proof}

We turn our attention to another class of rings defined by
their one-element annihilators: {\em Helly} rings.
A family of sets $\cF$ is \emph{Helly}
if for every triple $S_1, S_2, S_3 \in \cF$,
if every pair intersects, then
all three sets have a common element,
i.e.
$S_1 \cap S_2 \neq \emptyset$, $S_2 \cap S_3 \neq \emptyset$, and $S_1 \cap S_3 \neq \emptyset$ together implies that
$S_1 \cap S_2 \cap S_3 \neq \emptyset$.
A counterexample $(S_1, S_2, S_3) \in \cF^3$ to the Helly property is
called a \emph{tangle}; it is a triple of sets in $\cF$
in which every pair intersects, but all three have no common element.
We say that a ring $R$ is \emph{Helly}
if it does not admit a tangle over its one-element annihilator cosets.
If a ring $R$ admits a tangle, then it also admits a tangle with simple
structure.

\begin{lemma} 
  \label{lem:homogenised-tangle}
  Let $R$ be a finite commutative ring.
  If $R$ is non-Helly, then there exist
  elements $a,b,c,d \in R$ such that
  $\Ann(a)$, $\Ann(b)$ and $\Ann(c) + d$
  form a tangle.
\end{lemma}
\begin{proof}
Let $C_i = \Ann(e_i) + f_i$ for $i \in \{1,2,3\}$ form a tangle in $R$.
Suppose $g \in C_1 \cap C_2$.
Then $g = g_1 + f_1 = g_2 + f_2$ for 
some $g_1 \in \Ann(e_1)$ and $g_2 \in \Ann(e_2)$.
Consider cosets
$C'_i = C_i - g$ for $i \in \{1,2,3\}$.
Note that $C'_1 = \Ann(e_1) + f_1 - g_1 - f_1 = \Ann(e_1)$ and, 
analogously, $C'_2 = \Ann(e_2) + f_2 - g_2 - f_2 = \Ann(e_2)$.
Note that $C'_1, C'_2, C'_3$ adhere to the form required by the lemma.
It remains to show that $C'_1, C'_2, C'_3$ is indeed a tangle.
Observe that if $x \in C_i \cap C_j$,
then $x - g \in C'_i \cap C'_j$ by definition.
Moreover, if there is $y \in C'_1 \cap C'_2 \cap C'_3$,
then $y + g \in C_1 \cap C_2 \cap C_3$.
However, $C_1 \cap C_2 \cap C_3 = \emptyset$,
hence no such $y$ exists, and $C'_1, C'_2, C'_3$ is indeed a tangle.
Setting $a = e_1$, $b = e_2$, $c = e_3$ and $d = f_3 - g$ completes the proof.
\end{proof}

This kind of tangles are an important part of the non-approximability result in Section~\ref{sec:inapprox-non-helly}.
They are also useful for relating the set of lineal rings to the set
of Helly rings. 

\begin{proposition} \label{prop:lineal-is-helly}
Every finite commutative lineal ring is Helly.
\end{proposition}
\begin{proof}
Assume to the contrary that $R$ is lineal but not Helly.
Lemma~\ref{lem:homogenised-tangle} implies that there exists
elements $a,b,c,d \in R$ such that  $\Ann(a)$, $\Ann(b)$ and $\Ann(c) + d$
form a tangle.
Since $R$ is lineal, $\Ann(a)$ and $\Ann(b)$ are ordered under set inclusion.
Without loss of generality, assume $\Ann(a) \subseteq \Ann(b)$.
Then, $\emptyset = \Ann(a) \cap \Ann(b) \cap (\Ann(c)+d) = \Ann(a) \cap (\Ann(c)+d) \neq \emptyset$ 
which contradicts the Helly property.
\end{proof}

We exemplify the ring classes defined above with a non-Helly ring
$\ZZ_2[\rx,\ry] / (\rx^2,\ry^2)$ and a Helly non-lineal ring
$\ZZ_2[\rx,\ry] / (\rx^3,\rx\ry,\ry^3)$.
First of all, $\ZZ_2[\rx,\ry] / (\rx^2,\ry^2)$ is not Helly
since $\Ann(\rx)$, $\Ann(\ry)$ and $\Ann(\rx + \ry) + \rx$ form a tangle
with elements $0$, $\rx$ and $\ry + \rx\ry$
in pairwise intersections,
but no common element since
$\Ann(\rx + \ry) + \rx = \{\rx, \ry, \rx + \rx\ry, \ry + \rx\ry\}$
while
$\rx, \rx + \rx\ry \notin \Ann(\ry)$
and
$\ry, \ry + \rx\ry \notin \Ann(\rx)$.
Checking that $\ZZ_2[\rx,\ry] / (\rx^3,\rx\ry,\ry^3)$
is not lineal is easy, too: we see that $\rx^2 \neq 0$ and $\ry^2 \neq 0$ but $\rx\ry=0$ so non-lineality follows from the
magic square property. Proving that $\ZZ_2[\rx,\ry] / (\rx^3,\rx\ry,\ry^3)$ is Helly is more difficult.
It can, naturally, be done by tedious algebraic manipulations
or by a computer-assisted case analysis.
However, the geometric approach that we will introduce later
offers, in many cases, a superior method. 
We thus postpone the verification until the end of Section~\ref{sec:monomial-rings}.

We finally demonstrate that Helly rings are not necessarily local.
The next proposition shows that $R_1 \oplus R_2$ is Helly whenever $R_1$ and $R_2$
are Helly. If $R_1$ and $R_2$ are non-trivial rings, then $R_1 \oplus R_2$ is obviously
not a local ring.

 \begin{proposition} \label{prop:helly-direct-sum}
 Let $R$ be a finite commutative ring. If $R$ is isomorphic to a direct sum of Helly rings, then $R$ is Helly.
 \end{proposition}
 \begin{proof}
 Assume $R \cong \bigoplus_{i=1}^n R_i$ where $R_i$ are
 finite commutative Helly rings.
 We view an element $d \in R$ as a vector $(d_1,\dots,d_n)$
 and let $d[i]=d_i$, $1 \leq i \leq n$. We extend the notation to sets 
 $D \subseteq R$ by letting
 $D[i]=\{d[i] \; | \; d \in D\}$.
 Arbitrarily choose three one-element annihilator cosets
 $C_a=(a'_1,\dots,a'_n)+\Ann((a_1,\dots,a_n))$,
   $C_b=(b'_1,\dots,b'_n)+\Ann((b_1,\dots,b_n))$, and
    $C_c=(c'_1,\dots,c'_n)+\Ann((c_1,\dots,c_n))$. Assume
    $C_a \cap C_b \neq \emptyset$,  $C_a \cap C_c \neq \emptyset$,
 and  $C_b \cap C_c \neq \emptyset$.
    It follows that
    $C_a[i] \cap C_b[i] \neq \emptyset$,
     $C_a[i] \cap C_c[i] \neq \emptyset$, and 
     $C_b[i] \cap C_c[i] \neq \emptyset$ for
    $1 \leq i \leq n$. The ring $R_i$ is
 Helly so    $C_a[i] \cap C_b[i] \cap C_c[i] \neq \emptyset$.
 Define $d=(d_1,\dots,d_n)$ where $d_i$ is arbitrarily chosen in $C_a[i] \cap C_b[i] \cap C_c[i]$.
 The element $d$ is in $C_a \cap C_b \cap C_c$ so $R$ is Helly.
\end{proof}

\section{Approximation Algorithm for Bergen Rings}
\label{sec:approximation-algorithm}

This section is devoted to our fpt-approximation algorithm for Bergen
rings, i.e. we will show the following theorem.

\begin{restatable}{theorem}{thembergenalgo} \label{thm:bergen-algo}
  If $R$ is a Bergen ring,
  then $\GAP{\gamma}$-$\minlin{2}{R}$ 
  is randomized fixed-parameter tractable for 
  some constant $\gamma$ that only depends on $R$.
\end{restatable}

The section is organised as follows. We start with a technical
overview in Section~\ref{ssec:overview} that summarizes all the main ideas and main parts of
the algorithm. There we also provide the motivation for the notion of
Bergen rings for which we provide examples and context in
Section~\ref{sec:bergenrings}.
We then delve deeper into the various parts of the
algorithm, i.e. Section~\ref{sec:solving-equations-with-cosets}
provides the base case of the algorithm allowing us to solve
$\lin{2}{R}$ over cosets, Section~\ref{ssec:comp-hom-branch}
provides the preprocessing part of the algorithm, and
Section~\ref{ssec:class-assignments} provides the main working horse of the
algorithm and shows how to compute class
assignments using minimum cuts and shadow removal.

\subsection{Technical Overview}\label{ssec:overview}

We begin by presenting the algorithmic idea 
from a high-level perspective and then gradually 
introduce the necessary formal details. 
Consider $\linideal{2}{R}{I}$ -- a generalization of $\lin{2}{R}$
where every variable is restricted to values from 
a subset $I \subseteq R$.
The problem $\minlinideal{2}{R}{I}$
as well as the gap version
are defined analogously.
Note that the original problem is a special case with $I=R$.
We choose subsets $I$ to be ideals
so that the problem $\linideal{2}{R}{I}$ is 
solvable in polynomial time (see Section~\ref{sec:solving-equations-with-cosets}).

Intuitively, our algorithm is designed to 
take an instance of $\linideal{2}{R}{I}$ and 
produce an instance $\linideal{2}{R}{I'}$
with $|I'| < |I|$ while increasing the cost by 
at most a constant factor.
Thus, we make progress by reducing the number of 
values that each variable can take, and by
repeatedly applying this algorithm, 
we eventually arrive at a trivial problem.
A solution for the latter can be
transformed it into a solution for the original instance,
and since the number of steps is at most $|R|$ and
each step increases the cost at most linearly,
our solution is within a constant factor of the optimum.

An equation is \emph{simple}
if it is either a binary equation
of the form $u = rv$ for some $r \in R$
or a crisp unary equation
$u = r$ for some $r \in R$.
An instance $S$ of $\lin{2}{R}$ is simple
if every equation in $S$ is simple.
We use iterative compression,
homogenization and branching to assume that
the input instances $S$ are simple.
The details are deferred to 
Section~\ref{ssec:comp-hom-branch},
where we prove the following lemma.

\begin{theorem}
  \label{the:simple-instances}
  Let $R$ be a finite commutative ring, 
  $I$ be an ideal in $R$ and
  $\gamma \geq 1$ be a constant. 
  If $\GAP{\gamma}$ $\minlinideal{2}{R}{I}$
  restricted to simple instances is in \FPT,
  then $\GAP{\gamma}$ $\minlinideal{2}{R}{I}$ is in \FPT.
\end{theorem}

Now, let $\mincost(S,I)$ for an instance $S$ of 
$\lin{2}{R}$ and ideal $I$
be the minimum cost of an assignment $V(S) \to I$.
We abbreviate $\mincost(S, R)$ to $\mincost(S)$.
Consider two ideals $I' \subset I \subseteq R$.
Roughly, we want a procedure that transforms
an instance $S$ of $\lin{2}{R}$ into an instance $S'$ 
such that
$\mincost(S',I') \leq \gamma \cdot \mincost(S, I)$,
where $\gamma$ only depends on $I$.
For this, we require certain structure from $I$ and $I'$.
Let $\equiv$ be an equivalence relation on $I$ with classes $\EQ$. 
We sometimes abuse notation, for $r \in R$,
write $\EQ(r)$ to denote the equivalence class in $\EQ$ 
that contains $r$.
If the classes in $\EQ$ are cosets of ideals in $R$, 
we say that $\equiv$ is a \emph{coset partition of $I$}.
We say that $\equiv$ is \emph{matching} if
\begin{itemize}
  \item $\{0\}$ is an equivalence class in $\EQ$,
  \item for every $i,j \in I$ and $r \in R$,
  if $i \equiv j$ then $ri \equiv rj$, and
  \item for every $i,j \in I$ and $r \in R$,
  if $i \not\equiv j$, then either $ri \not\equiv rj$ or $ri=rj=0$.
\end{itemize}
Let $\EQ^{\neq 0} = \EQ \setminus \{\{0\}\}$.
The name ``matching'' comes from 
considering bipartite graphs $G^{\neq 0}_r$
defined by binary equations $u = rv$ 
for every $r \in R$ as follows:
let $V(G_r) = \EQ^{\neq 0} \uplus \EQ^{\neq 0}$ 
and let there be an edge between two classes 
$C_1$ on the left and $C_2$ on the right 
if and only if $i = rj$ for some 
$i \in C_1$ and $j \in C_2$.
Then $\equiv$ being matching implies
that every graph $G^{\neq 0}_r$ is a matching, 
i.e. every vertex has degree at most one.

The matching property of $\equiv$ is crucial
for the main algorithmic lemma, which we
state below and prove in Section~\ref{ssec:class-assignments}.
Let $\tau : V(S) \to \Gamma$ be a class assignment.
We say that an assignment 
$\alpha : V(S) \to I$ \emph{agrees with} $\tau$
if the value $\alpha(v)$ belongs to the class $\tau(v)$ 
for all $v \in V(S)$.

\begin{theorem}
  \label{the:get-class-assignment}
  Let $R$ be a finite commutative ring and
  $I$ be an ideal in $R$ that
  admits a matching coset partition $\equiv$
  with $d$ equivalence classes.
  There is a randomized algorithm that takes
  a simple instance $S$ of $\lin{2}{R}$,
  $\equiv$ and integer $k$,
  and in $\bigohs(2^{O(dk \log dk})$ time
  returns a class assignment $\tau : V(S) \to \EQ$
  such that the following holds:
  if $\mincost(S,I) \leq k$, then
  with probability at least
  $2^{-O(d^2k^2)}$,
  there exists an assignment $V(S) \to I$ 
  of cost at most $dk$
  that agrees with $\tau$.
\end{theorem}
We mention that the algorithm in the previous theorem
can be derandomized using standard methods~\cite{chitnis2015directed}.
While technical details are deferred to
Section~\ref{ssec:class-assignments},
we mention here that the 
equivalence class $\{0\}$ plays a special role:
since all binary equations in
$S$ are homogeneous,
they are satisfied by the all-zero assignment.
Intuitively, this observation
allows us to ``greedily''
assign value $0$ to the variables,
with the only obstacle being the crisp
unary equations in $S$.
After phrasing the class assignment problem
in terms of cuts in a certain auxiliary graph,
we apply the shadow removal technique.
We also mention that Theorem~\ref{the:get-class-assignment}
is the only place in the algorithm
where we sacrifice optimality
and resort to approximation.

Recall that our goal is to reduce
an instance of
$\linideal{2}{R}{I}$ into an instance of
\\$\linideal{2}{R}{I'}$.
To connect $I$ with $I'$, we require that 
\[
  i \equiv j \implies i - j \in I'
\]
for all $i, j \in I$.
To explain the utility of this property,
we need some definitions.
Choose an arbitrary representative element $\rep{C}$
from every equivalence class $C \in \EQ$.
Furthermore, for $i \in I$, 
let $\rep{i} = \rep{\EQ(i)}$ and
$\res{i} = i - \rep{i}$.
Now let $e$ be a simple equation and 
$\tau(V(e)) \to \EQ$ be a class assignment to its variables.
If $e = (u = r)$, then 
define equation $e' = \nxt(e, \tau, \equiv)$ as
\[ 
  v' = r - \rep{\tau(v)}.
\] 
If $e = (u = rv)$, then 
define equation $e' = \nxt(e, \tau, \equiv)$ as
\[
  u' = rv' + (r \rep{\tau(v)} - \rep{\tau(u)}).
\]

\begin{lemma}
  \label{lem:next-level}
  Let $R$ be a finite commutative ring and
  $I' \subseteq I$ be ideals in $R$.
  Suppose $\equiv$ is a matching coset partition of $I$
  such that $i \equiv j \implies i - j \in I'$
  for all $i,j \in I$.
  Let $e$ be a simple equation in $R$
  and $\tau : V(e) \to \EQ$
  be a class assignment.
  Then $e$ admits a satisfying assignment in $I$
  that agrees with $\tau$ 
  if and only if $e' = \nxt(e, \tau, \equiv)$ is satisfiable in $I'$.
\end{lemma}
\begin{proof}
  If $\alpha : V(e) \to I$ is a satisfying assignment
  that agrees with $\tau$, then $\alpha' : V(e') \to I'$
  defined as $\alpha'(v') = \res{\alpha(v)}$ for all $v \in V(e)$
  satisfies $e'$.
  On the other hand, if $\beta' : V(e') \to I'$
  satisfies $e'$, then
  $\beta : V(e) \to I$ defined as 
  $\beta(v) = \rep{\tau(v)} + \beta'(v)$
  for all $v \in V(e)$ agrees with $\tau$
  and satisfies $e$.
\end{proof}

Intuitively, Theorem~\ref{the:get-class-assignment}
provides us with a class assignment that agrees with
a value assignment of cost at most $O(k)$,
and Lemma~\ref{lem:next-level} allows us to
use this class assignment to reduce the set of values 
available to all variables, making the problem easier.
We combine all algorithmic and algebraic ingredients
presented so far in the following definition and theorem.

\begin{definition}[Bergen ring]
  A finite commutative ring $R$ is \emph{Bergen}
  if it admits a chain of ideals
  $R = I_0 \supset I_1 \supset \dots \supset I_\ell = \{0\}$
  such that, for all $0 \leq i < \ell$,
  ideal $I_i$ admits a matching coset partition $\equiv_i$
  and $a \equiv_i b \implies a - b \in I_{i+1}$
  for all $a,b \in I_i$.
\end{definition}

Let $\cI = \{(I_i)_{i=0}^{\ell}, (\equiv_i)_{i=0}^{\ell-1}\}$ be
a \emph{Bergen witness} for ring $R$,
i.e. a chain of ideals and corresponding equivalence relations
that satisfy the definition above. Note that a ring $R$ can only have
a finite number of Bergen witnesses.
Let $\gamma(\cI) = \prod_{i=0}^{\ell-1} d_i$,
where $d_i$ is the number of equivalence classes in $\equiv_i$,
and let $\gamma(R)$
be the minimum $\gamma(\cI)$
over all Bergen witnesses.

\thembergenalgo*
\begin{proof}
  Let $I_0 \supset I_1 \supset \dots \supset I_\ell$
  be the chain of ideals and
  $\equiv_0, \dots, \equiv_{\ell-1}$ be the
  the equivalence relations
  witnessing that $R$ is Bergen.
  Furthermore, let $\EQ_i$ be the equivalence
  classes of $\equiv_i$ and assume without loss of generality
  $\prod_{j=0}^{\ell-1} |\EQ_j| = \gamma(R)$.
  
  Define partial products
  $\gamma_i = \prod_{j=i}^{\ell-1} |\EQ_j|$
  for $0 \leq i < \ell$ and let $\gamma_\ell = 1$.
  We shall prove that $\GAP{\gamma_i}$ $\minlinideal{2}{R}{I_i}$
  for $0 \leq i \leq \ell$ is in FPT
  by reverse induction on $i$.
  First, note that $I_\ell = \{0\}$ and
  $\GAP{1}$ $\minlinideal{2}{R}{i}$ is in polynomial time:
  assign $0$ to all variables, check that no
  crisp equations are broken 
  (otherwise, the cost is infinite),  
  then count the number of broken 
  soft constraints and compare it to $k$.

  Now suppose $\GAP{\gamma_{i+1}}$ $\minlinideal{2}{R}{I_{i+1}}$
  is in FPT for some $0 \leq i < \ell$,
  and let $\cA$ be the algorithm solving it.
  Let $(S, k)$ be an instance of
  $\GAP{\gamma_i}$ $\minlinideal{2}{R}{I_{i}}$.
  By Theorem~\ref{the:simple-instances},
  we may assume without loss of generality
  that $S$ is simple.
  Run the algorithm from Theorem~\ref{the:get-class-assignment}
  to obtain a class assignment $\tau : V(S) \to \EQ_i$.
  Define the instance $S'$ of $\linideal{2}{R}{I_{i+1}}$
  as $S' = \{ \nxt(e, \tau, \equiv_i) \mid e \in S \}$,
  run $\cA$ on input $(S', |\EQ_i|k)$ and return its answer.
  This clearly requires FPT time.

  For correctness, first assume that 
  $(S, k)$ is a yes-instance,
  i.e. $\mincost(S, I_i) \leq k$.
  Then with probability $2^{-O(k^2)}$,
  there exists an assignment $\alpha : V(S) \to I_i$
  of cost $|\EQ_i|k$
  that agrees with the class assignment $\tau$.
  Then, by Lemma~\ref{lem:next-level},
  $\mincost(S', I_{i+1}) \leq |\EQ_i|k$,
  and $\cA$ accepts, and we accept as well.
  Now suppose $\cA$ accepts $(S', |\EQ_i|k)$,
  i.e. $\mincost(S', I_{i+1}) \leq 
  \gamma_{i+1} \cdot |\EQ_i| k \leq \gamma_i k$.
  Then, by Lemma~\ref{lem:next-level},
  $\mincost(S, I_i) \leq \gamma_i k$,
  hence it is a yes-instance,
  and our answer is correct.
\end{proof}

\subsection{Examples of Bergen Rings} \label{sec:bergenrings}

We will now present a few examples of Bergen rings.
Before we begin, we note that it is clearly a decidable problem to check whether a finite ring $R$ is Bergen via exhaustive enumeration of ideal chains and equivalence relations.
We first verify that
every finite field ${\mathbb F}$ is a Bergen ring. The field ${\mathbb F}$ only have two ideals
${\mathbb F}=I_0 \supseteq I_1 = \{0\}$ and we can choose $\equiv_0$ so that every $\{a\}$, $a \in {\mathbb F}$, is an equivalence class of $\equiv_0$.
Clearly, all of these equivalence classes are cosets since $\{0\}$ is
an ideal in ${\mathbb F}$.
Now consider ${\mathbb Z}_4$. Here, we have the chain of ideals $\{0,1,2,3\} \supset \{0,2\} \supset \{0\}$ and we let 
the equivalence classes of $\equiv_0$ be $\{0\}$, $\{2\}$, and $\{1,3\}$, and
the equivalence classes of $\equiv_1$ be $\{0\}$ and $\{2\}$.
Note that $\{1,3\}=\{0,2\}+1$ so it is a coset in ${\mathbb Z}_4$.
With these examples as inspiration, we can verify that every finite chain ring
is Bergen. We first study how the ideals in a chain ring are generated.

\begin{lemma}
  \label{lem:chain-basis}
  Let $R$ be a finite commutative chain ring
  with ideals
  $R = I_0 \supset I_1 \supset \dots \supset I_\ell \supset I_{\ell+1} = \emptyset$.
  Then there exist elements $g_0, g_1, \dots, g_\ell$ such that
  $I_j = (g_j, \dots, g_\ell)$ for all $0 \leq j \leq \ell+1$.
\end{lemma}
\begin{proof}
  We use reverse induction on $j$ starting with 
  $j = \ell+1$ and working towards $j = 0$.
  In the base case with $j = \ell+1$, the claim is trivially true.
  Now suppose $I_{j+1} = (g_{j+1}, \dots, g_\ell)$,
  pick any $r \in I_{j} \setminus I_{j+1}$.
  Consider the ideal $(r, g_j, \dots, g_\ell)$.
  Since $R$ is a chain ring, it equals an ideal $I_k$
  for some $k \leq j$.
  Moreover, $I_{j} \supseteq (r, g_j, \dots, g_\ell) \supsetneq I_{j+1}$, so $k = j$.
\end{proof}

\begin{lemma}
  \label{lem:chain-is-bergen}
  Every chain ring is Bergen.
\end{lemma}
\begin{proof}
  Let $R$ be a finite commutative chain ring
  with ideals
  $R = I_0 \supset I_1 \supset \dots \supset I_\ell \supset I_{\ell+1} = \emptyset$.
  Use Lemma~\ref{lem:chain-basis} to obtain
  elements $g_0, \dots, g_\ell$ such that
  $I_j = (g_j, \dots, g_\ell)$ for all $0 \leq j \leq \ell+1$.
  Observe that $g_j \notin (g_{j+1}, \dots, g_\ell)$ for all $j$,
  so every element $r \in R$ can be uniquely represented as
  $r = \sum_{j=0}^{\ell} r_{j} g_{j}$, where all $r_j$ are units.
  Let $\vec{r} = (r_0, \dots, r_\ell)$ be the vector of units
  corresponding to $r$, and for $r \neq 0$, let
  $\ord(r) = \min \{ i : r_i \neq 0 \}$
  be the index of the first nonzero coordinate in $\vec{r}$
  and 
  $\lsu(r) = r_{\ord(r)}$ be the least significant unit.
  Define $\ord(0) = \ell + 1$ and $\lsu(0) = 0$.
  Observe that for every $r \in R$, we have
  $r \in I_{\ord(r)} \setminus I_{\ord(r)+1}$.
  Moreover, for every $p, q \in R$,
  we have 
  \begin{equation}
    \label{eq:ord}
    \ord(p \cdot q) = \min(\ell+1, \ord(p) + \ord(q))
  \end{equation} 
  and
  \begin{equation}
    \label{eq:lsu}
    \lsu(p \cdot q) = \lsu(p) \cdot \lsu(q).
  \end{equation}

  To show that $R$ is Bergen,
  we use the chain $I_0 \supset \dots \supset I_{\ell}$
  and one equivalence relation $\equiv$ for all $i$ defined as
  $a \equiv b$ if and only if $\ord(a) = \ord(b)$ and $\lsu(a) = \lsu(b)$.
  First, observe that if $a \equiv b$ and $a, b \in I_i$, then
  $\ord(a - b) > \ord(a)$ since
  $\ord(a) = \ord(b)$ and $\lsu(a) = \lsu(b)$,
  thus $a - b \in I_{i+1}$.
  Moreover, the equivalence classes of $\equiv$ in ideal $I_i$ are
  $C_{p,q} = \{ r \in I_i : \ord(r) = p, \lsu(r) = q \}$,
  which are exactly cosets $q + I_{p+1}$, 
  so $\equiv$ is a coset partition.
  Now we claim that $\equiv$ has the matching property.
  Since $\ord(0) = \ell+1$ and $\ord(a) \leq \ell$ for all $a \neq 0$,
  we have that $\{0\}$ is an equivalence class.
  Moreover, if $a \equiv b$, then
  $ra \equiv rb$ for all $r \in R$ by 
  \eqref{eq:ord}~and~\eqref{eq:lsu}.
  Finally, if $a \not\equiv b$, then
  for all $r \in R$,
  \eqref{eq:ord}~and~\eqref{eq:lsu}
  imply that $\ord(ra) = \ord(rb) = \ell+1$ or
  $\ord(ra) \neq \ord(rb)$ or $\lsu(ra) \neq \lsu(rb)$.  
\end{proof}

\begin{sloppypar}
Certainly,
there are Bergen rings which are not chain rings.
For instance, $R = \ZZ_2[\rx_1,\rx_2]/(\rx_1^2,\rx_1\rx_2,\rx_2^2)$
is not a chain ring because the ideals
$(\rx_1)$ and $(\rx_2)$ are incomparable.
On the other hand, the chain of ideals
$R \supset (\rx_1,\rx_2) \supset (0)$
together with equivalence relation $\equiv_0$
with classes $\{0\}$, $\{\rx_1\}$, $\{\rx_2\}$, $\{\rx_1+\rx_2\}$ and
$\{1, 1+\rx_1, 1+\rx_2, 1+\rx_1+\rx_2\}$ and
equivalence relation $\equiv_1$ with classes $\{0\}$, $\{\rx_1\}$, $\{\rx_2\}$, 
and $\{\rx_1+\rx_2\}$
show that
$R$ is Bergen.
This ring is a member of a large class of non-chain Bergen rings. 
\end{sloppypar}

\begin{lemma} \label{lem:simplebergenexample}
Arbitrarily choose
$w_1, \ldots, w_n, T \in \QQ_{\geq 0}$ and a finite field $\FF$. Let
$R=\FF[\rx_1,\ldots,\rx_n]/I$ 
where the ideal $I$ is defined such that $\rx_1^{\alpha_1} \cdot \dots \cdot \rx_n^{\alpha_n}$ is in $I$ if and only if $\sum_{i=1}^n w_i \alpha_i \geq T$.
Then, $R$
is Bergen and $R$ is not a chain ring
when $n \geq 2$.
\end{lemma}
\begin{proof}
The ring $R$ is Bergen due to the forthcoming 
Lemmas~\ref{lm:characterize-fullyconvex} and \ref{lm:fully-convex-to-bergen}.
Clearly, 
$\rx_1 \in (\rx_1)$ and $\rx_2 \in (\rx_2)$.
Furthermore, there are no $r_1,r_2 \in R$ such that $r_1\rx_1=\rx_2$
and $r_2\rx_2=\rx_1$ so $\rx_1 \not\in (\rx_2)$ and
vice versa.
We conclude that $(\rx_1)$ and $(\rx_2)$ are incomparable ideals.
\end{proof}

We can (once again) verify that  $\ZZ_2[\rx_1,\rx_2]/(\rx_1^2,\rx_1\rx_2,\rx_2^2)$ is Bergen
by choosing weights $w_1=w_2=1$ and threshold $T=2$.
Additional
non-chain Bergen rings are identified in Section~\ref{sec:geometry}
together with a thorough discussion of the rings considered in the previous lemma (which form
a subclass of Bergen rings referred to as {\em fully convex} for geometric reasons).
While examples like these show that the ideals in a Bergen ring do not need to be
linearly ordered by set inclusion, the annihilators must have this property.

\begin{proposition} \label{prop:bergenislineal}
Every Bergen ring $R$ is lineal.
\end{proposition}
\begin{proof}
Assume $R$ admits a Bergen witness with ideal chain $R = I_{0} \supset \cdots \supset I_{\ell}$ and
equivalence relations $\equiv_i$, $0 \leq i \leq \ell-1$.
The ring $R$ is not lineal if and only if there exist elements $a,b,c,d \in R$
such that $ab=cd=0$, $ad \neq 0$, and $bc \neq 0$ due to the magic square property.
Assume to the contrary that these elements exist and consider the equivalence relation $\equiv_0$.
If we suppose that $a \not\equiv_0 a+c$, then either (1) $ad=(a+c)d=0$ or
(2) $ad \not\equiv_0 (a+c)d$. Both are false since (1) $ad \neq 0$ and (2)
$ad=(a+c)d$, so $a \equiv_0 a+c$.
Analogously, the fact $0 \neq cb=(a+c)b$ implies that $a+c \equiv_0 c$.
It follows that $a \equiv_0 c$ and $ab \equiv_0 cb$.
We know that $ab=0$ so $ab$ is in the equivalence class $\{0\}$ but $cb \neq 0$
so it is not in the equivalence class $\{0\}$.
We conclude that
$\equiv_0$ does not exist and $R$ is lineal.
\end{proof}

We note that every Bergen ring is Helly by Proposition~\ref{prop:lineal-is-helly}.
We will discuss structural properties of Bergen rings in more detail
when we have introduced the geometric approach (Section~\ref{sec:geometry}), and we will provide examples
of more intricate Bergen rings (such as $R_{\rm KNT}$ and $R_{347}$
that we encounter in Section~\ref{sec:mtl-not-bergen}).

\subsection{Solving $\lin{2}{R}$ over Cosets}
\label{sec:solving-equations-with-cosets}

We present a polynomial-time algorithm for $\lin{2}{R}$ over
cosets, which is needed in our algorithm.
Let $\Gamma$ be an arbitrary relational $\tau$-structure $(A;Q_1,Q_2,\dots)$. 
The {\em constraint satisfaction problem} over $\Gamma$ (CSP$(\Gamma)$) is defined as follows.

\pbDef{CSP($\Gamma$)}
{${\cal I} = (V,\cC)$ where $V$ is a set of variables and $\cC$ a set of constraints of the form $Q(x_1, \ldots, x_k)$ where $x_1, \ldots, x_k \in V$ and $Q \in \Gamma$}
{Does there exist a function $V \rightarrow A$ which satisfies all constraints, i.e. $(f(x_1), \ldots, f(x_k)) \in Q$ for every $S(x_1, \ldots, x_k) \in C$?}

One may equivalently view CSP$(\Gamma)$ as the problem of deciding whether a given $\tau$-structure $\Delta$ has a homomorphism to the $\tau$-structure $\Gamma$ or not.
The structure $\Gamma$ is, for obvious reasons, often referred to as the {\em constraint language}.
Let $R$ denote a finite ring and let $E_R$ be the relational structure with domain $R$ and unary relations  $\{x \in R^2 \; | \; ax=c\}$ together with binary relations $\{(x,y) \in R^2 \; | \; ax+by=c\}$ for arbitrary $a,b,c \in R$.
We see that $\lin{2}{R}$ is the same computational problem as
CSP$(E_R)$.

We utilize algebraic invariants of relations and
constraint languages that are known as {\em polymorphisms}~\cite{Barto:etal:polymorphisms}.
An operation 
$f:D^{m}\to D$ {\em preserves} a relation 
$Q\subseteq D^{k}$ if, for any choice of $m$ tuples 
$(t_{{11}},\dotsc ,t_{{1k}}),\dotsc ,(t_{{m1}},\dotsc ,t_{{mk}})$ from $Q$, it holds that the tuple obtained from these
$m$ tuples by applying $f$ coordinate-wise, i.e. 
\[(f(t_{{11}},\dotsc ,t_{{m1}}),\dotsc ,f(t_{{1k}},\dotsc ,t_{{mk}})),\] 
is in $Q$. If $f$ preserves $Q$, then $f$ is a {\em polymorphism} of $Q$.
A ternary operation $f:D^3 \rightarrow D$
is {\em affine} if $f(d_1,d_2,d_3)=d_1-d_2+d_3$ 
where $(D, +)$ is an Abelian
group.

\begin{theorem} \label{thm:affine-p}
  (Corollary 5.4.2. in~\cite{Jeavons:etal:jacm97}) 
  Assume that $\Gamma$ is a relational structure with finite signature and finite domain. If the relations in $\Gamma$ are preserved by an affine
  operation $f$, then $\csp{\Gamma}$ is solvable in polynomial time.
\end{theorem}

\begin{proposition} \label{prop:in-ideal}
  Let $R=(R,+,\cdot)$ be a finite ring and let $C_1,\dots,C_m$ denote cosets of $R$.
  There exists a polynomial-time algorithm for $\csp{E_R \cup \{C_1,\dots,C_m\}}$ and this algorithm outputs a solution if one exists.
\end{proposition}
\begin{proof}
  Consider the affine operation $f(a,b,c) = a - b + c$.
  Clearly, it preserves sets of solutions to any linear equation over $R$ so it preserves every relation in $E_R$.
Furthermore, it preserves any ideal $I$ (viewed as a unary relation) of $R$ since $(I,+)$ is a subgroup
of $(R,+)$.
Finally, it preserves any coset $C$ of $R$. Choose an ideal $I$ in $R$
and an element $r \in R$ and consider elements $i_1+r,i_2+r,i_3+r$ in the
coset $I+r$. We see that 
\[f(i_1+r,i_2+r,i_3+r)=i_1+r-i_2-r+i_3+r=i_1-i_2+i_3+r \in I+r.\]
Theorem~\ref{thm:affine-p} implies that $\csp{E_R \cup \{C_1,\dots,C_m\}}$ is
polynomial-time solvable.
Furthermore, every constraint $x=d$ with $d \in R$ can be expressed in $E_R$ so if an instance ${\cal I}$ of $\csp{E_R \cup \{C_1,\dots,C_m\}}$
is satisfiable, then we can compute a solution in polynomial time by using self-reducibility.
\end{proof}

\subsection{Iterative Compression and Homogenization via Branching}
\label{ssec:comp-hom-branch}

Here, we provide the preprocessing step of our algorithm that allows
us to reduce the general problem to a problem on simple instances,
i.e. we show Theorem~\ref{the:simple-instances}. To achieve this we
will use iterative compression together with branching. An important
part of the preprocessing is that all binary equations of a simple
instance are homogeneous and this is why we also refer to the process
as homogenization.
Let $R$ be a finite commutative ring and
consider the problem $\GAP{\gamma}$ $\minlinideal{2}{R}{I}$.
The standard technique of iterative compression
allows us to assume that the input of this problem,
together with an instance $(S, k)$,
comes with a set $X \subseteq S$
of at most $\gamma k + 1$ soft equations
such that $S - X$ is satisfiable.
In the following we show how 
we can use such a set to reduce the problem to a small number of simple instances.

\begin{lemma}
  \label{lem:branch-simple}
  Let $R$ be a finite commutative ring with an ideal $I$.
  Let $(S,k)$ be an instance of $\GAP{\gamma}$ $\minlinideal{2}{R}{I}$
  and let $X \subseteq S$ be a set of
  equations such that $S - X$ is satisfiable.
  There is an algorithm that takes
  $S$ and $X$ as input, and in time
  $|I|^{2\gamma|X|}(|S|)^{\bigoh(1)}$ produces
  a set $\cT$ of at most 
  $|I|^{2\gamma|X|}$ simple instances of $\GAP{\gamma}$ $\minlinideal{2}{R}{I}$
  , each with at most
  $|S| + 2|X|$ equations, such that
  $(S,k)$ is a yes-instance if and only if so is one of the instances
  in $\cT$. Moreover, if one of the instances $t \in \cT$ is a
  yes-instance, then a solution for $(S,k)$ can be computed from a
  solution of $t$ in polynomial-time.
\end{lemma}
\begin{proof}
  Let $(S, k)$ be an instance of $\GAP{\gamma}$ $\minlinideal{2}{R}{I}$.
  Without loss of generality, assume that all equations in $S$ are 
  binary: this can be achieved by replacing every unary equation of the form $e=(ax = b)$
  for variable $x$ with the two binary equations $ax - z_e = b$ and $ax
  + z_e = b$, where $z_e$ is a new variable. Note that this operation
  does not change the cost of an assignment since if $ax=b$ is not
  satisfied by an assignment $z_e$ can be used to satisfy one of the equations
  $ax - z_e = b$ and $ax + z_e = b$.
  
  Now we homogenize the equations in $S - X$.
  Using Proposition~\ref{prop:in-ideal},
  compute an assignment $\eta : V(S) \to I$ that satisfies $S - X$.
  Create a new set of equations $S_0$ with variables
  $v_0$  for all $v \in V(S)$.
  For every equation $e = (ax + by = c)$ in $S - X$,
  create a new variable $z_e$ in $S_0$ and add equations
  $e_1=(z_e = ax_0)$ and $e_2=(z_e = -by_0)$ to $S_0$.

  Next, we branch on all possible assignment to $V(X)$.
  For every mapping $\beta : V(X) \to I$,
  define a simple instance $S_0 \cup S_\beta$ of $\lin{2}{R}$,
  where $S_\beta$ contains the unary equation
  $x_0 = \beta(x) - \eta(x)$ for every $x \in V(X)$.
  Then, $\cT$ is the set of all simple instances $(S_0 \cup S_\beta,
  k_\beta)$ of $\GAP{\gamma}$ $\minlinideal{2}{R}{I}$ for every $\beta : V(X) \to I$,
  where $k_\beta=k-\cost_X(\beta)$.
  
  For the runtime of the algorithm first note that computing $\eta$
  using Proposition~\ref{prop:in-ideal} takes
  polynomial time. Moreover, for every of the at most
  $|I|^{|V(X)|} \leq |I|^{2\gamma k + 2}$
  functions $\beta : V(X) \to I$, we can produce
  the instance $(S_0 \cup S_\beta, k_\beta)$ in polynomial-time.

  We continue by proving correctness.
  First assume $(S, k)$ is a yes-instance and let $\phi : V(S) \to I$
  be an assignment with $\cost(\phi) \leq k$.
  We claim that $(S_0 \cup S_{\beta}, k_{\beta})$ is a yes-instance of $\GAP{\gamma}$ $\minlinideal{2}{R}{I}$, where
  $\beta$ is the restriction of $\phi$ to $V(X)$.
  To show this, let  $\phi_0 : V(S_0) \to I$ be the assignment defined
  as follows:
  \begin{itemize}
    \item set $\phi_0(x_0) = \phi(x) - \eta(x)$ for every $x \in V(S)$, and
    \item set $\phi_0(z_e) = a\phi_0(x_0)$ for every equation $e = (ax + by = c)$ in $S - X$.
  \end{itemize}
  Since $\phi$ and $\beta$ agree on $V(X)$,
  $\phi_0$ satisfies all unary equations in $S_\beta$.
  Now consider an equation $e = (ax + by = c)$ in $S - X$
  and the two corresponding equations $z_e = ax_0$ and $z_e = -by_0$ in $S_0$.
  By definition, $\phi_0$ satisfies $z_e = ax_0$.
  Furthermore, we claim that it satisfies $z_e = -by_0$ if $\phi$ satisfies $e$.
  Indeed, since $\eta$ satisfies $e \in S-X$ we have that
  $a\eta(x) + b\eta(x) = c$. Therefore, if $a\phi(x) + b\phi(y) = c$, 
  then $a(\phi(x) - \eta(x)) + b(\phi(y) - \eta(y)) = 0$.
  It follows that $a\phi_0(x_0) = -b\phi(y_0) = \phi_0(z_e)$,
  and $\phi_0$ satisfies $z_e = - by_0$.
  Hence, $$\cost_{S_0}(\phi_0) \leq \cost_{S - X}(\phi) =
  \cost_{S}(\phi) - \cost_{X}(\phi) \leq k - \cost_{X}(\phi) = k -
  \cost_{X}(\beta) = k_\beta.$$
  
  For the opposite direction, we suppose that $(S, k)$ is a
  no-instance. Assume, with the aim of reaching a contradiction, that
  there exists $\beta : V(X) \to I$
  such that $S_0 \cup S_\beta$ admits an assignment 
  $\phi_0 : V(S_0) \to I$ that satisfies $S_\beta$ 
  and has $\cost_{S_0}(\phi_0) \leq \gamma k_\beta$.
  We claim that $S$ admits an assignment $\phi : V(S) \to I$ of cost at most
  $\gamma k$, contradicting our assumption that $(S,k)$ is a no-instance.
  Define $\phi : V(S) \to I$ by letting
  $\phi(x) = \phi_0(x_0) + \eta(x)$ for every $x \in V(S)$.
  To calculate the cost of $\phi$ in $S$,
  first observe that $\phi(x) = \beta(x) - \eta(x) + \eta(x) = \beta(x)$
  for every $x \in V(X)$.
  Hence, $\cost_X(\phi) = \cost_X(\beta)$.
  Now consider an equation $e = (ax + by = c)$ in $S - X$.
  We claim that if $\phi_0$ satisfies $z_e = ax_0$ and $z_e = -by_0$,
  then $\phi$ satisfies $e$:
  indeed, if $\phi_0(z_e) = a\phi_0(x_0) = -b\phi_0(y_0)$,
  we have $a\phi_0(x_0) + b\phi_0(y_0) = 0$ and
  $a(\phi(x) - \eta(x)) + b(\phi(y) - \eta(y)) = 0$;
  combined with $a\eta(x) + b\eta(y) = c$, this yields
  $a\phi(x) + b\phi(y) = c$.
  Hence, $\cost_{S - X}(\phi) \leq \cost_{S_0}(\phi_0)$.
  Putting these results together, we obtain that
  $$\cost_{S}(\phi) = \cost_X(\phi) + \cost_{S - X}(\phi) \leq 
  \cost_X(\beta) + \cost_{S_0}(\phi_0) \leq
  \cost_X(\beta) + \gamma k_\beta \leq $$
  $$\cost_X(\beta) + \gamma (k - \cost_X(\beta)) =
  \gamma k - \cost_X(\beta)(\gamma - 1) \leq \gamma k$$
  since $\gamma \geq 1$.
  Note that the above analysis also shows that if 
  $(S_0 \cup S_\beta, k_\beta)$ has a solution, then the set $Z=Z_0^X\cup Z_0'$ with
  $Z_0^X=\SB e \in X \SM e \textup{ is not satisfied by } \phi_0\SE$
  and $Z_0'=\SB e \in S-X \SM e_1 \textup{ or }e_2 \textup{ is not
    satisfied by }\phi_0 \SE$ is a solution for $(S,k)$. This shows
  that we can compute a solution for $(S,k)$ from a solution of one of
  the instances in $\cT$ in polynomial-time and this completes the proof of the lemma.
\end{proof}

We are now ready to show the following more precise version of
Theorem~\ref{the:simple-instances} that essentially follows from
Lemma~\ref{lem:branch-simple} using iterative compression.
\begin{theorem}
  \label{the:simple-instances-pre}
  Let $R$ be a finite commutative ring, 
  $I$ be an ideal in $R$, and
  $\gamma \geq 1$ be a constant. 
  If $\GAP{\gamma}$ $\minlinideal{2}{R}{I}$
  restricted to simple instances can be solved in time $f(S,k)$,
  then $\GAP{\gamma}$ $\minlinideal{2}{R}{I}$ can be solved in time
  $f(S,k)|I|^{2\gamma k}(|S|)^{\bigoh(1)}$.
\end{theorem}
\begin{proof}
  We use iterative compression together with
  Lemma~\ref{lem:branch-simple} to obtain the result. Therefore, let
  $(S,k)$ be an instance of $\GAP{\gamma}$ $\minlinideal{2}{R}{I}$ and
  let $e_1,\dotsc, e_m$ be an arbitrary ordering of the equations in
  $S$. Moreover, for every $i$ with $0 \leq i
  \leq m$, let $S_i$ be $S$ restricted to the equations in
  $e_1,\dotsc,e_i$. Starting from $S_0=\emptyset$, we now iteratively
  compute a solution $X_i$ for every $S_i$ from a solution $X_{i-1}$
  of $S_{i-1}$, i.e. for every $i$ it holds that $X_i$ is a set of at
  most $k\gamma$ equations of $S_i$ such that $S_i-X_i$ is satisfiable
  as follows. Initially, $X_0=\emptyset$ is clearly a solution for
  $S_0=\emptyset$. Therefore, let $X_{i-1}$ be a solution for
  $S_{i-1}$. If $X_{i-1}\cup \{e_i\}$ is also a solution for $S_i$,
  then we simple set $X_i=X_{i-1}\cup \{e_i\}$. Otherwise,
  $|X_{i-1}\cup \{e_i\}|>k\gamma$ because $X_{i-1}$ is a solution for
  $S_{i-1}$ and we can employ
  Lemma~\ref{lem:branch-simple} for $S=S_i$ and $X=X_i$ to compute the
  set $\cT$ of simple instances of $\GAP{\gamma}$
  $\minlinideal{2}{R}{I}$ in time $|I|^{2\gamma|X|}(|S|)^{\bigoh(1)}$
  such that $(S_i,k)$ has a solution if and only if so has one of the
  instances in $\cT$. We then use the algorithm for simple instances
  of $\GAP{\gamma}$ $\minlinideal{2}{R}{I}$ on every instance in
  $\cT$ and if none of these instances has a solution, we correctly
  return no. Otherwise, we compute a solution $X$ for $S_i$ from a
  solution to one of the instances of $\cT$ in polynomial-time and set
  $X_i=X$. This completes the description of the algorithm, whose
  runtime can be easily seen to be at most $f(S,k)|I|^{2\gamma k}(|S|)^{\bigoh(1)}$.
\end{proof}

\subsection{Computing Class Assignments}
\label{ssec:class-assignments}

This section is devoted to a proof of
Theorem~\ref{the:get-class-assignment}. Throughout this section, let $R$
be a finite commutative ring and let $I$ be an
ideal in $R$ that admits a matching coset partition $\equiv$ with
classes $\Gamma$ and $|\Gamma|=d$. Informally, we will show that we
can compute a class assignment $\tau: V(S) \rightarrow
\Gamma$ such that if $\mincost(S,I) \leq k$, then there exists an
assignment $V(S) \to I$  of cost at most $dk$ that agrees with $\tau$
for a given simple instance $S$ of $\lin{2}{R}$ and integer $k$.
To achieve this we introduce the class assignment graph (Section~\ref{sec:class-graph}) and show that
certain cuts in this graph correspond to class assignments (Section~\ref{sec:class-cuts}), which
themselves correspond to solutions of $\minlinideal{2}{R}{I}$. We then
use shadow removal and branching to compute these cuts (Section~\ref{sec:goodbye-shadows}).

\subsubsection{The Class Assignment Graph}
\label{sec:class-graph}

\newcommand{\homoI}{simple}

In the following, let $S$ be a \homoI{} instance of
$\linideal{2}{R}{I}$ and let $k$ be integer.  Without loss of generality, we assume that every equation of $S$ is
consistent (otherwise we can remove the equation and decrease $k$ by
$1$)
and non-trivial (otherwise we can remove the equation).
Before we define the class assignment graph, we need some simple definitions and observations about the equivalence classes given by $\equiv$.

\begin{definition}
  Let $e=(ax=y)$ be an equation over $R$ for variables $x$ and $y$
  and $a \in R$, and let $R_e \subseteq R^2$ be the relation
  $R_e=\{(c,d) \colon ac=d\}$.
  We say that \emph{$e$
    maps class $C \in \EQ$ to class $D\in \EQ$} if for every $(c,d) \in
  R_e$, it holds that $c \in C$ if and only if $d \in D$.
\end{definition}

Note that the equivalence classes in $\Gamma$ have direct connections with the annihilators
of $R$, i.e. $\Ann(a)=\Ann(b)$ whenever $a \equiv b$.
Given a equivalence class $C \in \EQ$, we can thus let
$\Ann(C)$ denote the common annihilator $\Ann(c)$ of all
elements $c \in C$. The following observation follows immediately
since $\equiv$ is a matching partition.

\begin{observation}\label{obs:eq-mapping}
  Let $e=(ax=y)$ be an equation over $R$ for variables $x$ and $y$ and
  $a \in R$, and let $C \in \EQ$. Then, either:
  \begin{itemize}
  \item
    $a \in \Ann(C)$ and $e$ is satisfied by setting $x$ to any value
    in $C$ and $y$ to $0$,
  \item
    $a \notin \Ann(C)$ and there is a unique class $D\in \EQ^{\neq 0}$ such
    that $e$ maps $C$ to $D$.
  \end{itemize}
\end{observation}

\newcommand{\EQM}[2]{\pi_{#1,#2}}
\newcommand{\mundef}{\textsf{Nil}}

The above observation allows us to define a matching function
between equivalence classes of $\equiv$ for any equation as follows.
Let $\EQM{e}{i} :
\EQ \rightarrow \EQ^{\neq 0}$, for some equation $e=(ax=y)$ with $a
\in R\setminus\{0\}$, be the mapping defined as follows.  For every $C \in \EQ$, we set:
\begin{itemize}
\item $\EQM{e}{i}(C)=0$ if $ax=0$ for every $x \in C$ and 
\item $\EQM{e}{i}(C)=D$ otherwise, where $D$
  is the unique equivalence class such that $e$ maps $C$ to $D$ defined by Observation~\ref{obs:eq-mapping}.
\end{itemize}
Moreover, for a consistent equation $e=(x=a)$ with $a \in R$,
we let $\EQM{e}{i}$ be the unique equivalence class consistent with $e$.

\begin{figure}
 \centering
  \begin{tikzpicture}[node distance=0.5cm]
    \tikzstyle{no}=[draw,circle, inner sep=1pt, fill]
    \tikzstyle{edno}=[inner sep=3pt]
    \tikzstyle{ed}=[draw,color=black, line width=1pt]

    \draw[xscale=2,yscale=1] 
    node[no, label=left:$s$] (s) {}
    (s) +(1cm,-1cm) node[no, label=below:$b_O$] (b1) {}
    (b1) +(1cm,0cm) node[no, label=below:$d_E$] (d2) {}
    (d2) +(1cm,0cm) node[no, label=below:$r_E$] (r2) {}
    (r2) +(1cm,0cm) node[no, label=below:$u_E$] (u2) {}
    (u2) +(1cm,0cm) node[no, label=below:$d_O$] (d1) {}
    (s) +(1cm,1cm) node[no, label=above:$a_O$] (a1) {}
    (a1) +(1cm,0cm) node[no, label=above:$c_E$] (c2) {}
    (c2) +(1cm,0cm) node[no, label=above:$u_O$] (u1) {}
    (u1) +(1cm,0cm) node[no, label=above:$r_O$] (r1) {}
    (r1) +(1cm,0cm) node[no, label=above:$c_O$] (c1) {}

    (c1) +(1cm,-1cm) node[no, label=right:$t$] (t) {}
    ;

    \draw
    (s) -- node[edno, midway, rotate=-26, anchor=north] (sb1) {$b=1$} (b1)
    (b1) -- node[edno, midway, anchor=north] (db) {$2b=d$} (d2)
    (d2) -- node[edno, midway, anchor=north] (rd) {$d=r$} (r2)
    (r2) -- node[edno, midway, anchor=north] (ur) {$r=u$} (u2)
    ;

    \draw
    (s) -- node[edno, midway, rotate=26, anchor=south] (sa1) {$a=1$} (a1)
    (a1) -- node[edno, midway, anchor=south] (ac) {$2a=c$} (c2)
    (c2) -- node[edno, midway, anchor=south] (cu) {$c=2u$} (u1)
    (u1) -- node[edno, midway, anchor=south] (ur) {$u=r$} (r1)
    ;

    \draw
    (t) -- node[edno, midway, rotate=-26, anchor=north] (c1t1)
    {$c=2a$} (c1)
    (t) -- node[edno, midway, rotate=-26, anchor=south] (c1t2) {$c=2u$} (c1)
    ;
    \draw
    (t) -- node[edno, midway, rotate=26, anchor=north] (d1t1)
    {$d=2b$} (d1)
    ;

  \end{tikzpicture}
 \caption{Let $S$ be the instance of $\lin{2}{\ZZ_4}$ with variables
   $a$ , $b$, $c$, $d$, $u$, $r$ and equations $a=1$, $b=1$, $2a=c$,
   $c=2u$, $u=r$, $2b=d$, and $d=r$. The figure illustrate the class
   assignment graph $G=G_I(S)$ for the ideal $I=(1)$. Note that $I$
   has only two non-zero equivalence classes, namely, $E=\{2\}$ and
   $O=\{1,3\}$. Every edge of $G$ is annotated by the equation that
   implies it. A minimum conformal $st$-cut is given by the two edges
   that correspond to the equation $u=r$ and corresponds to the class
   assignment $a=O$, $b=O$, $c=E$, $d=E$, $u=O$, and
   $r=E$. Note that $G$ has only one minimal conformal $st$-cut
   closest to $s$, namely $\{a_Oc_E, b_Od_E\}$. This $st$-cut
   corresponds to the class assignment $a=O$, $b=O$, and
   $c=d=u=r=0$. Therefore, the optimum solution for $S$ only removes
   the equation $u=r$, however, any solution that corresponds to a
   minimum conformal $st$-cut closest to $s$ has to remove the
   equations $2a=c$ and $2b=d$.}
 \label{fig:classassignment}
\end{figure}
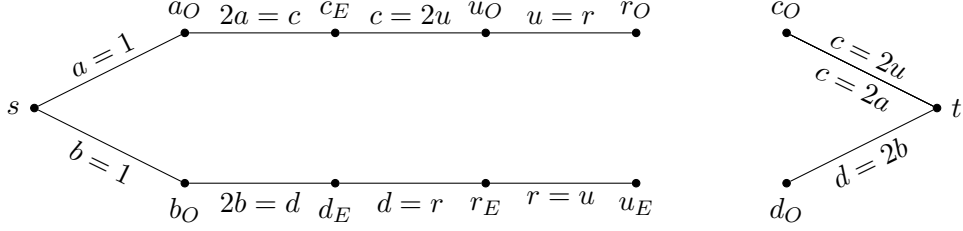

We are now ready to define the \emph{class assignment} graph
$G=G_I(S)$; see Figure~\ref{fig:classassignment} for an illustration
of such a graph. The graph $G$ has two distinguished vertices $s$ and $t$
together with vertices $x_{C}$ for every $x \in V(S)$ and every non-zero
class $C \in \EQ^{\neq 0}$. Moreover, $G$ contains the following edges for each equation.
\begin{itemize}
\item For a (crisp) equation $e=(ax=y)$ we do the following. For every $C
  \in \EQ$, we add the (crisp) edge
  $x_Cy_{\EQM{e}{i}(C)}$ if $\EQM{e}{i}(C) \neq 0$.
  For
  every $D \in \EQ^{\neq 0}$ such that $\EQM{e}{i}^{-1}(D)$ is
  undefined, we add the (crisp) edge $y_Dt$.
\item For a (crisp) equation of the form $e=(x = 0)$, we add the (crisp)
  edge $x_{C}t$ for every class $C \in \EQ^{\neq 0}$.
\item For a (crisp) equation of the form $e=(x = b)$, where $b\neq 0$ ,
  we add the (crisp) edge $sx_{\EQM{e}{i}}$.
\end{itemize}

Intuitively, every node of $G$ corresponds to a Boolean variable and
every edge $e$ of $G$ corresponds to an ``if and only if'' between the two
Boolean variables connected by $e$. Moreover, every
assignment $\varphi$ of the variables of $S$ naturally corresponds to a
Boolean assignment, denoted by $\BA{\varphi}$
of the vertices in $G$ by setting $s=1$, $t=0$, and:
\begin{itemize}
\item if $\varphi(x)$ belongs to the non-zero class $C$, then we set
  $x_{C}=1$ and $x_{C'}=0$ for every non-zero class $C'$ not
  equal to $C$,
\item if $\varphi(x)=0$, we set $x_{C}=0$ for every non-zero class $C$.
\end{itemize}
We say that an
edge $e$ of $G$ is satisfied by $\varphi$ if $\BA{\varphi}$ satisfies the
``if and only if'' Boolean constraint represented by that edge.
\begin{observation}\label{obs:ass-ba-ass}
  Let $S$ be a \homoI{} instance of $\lin{2}{R}$, let $\varphi$
  be an assignment of $S$ and let $e$ be an equation in
  $S$. If $\varphi$ satisfies $e$, then $\BA{\varphi}$ satisfies all
  edges corresponding to $e$ in $G(S)$.
\end{observation}
\begin{corollary}\label{cor:satassG}
  Let $S$ be a \homoI{} instance of $\lin{2}{R}$ having a
  satisfying assignment $\varphi$. Then, $\BA{\varphi}$ satisfies all edges
  of $G(S)$.
\end{corollary}

\subsubsection{Cuts in the Class Assignment Graph}
\label{sec:class-cuts}

In this section, we introduce {\em conformal cuts} and show how they
relate to class
assignments corresponding to solutions of
a $\minlinideal{2}{R}{I}$ instance. In the following, let $S$ be a \homoI{} instance of
$\linideal{2}{R}{I}$ and let $G=G_I(S)$.
We say that an $st$-cut $X$ in $G$ is \emph{conformal} if
for every variable $x \in V(S)$ at most one vertex $x_{C}$ for some $C
\in \EQ$ is
connected to $s$ in $G-X$. Please refer to
Figure~\ref{fig:classassignment} for an illustration of conformal cuts
in the class assignment graph.
If $X$ is a conformal $st$-cut in $G$, then we say that a variable $x$
is \emph{decided} with respect to $X$ if (exactly) one vertex $x_{C}$ is reachable from $s$ in
$G-X$; and otherwise we say that $x$ is \emph{undecided} with respect to $X$. Moreover, we denote by $\clasn{X}$ the assignment of variables
of $S$ to classes in $\EQ$ implied by $X$, i.e. $\clasn{X}(x)=0$ if
$x$ is undecided and otherwise $\clasn{X}(x)=C$, where $C$
is the unique non-zero class in $\EQ$ such that $x_C$ is
reachable from $s$ in $G-X$.
We say that an assignment $\varphi$ of $S$ \emph{agrees with} $X$
if $\varphi(x)$ is in the class $\clasn{X}(x)$.
Note that if some assignment agrees with $X$, then $X$ is
conformal. The following auxiliary lemma characterizes, which edges of
$G$ are satisfied by an assignment $\varphi$ of $S$ after removing a
set $Y$ of edges from $G$.

\begin{lemma}\label{lem:ze-ba-ass}
  Let $Y$ be a set of edges of $G$ and let $\varphi$ be an assignment
  of $S$. Then, $\BA{\varphi}$ satisfies all edges 
  reachable from $s$ in $G-Y$ if and
  only if $\BA{\varphi}$ sets all Boolean variables reachable from $s$ in
  $G-Y$ to $1$. Similarly, $\BA{\varphi}$ satisfies all edges
  reachable from $t$ in $G-Y$ if and
  only if $\BA{\varphi}$ sets all Boolean variables reachable from $t$ in
  $G-Y$ to $0$.
\end{lemma}
\begin{proof}
  This follows because $\BA{\varphi}(s)=1$ and $\BA{\varphi}(t)=0$ for any $\varphi$
  and the
  fact that every edge of $G$ corresponds to an ``if and only if''
  between the variables corresponding to its two endpoints.
\end{proof}

For a set $Z$ of equations of $S$, we denote by $\ed(Z)$ the set of
all edges of $G$ corresponding to an equation of $Z$.
Similarly, for a set $Y$
of edges of $G$, we denote by $\con(Y)$ 
all equations of $S$ having a corresponding edge in $Y$. Moreover, if
$Y$ is an $st$-cut in $G$, we denote by $\sep(Y)$ the unique
minimal $st$-cut contained in $Y$ that is \emph{closest} to $s$ in $G$.
Finally, for an optimal
solution $Z$ of $S$, we let $\comp{Z}$
be the set of equations
$Z\setminus \con(\sep(\ed(Z)))$, i.e. all equations in $Z$ that do
not have an edge in $\sep(\ed(Z))$.

We are now ready to show the main result of this subsection, which
informally shows that every solution $Z$ for $(S,k)$ gives rise to a
conformal $st$-cut of size at most $|\EQ^{\neq 0}| \cdot |Z|$ that
agrees with some satisfying assignment of $S-Z$ and thus establishes
the required connection between solutions of $(S,k)$ and conformal
$st$-cuts in the class assignment graph.
\begin{lemma} \label{lem:deleted-edges}
  Let $Z$ be a set of equations such that $S-Z$ is satisfiable and let $Y=\ed(Z)$.
  Then, $|Y|\leq |\EQ^{\neq 0}| \cdot |Z|$ and $Y$ is an $st$-cut. Moreover,
  $Y'=\sep(Y)$ satisfies:
  \begin{enumerate}
    \item $Y'$ is conformal.
    \item There is a satisfying assignment for $S-Z$ that agrees with $Y'$.
  \end{enumerate}
\end{lemma}
\begin{proof}
  Since every equation in $Z$ corresponds to at most
  $|\EQ^{\neq 0}|$ edges in $G$ by construction, we obtain that
  $|Y|\leq |\EQ^{\neq 0}| \cdot |Z|$.
  Let $\varphi$ be a satisfying assignment of $S-Z$.
  Corollary~\ref{cor:satassG} implies that $\BA{\varphi}$ satisfies
  all edges of $G-Y$. Therefore, it follows from
  Lemma~\ref{lem:ze-ba-ass} that $\BA{\varphi}$ sets all vertices
  reachable from $s$ in $G-Y$ to $1$ and all vertices reachable from $t$ in $G-Y$ to $0$.
  Thus, $Y$ is an $st$-cut, because otherwise $t$ would
  have to be set to $1$ by $\BA{\varphi}$ since it would be reachable from
  $s$ in $G-Y$. Therefore, $Y'=\sep(Y)$ exists. Because $Y'$ is
  closest to $s$, it holds that a vertex is reachable from $s$ in
  $G-Y$ if and only if it is reachable from $s$ in $G-Y'$.
  Therefore, if at least two vertices $x_{C}$ and $x_{C'}$ for some
  distinct non-zero classes $C$ and $C'$ are
  reachable from $s$ in $G-Y'$ for some variable $x$, then all of them must be set to $1$ by
  $\BA{\varphi}$, which is not possible due to the definition of
  $\BA{\varphi}$. We conclude that $Y'$ is conformal.

  Let $D$ be the set
  of all variables of $S$ such that no vertex $x_{C}$ is
  reachable from $s$ in $G-Y'$. Let $\varphi'$ be the assignment for $S$ such that $\varphi'(x)=0$ if
  $x \in D$ and $\varphi'(x)=\varphi(x)$ otherwise. Clearly, $\varphi'$ agrees
  with $Y'$, because $\varphi$ agrees with all variables not in $D$ and
  all other variables are correctly set to $0$ by $\varphi'$. It
  therefore only remains to show that $\varphi'$
  still satisfies $S-Z$.

  Consider a unary equation $e$ of $S-Z$ on variable $x$.
  If $x
  \notin D$, then $\varphi'(x)=\varphi(x)$ and $\varphi'$ satisfies
  $e$. So suppose that $x \in D$. If $e$ is of the form $ax=0$, then
  $\varphi'$ satisfies $e$. Otherwise, $e$ is of the form $ax=b$
  for some $b\neq 0$ and $G-Y'$ contains the edge
  $sx_{\EQM{e}{i}}$. Therefore, $x_{\EQM{e}{i}}$ is reachable from $s$ in
  $G-Y'$ contradicting our assumption that $x \in D$.

  Consider instead a binary equation $e$ of $S-Z$ on
  two variables $x$ and $y$. If $x,y \in D$, then $e$ is clearly
  satisfied by $\varphi'$. Similarly, if $x,y \notin D$, then $e$ is
  also satisfied by $\varphi'$, because $\varphi'(x)=\varphi(x)$ and
  $\varphi'(y)=\varphi(y)$ and $\varphi$ satisfies $e$. So suppose that
  $|\{x,y\}\cap D|=1$ and that $e$ is of the form $ax =y$ for some $a\in R$. Recall that
  $$\ed(e)=\SB x_{C}y_{\EQM{e}{i}(C)} \SM C \in \EQ \land
  \EQM{e}{i}(C)\neq 0\SE \cup \SB y_Dt \SM D \in \EQ \land
  \EQM{e}{i}^{-1}(D) \text{ is undefined} \SE$$
  and note that $\ed(e) \subseteq E(G)\setminus
  \ed(Z)$. Suppose first that $x \notin D$ and
  $y \in D$. Because $x \notin D$, there is a vertex
  $x_{C}$ reachable from $s$ in $G-Y'$. If $\EQM{e}{i}(C) \neq 0$,
  then because $x_{C}y_{\EQM{e}{i}(C)} \notin
  Y'$, the vertex $y_{\EQM{e}{i}(C)}$ is reachable from $s$ in $G-Y'$ contradicting our assumption
  that $y \in D$. Otherwise, $\EQM{e}{i}(C)=0$, which together with our assumption that
  $\varphi'(y)=0$ shows that $e$ is satisfied by $\varphi'$.
  
  Suppose now that $x \in D$ and
  $y \notin D$. Because $y \notin D$, there is a vertex
  $y_{D}$ reachable from $s$ in $G-Y'$. Because $e$ is satisfied by $\varphi$, it holds that
  $\EQM{e}{i}^{-1}(D)$ is defined. But then, because $x_{\EQM{e}{i}^{-1}(D)}y_{D} \notin
  Y'$, the vertex $x_{\EQM{e}{i}^{-1}(D)}$ is reachable from $s$ in $G-Y'$ and this contradicts our assumption
  that $x \in D$. 
\end{proof}

\subsubsection{Shadow Removal}
\label{sec:goodbye-shadows}

We continue by showing how the shadow removal technique
(introduced in~\cite{marx2014fixed} and improved
in~\cite{chitnis2015directed}) can be used for computing conformal cuts that correspond to solutions of
a $\minlinideal{2}{R}{I}$ instance.
We follow~\cite{chitnis2015directed} and begin by importing some definitions,
which we translate from directed graphs to undirected graphs to fit
our setting; to get back to directed graphs one simply has to think of
an undirected graph as a directed graph obtained after replacing each
undirected edge with two directed arcs in both directions.
Let $G$ be an undirected graph.
Let $\cF$ be a set of connected subgraphs of $G$.
A set $T \subseteq V(G)$ is an \emph{$\cF$-transversal} 
if $T$ intersects every subgraph in $\cF$.
Conversely, if $T$ is an $\cF$-transversal,
we say that $\cF$ is \emph{$T$-connected}.

\begin{theorem}[\cite{chitnis2015directed}]
  \label{thm:shadow-cover}
  Let $G$ be an undirected graph, $T \subseteq V(G)$ and $k \in \naturals$.
  There is a randomized algorithm that takes $(G, T, k)$
  as input
  and returns in $O^*(4^k)$ time a set $W \subseteq V(G) \setminus T$
  such that the following holds with probability $2^{-O(k^2)}$.
  For every $T$-connected family of connected subgraphs $\cF$ in $G$,
  if there is an $\cF$-transversal of size at most $k$ in $V(G) \setminus T$,
  then there is an $\cF$-transversal $Y \subseteq V(G) \setminus (W \cup T)$ 
  of size at most $k$ such that every vertex $v \notin W \cup Y$
  is connected to $T$ in $G - Y$.
\end{theorem}

The following lemma is crucial for the application of shadow removal
(Theorem~\ref{thm:shadow-cover}). Informally, it shows that if $Z$ is
a solution, i.e. a set of equations such that $S-Z$ is satisfiable,
then we can obtain a (not too large) new solution $Z'=(\comp{Z} \cup \con(Y'))$ by replacing
the corresponding conformal minimal $sA$-cut $Y=\sep(\ed(Z))$, where $A$
is the set of vertices in $G$ not reachable from $s$ in $G-Y$, by any
minimal $sA$-cut $Y'$.
\begin{lemma}
  \label{lem:anti-sol}
  Let $S$ be a \homoI{} instance of $\linideal{2}{R}{I}$ and let $G=G_I(S)$.
  Moreover, let $Z$ be a set of equations such that $S-Z$ is satisfiable, 
  $Y=\sep(\ed(Z))$, $A$ be the set of all
  vertices in $G$ that are not reachable from $s$ in $G-Y$, and let
  $Y'$ be an $sA$-cut in $G$.
  Then, there is an assignment $\varphi : V(S) \rightarrow I$ of $S$ that satisfies
  $S-Z'$ and agrees with $Y'$, where $Z'=(\comp{Z} \cup \con(Y'))$.
\end{lemma}
\begin{proof}
  Lemma~\ref{lem:deleted-edges} implies that $Y$ is conformal and there is a satisfying assignment
  $\varphi$ for $S-Z$ that agrees with $Y$.
  Because $Y'$ is also an $sA$-cut in
  $G$, if no vertex $x_{C}$ is reachable from $s$ in $G-Y$ for
  some variable $x$ of $S$, then
  the same applies in $G-Y'$. Let $D$ be the set of all variables
  $x$ of $S$ such that some vertex $x_{C}$ is reachable from $s$ in $G-Y$
  but that is not the case in $G-Y'$. Let $\varphi'$
  be the assignment obtained from $\varphi$ by setting all variables
  in $D$ to $0$. Then, $\varphi'$ agrees with $Y'$.
  We
  claim that $\varphi'$ also satisfies $S-Z'$, where $Z'=(\comp{Z} \cup
  \con(Y'))$.

  Consider a unary equation $e$ of $S-Z'$ on variable $x$.
  If $x
  \notin D$, then $\varphi'(x)=\varphi(x)$ and therefore $\varphi'$ satisfies
  $e$ (because $e$ is crisp and therefore $e \notin Z$). So suppose that $x \in D$. If $e$ is of the form $xa=0$, then
  $\varphi'$ satisfies $e$. Otherwise, $e$ is of the form $ax=b$
  for some $b\neq 0$ and $G-Y'$ contains the edge
  $sx_{\EQM{e}{i}}$. Therefore, $x_{\EQM{e}{i}}$ is reachable from $s$ in
  $G-Y'$ contradicting our assumption that $x \in D$.

  Now, consider a binary
  equation $e=(ax=y)$ of $S-Z'$ on variables $x$ and
  $y$ and first consider the case that $e \in Z$. Clearly, if neither
  a vertex $x_C$ nor a vertex $y_C$ is reachable from $s$ in $G-Y'$,
  then $\varphi'(x)=\varphi'(y)=0$ and therefore $e$ is satisfied by
  $\varphi'$. We show next
  that either no vertex $x_c$ or no vertex $y_C$ is reachable from $s$
  in $S-Y'$. Suppose for a contradiction that $x_{C_x}$ and
  $y_{C_y}$ are reachable from $s$ in $S-Y'$.
  Let $h$ be an arbitrary edge in $\ed(e)\cap Y$; note that
  such an edge $h$ exists because $e \in Z\setminus Z'$.
  Because $Y$ is a minimal st-cut, it follows that exactly one
  endpoint of $h$ is reachable from $s$ in $G-Y$ and therefore
  either $x_C$ or $y_C$ (endpoint of $h$) for some $C \in
  \EQ$ must be reachable from $s$ in $G-Y$. We assume
  without loss of generality that $x_C$ is reachable from $s$ in $G-Y$. Because $Y'$ is an
  $sA$-cut and $Y'$ does not contain $h$, it holds
  that $x_C$ is not reachable from $s$ in $G-Y'$. But then $C\neq C_x$
  and both $x_C$
  and $x_{C_x}$ are reachable from $s$ in $G-Y$, which contradicts our
  assumption that $Y$ is conformal.
  It remains to consider the case that there is a vertex $x_C$ that is reachable
  from $s$ in $G-Y'$ but no vertex $y_C$ is reachable from $s$ in
  $G-Y'$; the case that there is a vertex $y_C$ reachable from $s$ in
  $G-Y'$ but no vertex $x_C$ reachable from $s$ in $G-Y'$ is analogous.
  Since $Y'\cap \ed(e)=\emptyset$, we obtain that $\EQM{e}{i}(C)=0$
  since otherwise either $t$ or some $y_{C'}$ would be reachable from
  $s$ in $G-Y'$. Because $\varphi'(y)=0$, it follows that $e$ is
  satisfied by $\varphi'$. This completes the proof for the case that $e
  \in Z$.

  So suppose instead that $e \notin Z$. In this case $\varphi$ satisfies $e$ and
  therefore also $\varphi'$ satisfies $e$ unless exactly one of $x$ and
  $y$ is not in $D$. We distinguish the following cases:
  \begin{itemize}
  \item $x \notin D$ and $y \in D$. If there is no vertex $x_C$ that is
    reachable from $s$ in $G-Y'$, then the same holds in $G-Y$ so
    $\varphi'(x)=\varphi(x)=\varphi'(y)=0$, which shows that $\varphi'$
    satisfies $e$. Otherwise, let $x_C$ be reachable from $s$ in
    $G-Y'$. Then, $\EQM{e}{i}(C)=0$ since otherwise either $t$ or
    $y_{\EQM{e}{i}(C)}$ is also reachable from $s$ in $G-Y'$ (because
    $Y'\cap \ed(e)=\emptyset$), which in
    the former case contradicts our assumption that $Y'$ is an
    $st$-cut and which in the latter case contradicts our assumption
    that $y \in D$. Therefore, $\varphi'$ satisfies $e$ (because
    $\varphi'(y)=0$).
  \item $x \in D$ and $y \notin D$. We first show that there is no
    vertex $y_C$ that is reachable from $s$ in $G-Y'$. Suppose there
    is such a vertex $y_C$. Then, $\EQM{e}{i}^{-1}(C)$ is undefined
    since otherwise $x_{\EQM{e}{i}^{-1}(C)}$ is reachable from $s$ in
    $G-Y'$ (because $Y'\cap \ed(e)=\emptyset$), which contradicts our
    assumption that $x \in D$. But then $y_Ct \in E(G-Y')$ and
    $t$ is reachable from $s$ in $G-Y'$, which contradicts
    our assumption that $Y'$ is an $st$-cut. Hence, there is no vertex $y_C$ that is
    reachable from $s$ in $G-Y'$, which implies that the same holds in $G-Y$ so $\varphi'(y)=\varphi(y)=\varphi'(x)=0$, which shows that $\varphi'$
    satisfies $e$. 
  \end{itemize}

\end{proof}

Let $G=G_I(S)$ and for a set $W \subseteq
V(G)$, let $\delta(W)$ be the set of edges incident to a vertex in $W$
and a vertex in $V(G)\setminus W$.
The forthcoming Theorem~\ref{thm:our-shadow-cover} 
provides a version of shadow
removal adopted to our problem. 
Informally, it provides us with a set
$W \subseteq V(G)$ such that we only
have to look for conformal $st$-cuts that are subsets of $\delta(W)$ to
obtain our class assignment; in fact it even shows that for every component $C$ of $G[W]$ either all edges in $\delta(C)$ are part of the cut or no edge of $\delta(C)$ is part of the cut. We will use this fact in
Theorem~\ref{thm:list-conformal-cuts} to find a conformal $st$-cut by branching on
which components of $G[W]$ are reachable from $s$.

More formally, if $Z$ is a set
of at most $k$ equations such that $S-Z$ is satisfiable and $A$ is the
set of vertices not reachable from $s$ in $G$ minus the conformal
$st$-cut $\sep(\ed(Z))$ (see Lemma~\ref{lem:deleted-edges}), then the
theorem provides us with a set $W \subseteq V(G)$ such that there is a
conformal $sA$-cut $Y'$ within $\delta(W)$ of size at most
$|Z| \cdot |\Gamma^{\neq 0}|$ such that there is an assignment $\varphi
: V(S) \rightarrow I$ for the variables in $S$ that satisfies
$S-(\comp{Z} \cup \con(Y'))$ and
agrees with $Y'$. The main idea behind the proof is the application of Theorem~\ref{thm:shadow-cover} to the set $\FF$ of all walks from $s$ to $A$ in $G$ to obtain the set $W$ and to employ Lemma~\ref{lem:anti-sol} to obtain the new solution that corresponds to the minimum $sA$-cut $Y' \subseteq \delta(W)$.

\begin{theorem}
  \label{thm:our-shadow-cover}
  Let $S$ be a \homoI{} instance of $\linideal{2}{R}{I}$ and let $G=G_I(S)$.
  Moreover, let $Z$ be a set of equations of size $k$ such that $S-Z$ is satisfiable, 
  $Y=\sep(\ed(Z))$, and let $A$ be the set of all
  vertices in $G$ that are not reachable from $s$ in $G-Y$.
  There is a randomized algorithm that in $\bigohs(4^{|\EQ^{\neq 0}|k})$ time takes $(G,k)$
  as input and returns  a set $W \subseteq V(G) \setminus \{s\}$
  such that the following holds with probability $2^{-O((|\EQ^{\neq 0}|k)^2)}$.
  There is a (minimal) $sA$-cut $Y'$ of size at most $|\EQ^{\neq 0}|k$
  satisfying:
  \begin{enumerate}
  \item every vertex $v \notin W$ is connected to $s$ in $G - Y'$,
  \item %
    $Y' \subseteq \delta(W)$, where $\delta(W)$ is the set of edges
    incident to a vertex in $W$ and a vertex in $V(G)\setminus W$,
  \item there is an assignment $\varphi : V(S) \rightarrow I$ for the variables in $S$ that satisfies
    $S-(\comp{Z} \cup \con(Y'))$ and
    agrees with $Y'$.
  \end{enumerate}
  Moreover, for every component $C$ of $G[W]$ the following holds:
  \begin{itemize}
  \item[4.] either $Y'\cap\delta(C)=\emptyset$ or $\delta(C)\subseteq Y'$,
  \item[5.] if $t \in C$, then $\delta(C) \subseteq Y'$,
  \item[6.] if $x_{\alpha},x_{\alpha'} \in C$ for some variable $x$
    and $\alpha\neq \alpha'$, then $\delta(C)
    \subseteq Y'$,
  \item[7.] if $C$ contains some $x_\alpha$ for some decided
    variable $x$ w.r.t. $Y'$, then $\delta(C) \subseteq Y'$,
  \end{itemize}
\end{theorem}
\begin{proof}
  In order to apply Theorem~\ref{thm:shadow-cover}, we first transform $G$ into
  an undirected graph $G'$ as follows. 
  For every vertex $a \in V(G)$, we create a clique on $|\EQ^{\neq 0}| \cdot k+1$
  vertices $K(a) = \{a^1, \dots, a^{|\EQ| \cdot k+1}\}$ in $G'$.
  For every edge $ab \in E(G)$, we introduce an auxiliary vertex $z_{ab}$ and connect it to $a^i$ and $b^i$
  for all $1 \leq i \leq |\EQ^{\neq 0}| \cdot k+1$. Informally, the
  construction of $G'$ is required to ensure
  that every $sA$-cut in $G$ corresponds to an equally sized
  transversal in $G'$.

  We now apply Theorem~\ref{thm:shadow-cover} to the tuple
  $(G',T,|\EQ^{\neq 0}| \cdot k)$, where $T=K(s)$. Let $\cF$ be the set of all walks in $G'$ from
  $K(s)$ to $\bigcup_{a \in A}K(a)$; clearly $\cF$ is $T$-connected.

  We first show that that if $X \subseteq E(G)$ is an $sA$-cut in $G$, then the
  set $X'=\SB z_e \SM e \in X\SE$ is a $\cF$-transversal in $G'$. This
  is because any path from $K(s)$ to $\bigcup_{a \in A}K(a)$ in
  $G'-X'$ corresponds to a path from $s$ to $A$ in $G$. Similarly, if
  $X'\subseteq V(G')$ is an $\cF$-transversal in $G'$ of size at most
  $|\EQ^{\neq 0}| \cdot k$, then the set $X=\SB e \SM z_e \in X'\SE$ is an $sA$-cut in
  $G$. To see this, first observe that $|X'|\leq |\EQ^{\neq 0}|k$ implies
  that $X'$ cannot contain all $|\EQ^{\neq 0}| \cdot k+1$ vertices in $K(a)$ for any
  vertex $a \in V(G)$ and 
  therefore any path from $s$ to $A$ in $G-X$ corresponds to a path
  from $K(s)$ to $\bigcup_{a \in A}K(a)$ in $G'-X'$.

  Let $W' \subseteq V(G')\setminus T$ be the set of vertices of $G'$
  satisfying the properties stated in Theorem~\ref{thm:shadow-cover},
  i.e. if there is a $\cF$-transversal of size at most $|\EQ^{\neq 0}| \cdot k$ in $V(G') \setminus K(s)$,
  then there is an $\cF$-transversal $R$ in $V(G') \setminus (W' \cup K(s))$ 
  of size at most $|\EQ^{\neq 0}| \cdot k$ such that every vertex $v \notin W \cup R$
  is connected to $K(s)$ in $G' - R$. We claim that the set $W=\SB a \in
  V(G) \SM K(a)\subseteq W'\SE$ satisfies the properties given in the
  statement of the lemma.
  First note that without loss of generality we may assume that $z_e \in W'$ for every
  edge $e=uv \in E(G)$ such that $u,v \in W$; this is because $z_e$
  can never be reachable from $K(s)$ if $K(u)$ and $K(v)$ are both contained
  in $W'$.
  
  Because $W' \subseteq V(G')\setminus K(s)$, we obtain that $W
  \subseteq V(G)\setminus \{s\}$. Moreover, because $Y$ is an
  $sA$-cut of size at most $|\EQ^{\neq 0}|k$ in $G$, we obtain that the set $R=\SB z_e
  \SM e \in Y\SE$ is a $\cF$-transversal of size at most $|\EQ^{\neq 0}|k$ in
  $V(G')\setminus K(s)$. Therefore, by Theorem~\ref{thm:shadow-cover}
  there is a $\cF$-transversal $R' \subseteq V(G') \setminus (W' \cup K(s))$ 
  of size at most $|\EQ^{\neq 0}|k$ such that every vertex $v \notin W' \cup R'$
  is connected to $K(s)$ in $G' - R'$. Note that without loss of generality we may
  assume that $R'$ is a minimal $\cF$-traversal; otherwise we can
  remove unnecessary vertices from $R'$ to make it minimal without
  changing the properties of $R'$. Therefore, 
  $R' \subseteq \SB z_e \SM e \in E(G)\SE$ also holds.
  Let $Y'$ be the set $\SB e \in E(G) \SM z_e \in
  R'\SE$. Because $R'$ is a $\cF$-transversal of size at most $|\EQ^{\neq 0}| \cdot k$ in
  $G'$,
  it holds that $Y'$ is an $sA$-cut of size at most $|\EQ^{\neq 0}| \cdot k$ in
  $G$. Moreover, because $R' \subseteq V(G') \setminus (W' \cup K(s))$
  and, as we observed above, $W'$ contains $z_e$ for every edge
  $e=uv$ with $u,v\in W$,
  we obtain that
  $Y' \subseteq E(G)\setminus E(G[W])$. Since every vertex $v \notin
  W'\cup R'$ is connected to $K(s)$ in
  $G'-R'$, we see that every vertex $u \notin W$ is connected to $s$
  in $G-Y'$ showing 1. Moreover, because $Y'$ is a minimal $sA$-cut we
  also obtain $Y'\subseteq \delta(W)$, which shows 2.
  Since Item~3. follows directly from Lemma~\ref{lem:anti-sol}, it only
  remains to show 4.--7.

  Let $C$ be a component of $G[W]$. Because $Y'\cap E(C)=\emptyset$
  and every vertex in $N(C)$, where $N(C)$ denotes the set of all
  neighbours of $C$ outside of $C$,  is reachable from $s$ in
  $G-Y'$, it follows that if $\delta(C) \setminus Y'\neq
  \emptyset$, then all vertices in $C$ are reachable from $s$ in
  $G-Y'$. Therefore, because $Y'$ is an inclusion-wise minimal
  $sA$-cut in $G$, we can assume that either
  $Y'\cap\delta(C)=\emptyset$ or $\delta(C)\subseteq Y'$;
  otherwise removing the edges in $\delta(C)$ from $Y'$ gives again
  an $sA$-cut showing that $Y'$ is not inclusion-wise minimal. This shows 4.
  Therefore, if $t \in C$ then $\delta(C)\subseteq Y'$
  since otherwise $Y'$ is would not be an $sA$-cut in
  $G$, and this shows 5. Item 6 can now be shown similarly: if a component $C$ contains two distinct vertices $x_C$
  and $x_{C'}$ for some variable $x$ of $S$, then $\delta(C)\subseteq Y'$
  because otherwise $Y'$ is not conformal. Moreover, the same
  applies if $C$ contains some $x_C$ for some decided
  variable $x$ w.r.t. $Y'$, which shows 7.
\end{proof}

\newcommand{\trip}{\mathcal{T}}
\begin{algorithm}[htb]
  \caption{Method for branching.} \label{alg:branch}
  \small
  \begin{algorithmic}[1]
    \INPUT the set $\CCC$ of remaining components, the
    remaining budget $b$
    \OUTPUT the set $\cY$ of conformal cuts
    \Function{\textsc{branch}}{$\CCC$, $b$}
    \State $U\gets$ set of $W$-undecided variables\label{alg:branch-lu}
    \State \Return \Call{branchUndecided}{$\CCC$, $b$, $U$}
    \EndFunction
  \end{algorithmic}
\end{algorithm}

\begin{algorithm}[htb]
  \caption{Method for branching on undecided variables.} \label{alg:branchundecided}
  \small
  \begin{algorithmic}[1]
    \INPUT the set $\CCC$ of remaining components, the
    remaining budget $b$, set of remaining undecided variables $U$
    \Function{\textsc{branchUndecided}}{$\CCC$, $b$, $U$}
    \State $\cY\gets \emptyset$
    \If{$U \neq \emptyset$}
    \State $u \gets $ $U$.first()
    \State $\CCC_u \gets \kappa(\CCC,u)$\label{alg:bu:lubb}
    \If{$\CCC_u \neq \emptyset$}
    \If{$\delta(\CCC_u)\leq b$}\label{algbu:lbb}
    \State $\cY\gets \cY \; \cup \; $\Call{branchUndecided}{$\CCC\setminus \CCC_u$,
      $b-\delta(\CCC_u)$, $U\setminus \{u\}$}
    \EndIf
    \For{$C \in \CCC_u$}
    \State $R \gets (\CCC_u\setminus \{C\})$
    \If{$\delta(R) \leq b$}\label{alg:bu:tb}
    \State $\cY\gets \cY \; \cup \; $\Call{branchUndecided}{$\CCC\setminus R$, $b-\delta(R)$, $U\setminus\{u\}$}
    \EndIf
    \EndFor
    \EndIf\label{algbu:lbe}
    \Else
    \State $O \gets$ \Call{getUnSatisfied}{$\CCC$}\label{alg:un:lgetunsat}
    \If{$|O|\leq b+k$}
    \State $\cY\gets \cY \; \cup \; $\Call{branchUnSatisfied}{$\CCC$, $b$, $O$, $0$}\label{alg:bu:lunsat}
    \EndIf
    \EndIf
    \State \Return $\cY$\label{alg:bu:lf}
    \EndFunction
  \end{algorithmic}
\end{algorithm}

\begin{algorithm}[htb]
  \caption{Method for branching on unsatisfied components.} \label{alg:branchunsat}
  \small
  \begin{algorithmic}[1]
    \INPUT the set $\CCC$ of remaining components, the
    remaining budget $b$, the set $O$ of remaining unsatisfied
    components, number $p$ of kept unsatisfied components
    \Function{\textsc{branchUnSatisfied}}{$\CCC$, $b$, $O$, $p$}
    \State $\cY\gets \emptyset$
    \If{$O \neq \emptyset$}
    \State $o \gets $ $O$.first()
    \If{$p<k$}
    \State $\cY \gets \cY \; \cup \;$\Call{branchUnSatisfied}{$\CCC$, $b$, $O\setminus \{o\}$,
      $p+1$}
    \EndIf
    \If{$\delta(o)\leq b$}\label{alg:unsat:lt}
    \State $\cY\gets \cY \; \cup \;$\Call{branchUnSatisfied}{$\CCC\setminus \{o\}$,
      $b-\delta(o)$, $O\setminus \{o\}$, $p$} 
    \EndIf
    \State \Return $\cY$
    \Else
    \State \Return $\{\delta(\CCC_0\setminus \CCC)\}$
    \EndIf
    \EndFunction
  \end{algorithmic}
\end{algorithm}

The following theorem now uses the set $W \subseteq V(G)$ computed in
Theorem~\ref{thm:our-shadow-cover} to compute a set of at most
$(|\EQ^{\neq 0}| \cdot k+k)^{|\EQ^{\neq 0}| \cdot k+k}$ conformal cuts
$\cY$ each of size at most $|\Gamma| \cdot k$ such that if $\mincost(S,I) \leq k$, then
there exists an assignment of cost at most $|\EQ^{\neq 0}| \cdot k+k$ to $S$
that agrees with one of the conformal cuts in $\cY$.
Note that Theorem~\ref{the:get-class-assignment} is now an immediate consequence
of Theorem~\ref{thm:list-conformal-cuts}, i.e. instead of returning
the set $\cY$ of conformal cuts, we choose one conformal cut $Y \in
\cY$ uniformly at random and output the class assignment corresponding
to $Y$.
The main idea
behind computing $\cY$ is to use the fact that we only need to
consider conformal cuts that are within $\delta(W)$ and this
allow us to branch on which components of $G[W]$ are
reachable from $s$ (see also property 4. in Theorem~\ref{thm:our-shadow-cover}).

\begin{theorem}
  \label{thm:list-conformal-cuts}
  Let $S$ be a simple instance of
  $\lin{2}{R}$, $I$ an ideal of $R$, $G=G_I(S)$, and $k$ an integer.
  There is a randomized algorithm that 
  takes $(G,k)$ 
  as input and returns 
  in $\bigohs((|\EQ^{\neq 0}| \cdot k+k)^{|\EQ^{\neq 0}| \cdot k})$ time
  a set $\cY$ of at most $(|\EQ^{\neq 0}| \cdot k+k)^{|\EQ^{\neq 0}| \cdot k+k}$ conformal cuts, 
  each of size at most $|\EQ^{\neq 0}| \cdot k$,
  such that with probability at least
  $2^{-O((|\EQ^{\neq 0}| \cdot k)^2)}$
  there is a cut $Y \in \cY$ with the following property:
  if $\mincost(S,I) \leq k$, then
  there exists an assignment of cost at most $|\EQ^{\neq 0}| \cdot k+k$ to $S$
  that agrees with $Y$.
\end{theorem}
\begin{proof}
  The algorithm starts by computing the set $W\subseteq V(G)\setminus
  \{s\}$ from $(G,k)$ using Theorem~\ref{thm:our-shadow-cover} such
  that with probability $2^{-O((|\EQ^{\neq 0}|k)^2)}$ there is a (minimal)
  $sA$-cut $Y'$ of size at most $|\EQ^{\neq 0}| \cdot k$ satisfying 1.--7..
  Given
  $W$ and the fact that $Y'$ has to satisfy 1.--7.,
  the algorithm now branches over all possibilities for $Y'$
  by
  guessing which components of $G[W]$ will be reachable from $s$ in
  $G-Y'$. In particular, using 2., and 4., it follows that
  $Y'=\delta(\CCC)$, where $\CCC$ is the set of all components of $G[W]$
  that are not reachable from $s$ in $G-Y'$. Moreover, because of 3.,
  $Y'$ is a conformal $sA$-cut in $G$. Therefore, the
  algorithm can correctly reject the choice of $W$ (and return $\cY=\emptyset$), if either $t \notin
  W$ or there is a variable $x$ such that $x_\alpha\notin W$ and
  $x_{\alpha'} \notin W$ for two distinct classes $\alpha, \alpha' \in \EQ^{\neq 0}$.
  So in the following we can assume that this is not the case, i.e.
  $t \in W$ and for every variable $x$ of $S$ at most one vertex
  $x_\alpha$ is not in $W$. We
  say that $x$ is \emph{$W$-decided} if there is such a class and
  otherwise we say that $x$ is \emph{$W$-undecided}. Let $\CCC_0$ be the set
  of all components of $G[W]$ and let $\CCC^-$ be the subset of $\CCC_0$
  consisting of all components $C$ such that either:
  \begin{itemize}
  \item $t \in C$,
  \item $C$ contains $x_\alpha$ and $x_{\alpha'}$ for some variable
    $x$ of $S$ and two distinct classes $\alpha$ and $\alpha'$, or
  \item $C$ contains $x_\alpha$ for some $W$-decided variable $x$ of $S$.
  \end{itemize}
  Because $Y'$ has to be a conformal $sA$-cut, no
  component in $\CCC^-$ can be reachable from $s$ in $G-Y'$ so $\delta(\CCC^-)\subseteq Y'$, where $\delta(\CCC^-)=\bigcup_{C \in
    \CCC^-}\delta(C)$. Consequently, if $|\delta(\CCC^-)|>|\EQ^{\neq 0}|k$, then
  the algorithm can correctly reject the choice of $W$ by returning
  $\cY=\emptyset$. Otherwise,
  every component $C$ in $\CCC_1=\CCC_0\setminus \CCC^-$ satisfies:
  \begin{itemize}
  \item $t \notin C$,
  \item $C$ does not contain $x_\alpha$ for any $W$-decided variable $x$
    of $S$,
  \item $C$ contains at most one vertex $x_\alpha$ for every $W$-undecided
    variable $x$ of $S$.
  \end{itemize}
  Next, the algorithm branches on which components in
  $\CCC_1$ are separated from $s$ in a solution. This
  part of the algorithm, which is given in Algorithm~\ref{alg:branch},
  returns the set $\cY$ of conformal cuts and
  is initially called with \textsc{branch}($\CCC_1$,
  $|\EQ^{\neq 0}|k-|\delta(\CCC^-)|$). The algorithm additionally employs the
  following notation. For a set $\CCC$ of components of $G[W]$ and
  a variable $x$ of $S$, we let $\kappa(\CCC,x)$ denote the subset of $\CCC$
  containing all components that contain a
  vertex $x_\alpha$ for some $\alpha \in \EQ^{\neq 0}$.
  The next aim of Algorithm~\ref{alg:branch} is to ensure that $Y'$ is
  conformal. For this it remains to ensure that for every
  $W$-undecided variable it holds that at most one component in
  $\kappa(\CCC,u)$ is reachable from $s$ in $G-Y'$.
  To achieve this the algorithm starts by computing the set $U$ of $W$-undecided
  variables in Line~\ref{alg:branch-lu} of Algorithm~\ref{alg:branch}.
  The algorithm then considers every $W$-undecided variable $u \in U$ with
  $\kappa(\CCC,u)\neq \emptyset$ and branches into
  $|\kappa(\CCC,u)|+1$ cases, i.e. the case that no component in
  $\kappa(\CCC,u)$ is reachable from $s$ and therefore
  $\delta(\kappa(\CCC,u))\subseteq Y'$ and for every component $C \in
  \kappa(\CCC,u)$ the case that only $C$ is reachable from $s$ and
  therefore $\delta(\kappa(\CCC,u)\setminus \{C\})\subseteq Y')$.
  This
  is achieved in Lines~\ref{alg:bu:lubb}--~\ref{algbu:lbe} in
  Algorithm~\ref{alg:branchundecided}. Note that any of these cases is
  only considered if the current budget
  $b$ does allow for the corresponding edges to be added to the
  current cut and otherwise the
  algorithm correctly returns $\cY=\emptyset$ in Line~\ref{alg:bu:lf}
  of Algorithm~\ref{alg:branchundecided}.

  Once it is ensured that $Y'$ is conformal and the class assignment is fixed for every variable, the algorithm
  still has to ensure item 3., i.e. that the instance
  $S-(\comp{Z}\cup\con(Y'))$ is satisfiable.
  Here and in the
  following we say that a component $C$ in $\CCC_1$ is \emph{self-satisfiable} if the
  equations of $S$ between variables corresponding to vertices in $C$
  can be satisfied using an assignment that assigns every variable $x$
  with vertex $x_\alpha$ in $C$ to a value in $\alpha$; note that
  because $C \in \CCC_1$ every variable $x$ with a vertex $x_\alpha$
  in $C$ has exactly one such vertex $x_\alpha$ in $C$ and therefore
  $\alpha$ is uniquely defined.
  To ensure that the instance $S-(\comp{Z}\cup\con(Y'))$ is satisfiable
  it is necessary to ensure that all but at most $k$ components that are
  currently reachable from $s$ are self-satisfiable; this is
  because $|\comp{Z}|\leq k$ and $\comp{Z}$ has to contain at least
  one constraint from every component that is not self-satisfiable.
  This is achieved in Algorithm~\ref{alg:branchunsat},
  which is initially called in Line~\ref{alg:bu:lunsat} of
  Algorithm~\ref{alg:branchundecided} as
  \Call{branchUnSatisfied}{$\CCC$, $b$, $O$, $0$}, where $\CCC$ is the
  current set of remaining components, $b$ is the remaining budget,
  and $O$ is the set of all components in $\CCC\subseteq
  \CCC_1$ that are not self-satisfiable. Note that the set $O$ of not self-satisfiable components $\CCC\subseteq
  \CCC_1$ is computed by a call to \Call{getUnSatisfied}{$\CCC$} in
  Line~\ref{alg:un:lgetunsat} of Algorithm~\ref{alg:branchundecided}
  and the function \Call{getUnSatisfied}{$\CCC$} can be implemented to
  run in polynomial-time as follows. For every component $C$
  in $\CCC$ \Call{getUnSatisfied}{$\CCC$}, check whether $C$ is
  self-satisfiable as follows. Let $V_C$ be the set of all
  variables $x$ of $S$ for which $C$ contains a vertex $x_\alpha$ for
  some $\alpha \in \EQ^{\neq 0}$.
  To decide whether $C$ is self-satisfiable, we will
  construct the $\csp{E_R \cup \EQ^{\neq 0}}$ instance
  $\compI{S}{C}$ over the variables in $V_C$ as
  follows.
  $\compI{S}{C}$ contains every equation of $S$ over the
  variables in $V_C$ as a constraint in $E_R$. Additionally,
  $\compI{S}{C}$ contains the unary constraint with scope $x$ and relation
  $\alpha$ for every $x \in V_C$ with $x_\alpha$ in $C$. Note that
  since every equivalence class $\alpha \in \EQ^{\neq 0}$ is a coset of $R$,
  it follows from Proposition~\ref{prop:in-ideal} that $\compI{S}{C}$ can be
  solved in polynomial-time. This completes the description of the
  algorithm. Because of the tests in Lines~\ref{algbu:lbb}
  and~\ref{alg:bu:tb} of Algorithm~\ref{alg:branchundecided} as well as in Line~\ref{alg:unsat:lt}
  of Algorithm~\ref{alg:branchunsat}, every edge set in $\cY$ returned by
  the algorithm has size at most the initial budget of $|\EQ^{\neq 0}|k$.
    
  We now analyse the running time of the algorithm leaving aside
  polynomial factors. Computing the set $W$ using
  Theorem~\ref{thm:our-shadow-cover} takes time
  $\bigohs(4^{|\EQ^{\neq 0}|k})$. Moreover, for every fixed $W$ the running
  time of the algorithm is dominated by the running time of
  Algorithm~\ref{alg:branch}, which can be obtained as follows.
  First branching on a $W$-undecided
  variable $u$ in Lines~\ref{algbu:lbb}--\ref{algbu:lbe} of Algorithm~\ref{alg:branchundecided} leads to
  at most $|\EQ^{\neq 0}|+1$ branches, where in every branch the initial budget
  of $|\EQ^{\neq 0}|k$ is decreased by at least one, which gives a running
  time of $\bigohs((|\EQ^{\neq 0}|+1)^{|\EQ^{\neq 0}|k})$. Moreover, dealing with
  the unsatisfied components in
  Algorithm~\ref{alg:branchunsat} requires to choose at most $k$
  unsatisfied components out of at most $|\EQ^{\neq 0}|k+k$ unsatisfied
  components, which gives a running time of
  $\bigohs((|\EQ^{\neq 0}|k+k)^{k})$. Therefore, the total running time of the
  entire algorithm is at most $$\bigohs(\max \{(|\EQ^{\neq 0}|+1)^{|\EQ^{\neq 0}| \cdot k},(|\EQ^{\neq 0}| \cdot k+k)^{k}\})=\bigoh(|\EQ^{\neq 0}| \cdot k+k)^{|\EQ^{\neq 0}| \cdot k}).$$
  Note that the above analysis also shows that $$|\cY|\leq \max
  \{(|\EQ^{\neq 0}|+1)^{|\EQ^{\neq 0}| \cdot k},(|\EQ^{\neq 0}| \cdot k+k)^{k}\}\leq (|\EQ^{\neq 0}| \cdot k+k)^{|\EQ^{\neq 0}| \cdot k+k}.$$

  It remains to show that if $(S,k)$ is a yes-instance, then with
  probability at least $2^{-O((|\EQ^{\neq 0}| \cdot k)^2)}$ there is a cut
  $Y \in \cY$ and an assignment of cost at most $|\EQ^{\neq 0}| \cdot k+k$ for $S$ that
  agrees with $Y$.
  Let $Z$ be a solution for $(S,k)$ and let $\varphi : V(S) \rightarrow I$ be a satisfying
  assignment for $S-Z$. Moreover, let $G=G_I(S)$, $Y=\sep(\ed(Z))$, $A$ be
  the set of vertices reachable from $s$ in $G-Y$, and let $W\subseteq V(G)\setminus \{s\}$ be the set
  obtained from $(G,k)$ using Theorem~\ref{thm:our-shadow-cover}. Then, with probability at least
  $2^{-\bigoh((|\EQ^{\neq 0}| \cdot k)^2}$, there is a minimal $sA$-cut $Y'$ of size at
  most $|\EQ^{\neq 0}| \cdot k$ satisfying 1.-7.. We will show that the set $\cY$
  returned by the algorithm contains a subset of $Y'$ satisfying 1.-7., which concludes the theorem.
  Let $\CCC_{Y'}$ be the set of components
  of $G[W]$ that are reachable from $s$ in $G-Y'$. Because $Y'$
  satisfies 1., 2. and 4., it follows that $Y'=\delta(\CCC_0\setminus
  \CCC_{Y'})$. Moreover, because $Y'$ satisfies 5., 6., and 7., it holds
  that $\CCC_{Y'}\subseteq \CCC_1$. We claim that the branching algorithm
  has a leaf with remaining components $\CCC_L$ such that $\CCC_{Y'}
  \subseteq \CCC_L$ and for every component $C$ in $\CCC_L\setminus
  \CCC_{Y'}$, it holds that $\compI{S}{C}$ has a solution. 
  This can be seen as follows. First the branching algorithm starts with
  the set $\CCC_1$ of remaining components and we already know that
  $\CCC_{Y'} \subseteq \CCC_1$. The algorithm then branches
  exhaustively on any $W$-undecided variable $x$ to decide which (if any)
  component in $\kappa(\CCC_1,x)$ will be reachable from $s$. Because
  $Y'$ satisfies 3. and is therefore conformal and $|Y'|\leq |\EQ^{\neq 0}|k$, at the end of that branching
  process there must be a node in the branching tree with remaining
  components $\CCC_N$ such that $\CCC_{Y'}\subseteq \CCC_N$ and every
  $W$-undecided variable has been branched on. The algorithm now proceeds
  with computing the set $O$ of \emph{unsatisfied components} in $\CCC_N$,
  i.e. a component $C$ is unsatisfied if $\compI{S}{C}$ has no
  solution.
  Note that $\CCC_{Y'}$ can contain at most $k$
  unsatisfied components because $Y'$ satisfies 3. and $\comp{Z}$ has to contain at least
  one equation from $\compI{S}{C}$ for every unsatisfied component $C
  \in \CCC_{Y'}$. Since the algorithm branches exhaustively on every
  possible choice that results in at most $k$ remaining unsatisfied components, there
  must a leaf of the branching algorithm with remaining components
  $\CCC_L$ such that $\CCC_L$ contains exactly the same unsatisfied
  components as $\CCC_{Y'}$.
  
  Therefore, $\CCC_{Y'}
  \subseteq \CCC_L$ and for every component $C$ in $\CCC_L\setminus
  \CCC_{Y'}$, it holds that $\compI{S}{C}$ has solution. Because
  $Y'$ satisfies 3., there is an assignment $\varphi$ of $S$ that
  satisfies $S-(\comp{Z}\cup \con(Y'))$ and agrees with $Y'$.
  Consider
  a component $C \in \CCC_L\setminus\CCC_{Y'}$. Because $\compI{S}{C}$
  has a solution, we can extend $\varphi$ into a satisfying assignment of
  $S-(\comp{Z}\cup \con(Y'\setminus \delta(C)))$ that agrees with
  $Y'\setminus \delta(C)$. Hence, after doing this for every
  component $C \in \CCC_L\setminus\CCC_{Y'}$, we obtain an assignment
  $\varphi'$ that satisfies $S-(\comp{Z}\cup \con(Y''))$ and that agrees
  with $Y''$, where $Y''=\delta(\CCC_0\setminus \CCC_L)$. Since
  $Y''\in \cY$ and $\varphi'$ satisfies all constraints of $S$ apart from the at most
  $|\EQ^{\neq 0}| \cdot k+k$ constraints in $(\comp{Z}\cup \con(Y'')$, this
  completes the proof of the theorem.
\end{proof}

\section{Lower Bounds}
\label{sec:lowerbounds}

Let
$G$ denote an arbitrary Abelian group.
An expression $x_1+\dots+x_r=c$ is an {\em equation over} $G$ if
$c \in D$ and $x_1,\dots,x_r$ are either variables or inverted variables with domain $D$.
We say that it is an {\em $r$-variable equation} if it contains at most $r$
distinct variables. We sometimes consider the natural group-based variants of
the $\lin{3}{R}$ and $\minlin{r}{R}$ problems in the sequel.
We prove the following results in this section.

\begin{enumerate}
\item
$\minlin{3}{G}$ is \W{1}-hard to FPT-approximate within any constant when $G$ is a
non-trivial Abelian group. This directly implies that $\minlin{3}{R}$ is \W{1}-hard to FPT-approximate within any constant when $R$ is a
non-trivial ring.

\item
$\minlin{2}{R}$ is \W{1}-hard to FPT-approximate within any constant if the finite and commutative ring
$R$ is non-Helly.

\item
$\minlin{2}{R}$ is \ETH-hard to FPT-approximate within $2-\eps$, $\eps > 0$, if
the finite and commutative ring $R$ is not lineal. 

\end{enumerate}

We briefly discuss some of the results below.
The hardness result for $\minlin{3}{G}$ is related to several results
in the literature. Theorem 6.1 in \cite{Dabrowski:etal:soda2023} shows that
$\minlin{3}{G}$ is \W{1}-hard to solve exactly, so our result is a direct strengthening
of this result. Furthermore,
it is known that $\minlin{3}{G}$ is not constant-factor approximable in polynomial time
unless P = \NP.
Håstad~\cite{haastad2001some} proves that for every $\eps > 0$, it is \NP-hard to distinguish instances of
$\lin{3}{G}$
(where $G$ is a non-trivial Abelian group), that are $(1/|G| + \eps)$-satisfiable from those that
are $(1-\eps)$-satisfiable.
Theorem 3.3(1) in \cite{dalmau2013robust} combined with Håstad's result shows that $\minlin{3}{G}$ is not robustly approximable with any loss, and finally
Dalmau, Krokhin, and Manokaran~\cite{Dalmau:etal:soda2015} point out that robust approximation of CSP$(\Gamma)$ with linear loss is equivalent to constant-factor polynomial-time approximation of $\mincsp{\Gamma}$. A direct proof based on contradicting Håstad's result goes like this:
Assume that $\minlin{3}{G}$ is $c$-approximable in polynomial time for some constant $c$. 
Pick $\eps > 0$ such that $(1-c \cdot \eps) \geq (1/|G| + \eps)$. Take an instance $S$ of $\lin{3}{G}$, set $k = \eps \cdot |S|$, run the factor-$c$ approximation algorithm on $S$, and deduce
the correct answer from the output of the algorithm. Our hardness result shows that
$\minlin{3}{G}$ is not constant-factor approximable even if we are allowed to use FPT
time instead of polynomial time (under the assumption \FPT $\neq$ \W{1}).

The hardness for \textsc{Min-2-Lin$(R)$} for non-Helly rings $R$
follows by reduction from $\minthreelin$. Here, the surprise
is that \textsc{Min-2-Lin}, which corresponds to a \textsc{MinCSP}
over a restricted language $\Gamma$ of binary constraints, is
capable of expressing constraints corresponding to non-binary linear
equations.
Let us expand on that a bit and provide some background. It is known
that the existence of constant-factor (classical or FPT) approximations
for \textsc{MinCSP$(\Gamma)$} depends on the polymorphisms of $\Gamma$~\cite{bonnet2016fixed,Barto:etal:polymorphisms}.
One of the most basic examples is a \emph{majority polymorphism}, which
implies that constraints definable over $\Gamma$ can be expressed as conjunctions
of binary constraints. Examples of problems \textsc{MinCSP$(\Gamma)$} where
$\Gamma$ has a majority polymorphism include \textsc{Almost 2-SAT},
\textsc{Unique Label Cover} over a fixed domain, and
\textsc{Min-2-Lin$(\FF)$} for a finite field $\FF$; in all these
examples, the parameterized problem is even FPT~\cite{bonnet2016fixed,Dabrowski:etal:soda2023,chitnis2016designing}.
Conversely, the most common sources of hardness for constant-factor FPT-inapproximability
for \textsc{MinCSP$(\Gamma)$} -- namely, \textsc{Hitting Set}
and \textsc{Maximum Likelihood Decoding} (MLD) -- correspond to languages
with much weaker polymorphisms; in fact, they are the canonical
negative examples for the larger class of languages of \emph{bounded linear width}~\cite[Section~4.4]{Barto:etal:polymorphisms}.
Given that \textsc{2-Lin$(R)$}-constraints appear much less powerful
than \textsc{3-Lin$(R)$}-constraints, and given the example languages above, 
it is an eye-opener that there are finite rings $R$ such that
the constraints of \textsc{2-Lin$(R)$} fail to have majority (or even
bounded linear width) and \textsc{Min-2-Lin$(R)$} is MLD-hard.
We note in passing that in every case where a
dichotomy for FPT-approximations is known (which, admittedly,
is just two cases), the defining line for constant-factor
FPT-approximations is precisely the inability to express
\textsc{Hitting Set} or MLD~\cite{bonnet2016fixed,osipov2023parameterized}.

For the ETH-based lower bound for non-lineal rings,
we first prove that the {\sc Paired Min Cut} problem is not FPT-approximable within $2-\eps$, $\eps > 0$,
under the \ETH. 
\textsc{Paired Min Cut} is commonly used as a 
source problem for proving 
\W{1}-hardness (see e.g.~\cite{Dabrowski:etal:ipec2023,KimKPW21flow,marx2009constant,Osipov:Wahlstrom:esa2023}).
Our bound is optimal since it is easy to approximate {\sc Paired Min Cut}
within 2 even in polynomial time.
The result for lineal rings follows by
combining the {\sc Paired Min Cut} result with a versatile gadget that has been used, for instance, when proving Theorem 27 in~\cite{dabrowski2023parameterized}
and Theorem 6.2 in~\cite{Dabrowski:etal:soda2023}.
Our bound is best possible since there are finite, commutative, and non-lineal
rings $R$ such that \minlin{2}{R} is FPT-approximable within 2. 
Let ${\mathbb F},{\mathbb F}'$ denote two non-trivial finite fields. 
The ring ${\mathbb F} \oplus {\mathbb F}'$ is not local and thus not
lineal by Proposition~\ref{prop:lineal-is-local}.
It is known that
\minlin{2}{{\mathbb F}} and
\minlin{2}{{\mathbb F}'}
are in \FPT~\cite{Dabrowski:etal:soda2023}. These two problems are consequently FPT-approximable within 1 and
Proposition~\ref{prop:sumapproximation}(2) implies that \minlin{2}{{\mathbb F} \oplus {\mathbb F}'}
is FPT-approximable within 2. Clearly, we also have an optimal
computational bound under \ETH: if $R$ is the direct sum of two finite fields, then
\minlin{2}{R} is FPT-approximable within 2 but it cannot be FPT-approximated
within $2-\eps$, $\eps > 0$, if the \ETH holds.

\subsection{Inapproximability Result for \minthreelin}

Let $G$ be a finite Abelian group and
consider the following problem.

\pbDefGap{$\GAP{\gamma}$-$\minlin{3}{G}$}
{A system $S$ of equations of the form $x_1 + x_2 + x_3 = a$,
where $x_1,x_2,x_3$ are variables, $a \in G$, and an integer $k$.}
{$k.$}
{There exists $Z \subseteq S$, $|Z| \leq k$
such that $S - Z$ is consistent, or}
{For all $Z \subseteq S$ with $|Z| \leq \gamma \cdot k$, $S - Z$ is inconsistent.}

We show that this problem is \W{1}-hard
for all $\gamma \geq 1$ and non-trivial Abelian groups $G$.
This readily implies that $\GAP{\gamma}$-$\minlin{3}{R}$ is \W{1}-hard
  for every $\gamma \geq 1$ and every 
  non-trivial finite ring $(R,+,\cdot)$ since $(R,+)$ is an Abelian group.
The starting point for our reduction is the following problem
and the corresponding hardness result.
For a vector $x \in \FF^{n}$ in a field $\FF$,
the \emph{Hamming norm} $||x||_0$ is 
the number of nonzero components in $x$.

\pbDefGap{Gap$_{\gamma}$ Maximum Likelihood Decoding Problem over $\FF_p$ ($\GAPMLD$)}
{A matrix $A \in \FF^{n \times m}_p$, a vector $y \in \FF^n_p$, and
and an integer $k \in {\mathbb N}$.}
{$k$.}
{There exists $x \in \{0, 1\}^m$
with $||x||_0 \leq k$ such that $Ax = y$, or}
{For all $x \in \FF_p^m$ with $||x||_0 \leq \gamma k$, 
$Ax \neq y$.}

\begin{theorem}[Theorem~5.1~in~\cite{bhattacharyya2021parameterized}]
  \label{thm:gapmld-is-hard}
  $\GAPMLD$ is \W{1}-hard for every 
  $\gamma \geq 1$ and every prime $p$. 
\end{theorem} 

Consider an instance $(A, y, k)$ of $\GAPMLD$.
Denote the entry in row $i$ and column $j$ of $A$ by $a_{ij}$.
The condition $Ax = y$ holds if and only if
\begin{equation} \label{eq:row-equation}
  \sum_{j=1}^{n} a_{ij} x_{j} = y_{i}.
\end{equation}
for every row $i$.
This equation is linear, so we have a simple reduction to
$\GAP{\gamma}$-$\minlin{*}{\GG}$ over $\ZZ_p$ 
without restriction on the number of variables per equation:
make an instance with $x_1, \dots, x_n$ as variables,
add row equations as crisp equations,
and add soft equations $x_j = 0$ for all $j \in [n]$.
Then every assignment of cost at most $\gamma k$
can only choose vectors $x$ with at most $\gamma k$ nonzero entries.
To show that $\GAP{\gamma}$-$\minlin{3}{\GG}$ is hard, 
we need to implement row equations using group equations,
and transfer the hardness from cyclic groups of prime order
to all finite groups. 

We start by proving hardness of $\GAP{\gamma}$-$\minlin{3}{\ZZ_p}$.
For further convenience, we define a restricted class of instances 
and prove hardness for this case.

\begin{definition}
  An instance of $\minlin{3}{\GG}$
  is \emph{almost homogeneous} if every equation in it is of the form
  \begin{itemize}
    \item $x + y + z = 0$ and crisp, or
    \item $x = a$ for some $a \in G$ and crisp, or
    \item $x = 0$.
  \end{itemize}
\end{definition}

\begin{lemma} \label{lem:min3lin-special}
  $\GAP{\gamma}$-$\minlin{3}{\ZZ_p}$
  is \W{1}-hard for all $\gamma \geq 1$ and primes $p$
  even when restricted to almost homogeneous instances.
\end{lemma}
\begin{proof}
Let $(A, y, k)$ be an instance of $\GAPMLD$.
As usual, we let $a_{ij}$ denote the entry in row $i$ and column $j$ of $A$.
Construct an almost homogeneous instance of $(S, k)$ of 
{$\GAP{\gamma}$-$\minlin{3}{\ZZ_p}$} as follows.
Start by introducing primary variables $x_1, \dots, x_n$ to $V(S)$,
and adding soft constraints $x_i = 0$ for all $i \in [n]$.
For every $a \in \range{0}{p-1}$,
introduce auxiliary variables $z_a$ and
add crisp equations $z_a = a$ to $S$.
Consider an arbitrary row equation $\sum_{j=1}^{n} a_{ij} x_{j} = y_{i}$.
Since we are working over the field $\ZZ_p$,
this equation can be written as
\begin{equation} \label{eq:row-equation-additive}
  \underbrace{x_{1} + \dots + x_{1}}_{a_{i1} \textrm{ times}} 
  + \dots +
  \underbrace{x_{n} + \dots + x_{n}}_{a_{in} \textrm{ times}}
   - y_{i} = 0 \mod p.
\end{equation}
Let $L = \sum_{j=1}^{n} a_{ij}$ and
define $s_1, \dots, s_L$,
where $s_1 = \dots = s_{a_{i1}} = x_1$,
$s_{a_{i1}+1} = \dots = s_{a_{i2}} = x_2$,
and so on.
Note that $\sum_{\ell=1}^{L} s_\ell = \sum_{j=1}^{n} a_{ij} x_j$.
To express~\eqref{eq:row-equation-additive} 
as a system of homogeneous $3$-variable equations,
we introduce auxiliary variables
$u_1, \bar{u}_1 \dots, u_L, \bar{u}_L$, and
add crisp equations
$u_\ell + \bar{u}_{\ell} + z_{0} = 0$ for all $\ell$.
Note that these constraints force $u_\ell$ and $\bar{u}_{\ell}$ 
to be additive inverses modulo $p$.
Add crisp equations
\begin{equation} \label{eq:chop}
  \begin{aligned}
    z_0 + s_1 + u_1 &= 0 \mod p, \\
    \bar{u}_{\ell-1} + s_{\ell} + u_{\ell} &= 0 \mod p \quad \text{ for } \ell \in \range{2}{L-1}, \\
    \bar{u}_{L-1} + s_{L} + z_{-y_i} &= 0 \mod p.
  \end{aligned}
\end{equation}
Observe that the sum of the equations above, after cancellation of auxiliary variables,
implies $\sum_{\ell=1}^{L} s_\ell + z_{-y_i} = \sum_{i=1}^{n} a_{ij} x_i - y_i = 0$.
Applying the same reduction to all rows of $A$,
we construct $(S, k)$ in polynomial time.

For correctness, first assume that $(A, y, k)$ is a yes-instance,
i.e. there exist $x_1, \dots, x_n \in \ZZ_p$
that satisfy row equations~\eqref{eq:row-equation} for all $i \in [n]$,
and at most $k$ values in $(x_1, \dots, x_n)$ are nonzero.
Observe that it can be extended to auxiliary variables $u_i$, $\bar{u}_i$
in a unique way.
The obtained assignment satisfies all crisp and all but $k$ soft equations $x_i = 0$.

For the other direction, suppose $(S, k)$
admits an assignment $\alpha$ that violates at most $\gamma k$ constraints.
Define vector $\alpha(x) = (\alpha(x_1), \dots, \alpha(x_n))$.
Since all soft constraints in $S$ are of the form $x_i = 0$, 
we have $||\alpha(x)||_0 \leq \gamma k$.
By construction of $S$, we have $A \alpha(x) = y$.
Hence, if $(A, y, k)$ is a no-instance, then
$(S, k)$ is a no-instance.
\end{proof}

Before proving inapproximability for group equations,
we recall some basic facts.
The \emph{direct sum} of two groups
$\GG_1 = (G_1, +_1)$ and $\GG_2 = (G_2, +_2)$ 
is a group $\GG_1 \oplus \GG_2$ with pairs 
$\{ (g_1, g_2) \mid g_1 \in G_1, g_2 \in G_2 \}$ as elements,
and with group operation applied componentwise,
i.e. $(g_1, g_2) + (g'_1, g'_2) = (g_1 +_1 g'_1, g_2 +_2 g'_2)$.
The fundamental theorem of finite Abelian groups (Theorem~\ref{thm:sumoflocalrings}) asserts the following.

\begin{theorem} \label{thm:fundamental-abelian-group}
Every finite Abelian group $G$ decomposes into
a direct sum of cyclic groups of prime power order,
i.e. $G \cong \bigoplus_{i=1}^{t} \ZZ / p_i^{\ell_i}\ZZ$
for primes $p_1,\dots,p_t$ and positive integers $\ell_1,\dots,\ell_t$.
\end{theorem}

\begin{theorem} \label{thm:gap-min3lin(group)-is-hard}
  $\GAP{\gamma}$-$\minlin{3}{G}$ is \W{1}-hard
  for every $\gamma \geq 1$ and every 
  non-trivial finite Abelian group $G$.
\end{theorem}
\begin{proof}
Let $G \cong \bigoplus_{i=1}^{t} \ZZ_{p_i^{\ell_i}}$ and set
$p = p_1$ and $\ell = \ell_1$.
Consider an almost homogeneous instance $(S, k)$ of $\minlin{3}{\ZZ_p}$.
Create an instance $(S', k)$ of $\minlin{3}{G}$ as follows.
Let $V(S')$ contain variable $v', v''$ for every $v \in V(S)$.
Add crisp equations equivalent to
\begin{equation} \label{eq:multiplicity}
  \underbrace{v' + \dots + v'}_{p^{\ell-1}} = v''    
\end{equation}
using the same reduction from many summands to three as
in Equation~\eqref{eq:chop} in the proof of Lemma~\ref{lem:min3lin-special}.
For every equation of the form $x + y + z = 0$ in $S$,
add $x'' + y'' + z'' = 0$ to $S'$.
For every equation of the form $x = a$ in $S$,
add $x'' = (p^{\ell-1} a, 0, \dots, 0)$ in $S'$.
This completes the construction.

For one direction,
consider an arbitrary assignment $\alpha : V(S) \to \ZZ_p$
and define $\alpha' : V(S') \to \GG$
by setting $\alpha(v')$ equal to the element
$(\alpha(v), 0, \dots, 0)$
and $\alpha(v'')$ to the element 
$(p^{\ell-1} \alpha(v), 0, \dots, 0)$.
Note that $\alpha'$ satisfies equations~\ref{eq:multiplicity}
for all $v', v'' \in V(S')$.
Furthermore, if
$\alpha(x) + \alpha(y) + \alpha(z) = 0 \bmod p$, then
$p^{\ell-1} (\alpha(x) + \alpha(y) + \alpha(z)) = 0 \bmod p^\ell$
and
$\alpha'(x'') + \alpha'(y'') + \alpha'(z'') = 
(0, \dots, 0) = 0$.
Finally, if $\alpha(x) = a$, then
$\alpha'(x'') = (p^{\ell-1} a, 0, \dots, 0)$.
We conclude that $\alpha'$ violates
at most as many equations as $\alpha$.

For the other direction,
consider an arbitrary assignment $\beta' : V(S') \to \GG$
and define $\beta : V(S) \to \ZZ_p$
by setting $\beta(v)$ to the projection
of $\beta'(v')$ onto the first component.
By the crisp equations~\eqref{eq:multiplicity}, 
the first component of $\beta(v'')$ equals $p^{\ell-1} \beta(v)$.
Thus, if $\beta(x'') = (p^{\ell-1} a, 0, \dots, 0)$, then $\beta(x) = a$.
Moreover, if $\beta(x'') + \beta(y'') + \beta(z'') = 0$,
then the equation holds in the first component,
hence $p^{\ell-1} (\beta(x) + \beta(y) + \beta(z)) = 0 \bmod p^\ell$
and $\beta(x) + \beta(y) + \beta(z) \bmod p$.
Thus, $\beta'$ violates at most as many equations as $\beta$.
\end{proof}

We obtain the following corollary
due to the additive group of every ring being Abelian.

\begin{corollary} \label{cor:gap-min3lin(ring)-is-hard}
  $\GAP{\gamma}$-$\minlin{3}{R}$ is \W{1}-hard
  for every $\gamma \geq 1$ and every 
  non-trivial finite ring $R$.
\end{corollary}

\subsection{Inapproximability Result for Non-Helly Rings}
\label{sec:inapprox-non-helly}

Let $R$ be a finite commutative ring.
Recall from Section~\ref{sec:annihilators} that 
a triple $C_1, C_2, C_3$ of one-element annihilator cosets of $R$
forms a tangle if $C_i \cap C_j \neq \emptyset$
for all $i,j \in \{1,2,3\}$,
and $C_1 \cap C_2 \cap C_3 = \emptyset$.
We say that $R$ is Helly if it does not admit any tangle.
We will show that if $R$ is not Helly,
then $\minlin{2}{R}$ cannot be FPT-approximated
within {\em any} constant unless \FPT = \W{1}. 
The idea is to use the \W{1}-hardness result for 
$\GAP{\gamma}$-$\minlin{3}{R}$
(Corollary~\ref{cor:gap-min3lin(ring)-is-hard}) and express
equations of length 3 using equations of length 2.
We assume without of generality (due to Proposition~\ref{prop:sumapproximation}(1)) that $R$ is local.
The reduction is from $\minlin{3}{\FF}$ where $\FF$
is the {\em residue field} of $R$, i.e. $\FF=R/M$ where $M$ is
the maximal ideal in the local ring $R$.
To illustrate the idea, we consider the non-Helly ring
$R = \ZZ_2[\rx,\ry]/(\rx^2,\ry^2)$.
In this case, $\FF=\ZZ_2$.
For example, consider an equation
$a + b + c = 0$ over $\ZZ_2$.
To express it using binary equations over $R$,
we introduce a fresh variable $v$
and consider three equations
\begin{align*}
  \rx v &= \rx\ry b, \\
  \ry v &= \rx\ry a, \\
  (\rx+\ry) v &= - \rx\ry c.
\end{align*}
Summing up the first two equations,
we obtain $(\rx + \ry) v = \rx\ry (a + b)$.
Together with the third one,
this implies $\rx\ry (a + b + c) = 0$.
On the other hand,
any assignment that satisfies
$\rx\ry (a + b + c) = 0$
can be extended as $v = \rx a + \ry b$
to satisfies all three binary equations.
We formalize this idea in the following theorem.

\begin{theorem} \label{thm:non-helly-hard}
  Let $R$ be a finite and commutative ring that is not Helly.
  Then $\minlin{2}{R}$ cannot be FPT-approximated
  within any constant unless $\FPT=\W{1}$.
\end{theorem}
\begin{proof}
By Proposition~\ref{prop:sumapproximation},
we can assume without loss of generality that
$R$ is a local ring.
Let $M$ be the maximal ideal of $R$ and
$\FF = R / M$ be the residue field.
Define the natural projection $\pi : R \to \FF$
such that for arbitrary $r \in R$, $\pi(r) = u \iff r\in M + u$.

We apply Lemma~\ref{lem:homogenised-tangle} and let $a, b, c, d \in R$ denote the elements
such that
$A = \Ann(a)$, $B = \Ann(b)$ and $C = \Ann(c) + d$
form a coset tangle.
Further, let $s \in A \cap C$ and $t \in B \cap C$.
Note that $C = \Ann(c) + t$ and $a,b,c,d,s,t$ are non-units.

Consider an almost homogeneous instance $(S, k)$ of $\minlin{3}{\FF}$.
We create an instance $(S_R, k)$ of $\minlin{2}{R}$
with the same parameter value as follows.
Start by adding variable $x_R$ to $V(S_R)$ for every $x \in V(S)$.
For every (crisp) equation of the form $x = i$ in $S$,
add a (crisp) equation 
\begin{equation}
  \label{eq:const}
  x_R = i    
\end{equation} 
to $S_R$.
Note that equations of the form $x = 0$ become $x_R = 0$.
For every equation of the form 
$x + y + z = 0$ in $S$,
create a new variable $v_R$ in $V(S_R)$ and 
add the following crisp equations to $S_R$:
\begin{align}
  \label{eq:a} a \cdot v_R &= a t  \cdot y_R, \\
  \label{eq:b} b \cdot v_R &= b s  \cdot x_R, \\
  \label{eq:c} c \cdot v_R &= -c t \cdot z_R.
\end{align}
This completes the construction.

Towards correctness, first consider an assignment
$\alpha : V(S) \to \FF$ and define
$\alpha' : V(S_R) \to R$ as $\alpha'(x_R) = \alpha(x)$ 
for all $x \in V(R)$.
Clearly, if $\alpha$ satisfies an equation
of the form $x = i$, then $\alpha'$ satisfies
the corresponding constraint $x_R = i$ in $S'$.
Moreover, if $\alpha$ satisfies an equation
of the form $x + y + z = 0$ in $S$, 
we claim that $\alpha'$ extended with
$\alpha'(v_R) = s \alpha'(x_R) + t \alpha'(y_R)$
satisfies equations~(\ref{eq:a}-\ref{eq:b}) in $S'$.
Indeed, since $s \in \Ann(a)$ and $t \in \Ann(b)$, we have
\[
  a\alpha'(v_R) = a(s \alpha'(x_R) + t \alpha'(y_R)) = at \alpha'(y_R)
\]
and
\[
  b\alpha'(v_R) = a(s \alpha'(x_R) + t \alpha'(y_R)) = bs \alpha'(x_R).
\]
Furthermore, since $s - t \in \Ann(c)$ and
$\alpha'(x_R) + \alpha'(y_R) = -\alpha'(z_R)$
we have
\begin{align*}
  c\alpha'(v_R) &= c(s \alpha'(x_R) + t \alpha'(y_R)) = \\
  &= cs \alpha'(x_R) + ct \alpha'(y_R) + ct \alpha'(x_R) - ct \alpha'(x_R) = \\
  &= c(s-t) \alpha'(x_R) + ct(\alpha'(x_R) + \alpha'(y_R)) = \\
  &= ct(\alpha'(x_R) + \alpha'(y_R)) = -ct\alpha'(z_R).
\end{align*}

For the opposite direction, let 
$\beta : V(S_R) \to R$ be an assignment to $S_R$,
and consider the assignment 
$\beta' : V(S) \to \FF$ defined as
$\beta'(x) = \pi(\beta(x_R))$ for all $x \in V(S)$.
Suppose $\beta$ satisfies a unary equation $x_R = i$.
Then $\beta'$ satisfies the corresponding equation 
$x = i$ because $\pi$ preserves $\FF$.
Now, assume $\beta$ satisfies equations~(\ref{eq:a}-\ref{eq:c}),
and consider a ring element $w_R = v_R - sx_R - ty_R$.
We claim that $\beta$ satisfies
\begin{align}
  \label{eq:aa} w_R &\in \Ann(a), \\
  \label{eq:bb} w_R &\in \Ann(b), \\
  \label{eq:cc} w_R &\in \Ann(c) - t(x_R + y_R + z_R).
\end{align}
For equation~\eqref{eq:aa}, observe $w_R = (v_R - ty_R) - sx_R$,
$s \in \Ann(a)$ and $\beta$ satisfies $a(v_R - ty_R) = 0$,
hence $\beta(w_R) \in \Ann(a)$.
For equation~\eqref{eq:bb}, observe that 
$w_R = (v_R - sx_R) - ty_R$, $t \in \Ann(b)$
and $\beta$ satisfies $b(v_R - sx_R) = 0$ so $\beta(w_R) \in \Ann(b)$.
Towards equation~\ref{eq:cc}, observe that
\begin{align*}
  c(w_R + tx_R + ty_R) 
  &= c(v_R - sx_R - ty_R + tx_R + ty_R) = \\
  &= cv_R - c(s-t)x_R = cv_R.
\end{align*}
Since $\beta$ satisfies equation~\eqref{eq:c}, 
it also satisfies $c(w_R +tx_R + ty_R) = -ctz_R$,
which is equivalent to
$cw_R = ct(x_R + y_R + z_R)$, which in turn,
is equivalent to equation~\eqref{eq:cc}.
We claim that this implies that $\beta'$ satisfies 
$x + y + z = 0$,
or, equivalently, that 
$\beta(x_R) + \beta(y_R) + \beta(z_R)$ is a non-unit.
Assume towards contradiction that
$\beta(x_R) + \beta(y_R) + \beta(z_R) = u$ for some unit $u \in R$
and consider the element $-u^{-1} w_R$.
Clearly, 
$$-u^{-1} \beta(w_R) \in \Ann(a) \cap \Ann(b)$$
and 
$$-u^{-1} \beta(w_R) \in \Ann(c) + 
u^{-1} t (\beta(x_R) + \beta(y_R) + \beta(z_R)) = \Ann(c) + t.$$
This contradicts the fact that $A$, $B$, $C$ form a tangle.
\end{proof}

\subsection{Inapproximability Result for Non-Lineal Rings}

Our final hardness result show that \minlin{2}{R} is not FPT-approximable (under \ETH) within $2-\eps$ for any $\eps > 0$ if $R$ is a finite and commutative ring that is
not lineal. The proof is based on an auxiliary result that
proves inapproximability of the {\sc Split Paired Min Cut} problem.

\subsubsection{Inapproximability of \textsc{Paired Min Cut}}
\label{sssec:paired-min-cut-lb}

Consider the following problem.

\pbDefP{Paired Min Cut}
{A graph $G$, vertices $s, t \in V(G)$, 
a set of disjoint edge pairs $\cP \subseteq \binom{E(G)}{2}$,
and an integer $k$}
{$k$.}
{Is there an $st$-mincut $Z \subseteq E(G)$ 
which is the union of $k$ pairs in $\cP$?} 

The \W{1}-hardness of this problem has been shown in~\cite{marx2009constant}. Moreover, 
\textsc{Paired Min Cut} admits a simple $2$-approximation in polynomial time:
Given an instance $(G, s, t,\cP, k)$,
assign unit cost to every edge in $\bigcup \cP$ and infinite cost to every edge outside $\bigcup \cP$.
Compute the minimum cost $st$-maxflow in $G$; if it exceeds $2k$, then reject and,
otherwise, any minimum cost $st$-mincut intersects at most $2k$ pairs in $\cP$.
Surprisingly, we show that a better approximation is 
unlikely even if we allow the algorithm to run in FPT time in $k$, i.e.
no FPT algorithm can approximate \textsc{Paired Min Cut}
within a factor better than $2$, assuming \ETH. In fact, this holds even
for severely restricted instances.
We say that an instance $(G,s,t,\cP,k)$ of \textsc{Paired Min Cut} is
\emph{split} if $V(G)\setminus \{s,t\}$ has a partition into $U_1$ and
$U_2$ such that:
\begin{enumerate}
\item $G$ has no edge between a vertex in $U_1$ and a vertex in $U_2$,
  and
\item for every pair in $\cP$ one edge lies in $G_1=G
  \setminus U_2$ and the other lies in $G_2=G\setminus U_1$.
\end{enumerate}
We define the following computational problem.

\pbDefGap{$\GAP{\gamma}$-Split Paired Min Cut}
{An instance $(G,s,t,G_1,G_2,\cP,k)$ of \textsc{Split Paired Min Cut}.}
{$k$.}
{there is an $st$-mincut $Z \subseteq E(G)$ 
which is the union of $k$ pairs in $\cP$,}
{every $st$-cut $Z \subseteq E(G)$ intersects 
more than $\gamma k$ pairs in $\cP$.}

In the following we will exclusively work with 
the {\sc $\GAP{\gamma}$-Split Paired Min Cut}
problem.
We will show that, unless \ETH is false,
there is no FPT algorithm solving 
\textsc{$\GAP{2-\eps}$-Split Paired Min Cut}
for any $\eps > 0$.
The starting point of our reductions is a gap version of 2CSP.
This is the variant of CSP with binary constraints and no language restrictions 
(i.e. any relation may appear in the constraints).
In this setting it makes sense to include the domain $D$ as part of the input.

\pbDefGap{$\GAP{\eps}$-2CSP}
{An instance $\cI = (V,D,\cC)$ of 2CSP.}
{$|V|$.}
{there exists an assignment $V \to D$ that satisfies all constraints in $\cC$,}
{every assignment $V \to D$ satisfies at most an $\eps$-fraction of constraints in $\cC$.}

The \emph{Parameterized Inapproximability Hypothesis (PIH)}~\cite{lokshtanov2020parameterized} 
asserts that $\GAP{\eps}$-2CSP is \W{1}-hard for some $0 < \eps < 1$.
Combined with Raz' parallel repetition theorem~\cite{Raz:sicomp98},
PIH implies that $\GAP{\eps}$-2CSP is hard for every $0 < \eps < 1$.
A recent breakthrough result showed that PIH 
holds under \ETH{}~\cite{Guruswami:etal:stoc2024}.
It will be convenient for us to additionally assume
that the instances of 2CSP we reduce from are sparse,
i.e. the number of constraints is in $O(|V|)$.
This is not explicitly stated in~\cite{Guruswami:etal:stoc2024},
but can be obtained by applying e.g. 
\cite[Lemma~4.1]{dinur2007pcp} to sparsify any 2CSP instance
in an approximation-preserving way,
or taking the statement directly from~\cite{guruswami2024almost}.

\begin{theorem}[See e.g. Corollary~1.3~in~\cite{guruswami2024almost}\protect\footnotemark]
\footnotetext{Corollary~1.3 in~\cite{guruswami2024almost} is phrased as hardness for the search problem. Theorem~\ref{thm:gap-2csp} follows immediately by combining Theorem~4.1~of~\cite{guruswami2024almost} with parallel repetition in projection games~\protect\cite{Rao:sicomp2011}.}
  \label{thm:gap-2csp}
  Assuming \ETH,
  $\GAP{\eps}$-\textsc{2CSP} is not in \FPT for any $0 < \eps < 1$
  even when restricted to instances $(V,D,\cC)$ with $O(|V|)$ constraints.
\end{theorem}

The \emph{primal graph} of a 2CSP instance $\cI = (V,D,\cC)$ is
a simple undirected graph $G_{\cI} = (V, E)$ with vertices 
being identified with the variables of $\cI$ and 
an edge $\{x,y\} \in E$ if and only if $\cC$ contains a constraint $R(x,y)$.
As an intermediate step, we show that
the hardness result of Theorem~\ref{thm:gap-2csp} holds for
2CSP restricted to instances with a bipartite primal graph.

\begin{lemma} \label{lem:gap-2csp-bipartite}
  Assuming \ETH,
  $\GAP{\eps}$-\textsc{2CSP} is not in FPT for any $0 < \eps < 1$
  even when restricted to instances $(V,D,\cC)$ with $O(|V|)$ constraints and a bipartite primal graph.
\end{lemma}
\begin{proof}
  Let $\cI = (V, D, \cC)$ be an instance of 2CSP
  with $m = |\cC| = O(|V|)$.
  For a bipartition $(A, B)$ of $V$,
  we say that it \emph{retains} a constraint
  $R(x,y) \in \cC$ if $x \in A$ and $y \in B$.
  We claim that there is bipartition $(A,B)$ that retains
  at least half of the constraints.
  Indeed, a random bipartition retains a constraint
  with probability $\frac{1}{2}$, and by linearity of expectation,
  it retains $\frac{m}{2}$ constraints in total.
  Thus, there must exist a bipartition that 
  retains at least $\frac{m}{2}$ constraints,
  and it can be found in polynomial time
  using the standard method of conditional expectations
  (or in $2^k$ time by considering every bipartition).
  Let us fix such a bipartition, and let $\cI'$
  be the subinstance of $\cI$ consisting of
  the constraints retained by the bipartition.
  
  Clearly, the primal graph of $\cI'$ is bipartite.
  Moreover, if $\cI$ is satisfiable, then $\cI'$
  is satisfiable as well because it is a subinstance.
  On the other hand, assume $\cI'$ admits an assignment
  that satisfies an $\eps$-fraction of the constraints in $\cI'$
  for some $0 < \eps < 1$.
  Then the same assignment satisfies at least $\frac{\eps}{2}$-fraction of the constraints in $\cI$.
  Thus, if $\GAP{\eps}$-\textsc{2CSP} is in FPT
  restricted to bipartite instances, then
  $\GAP{\eps/2}$-\textsc{2CSP} is in FPT on general
  instances, contradicting the \ETH by Theorem~\ref{thm:gap-2csp}.
\end{proof}

We are now ready to prove the main result.

\begin{theorem} \label{thm:gap-paired-cut}
  Assuming \ETH, $\GAP{(2-\eps)}$-\textsc{Split Paired Min Cut}
  is not in \FPT for any $0 < \eps < 1$. 
\end{theorem}
\begin{proof}
  Let $\cI = (V, D, \cC)$ with $|\cC|=m$ be an instance of 2CSP with a bipartite primal graph and $m = O(|V|)$.
  We will construct an instance 
  $\cI' = (G, s, t, G_1, G_2, \cP, k')$ of 
  \textsc{Split Paired Min Cut}
  with $k' = m$ such that
  \begin{itemize}
    \item if $\cI$ is satisfiable,
    then $\cI'$ is a yes-instance, and
    \item if $G$ admits an $st$-cut
    that intersects at most $(2 - \eps)m$ pairs in $\cP$
    for any $0 < \eps < 1$,
    then $\cI$ admits an assignment that satisfies
    at least $\eps m$ constraints.
  \end{itemize}
  The hardness then follows by Lemma~\ref{lem:gap-2csp-bipartite}.
  
  Now we construct $\cI'$.
  Let $(A, B)$ be a bipartition of $V$, where
  $A = \{a_1, \dots, a_{|A|}\}$ and $B = \{b_1, \dots, b_{|B|}\}$.
  Further, let $D = \{1, \dots, n\}$.
  Construct a graph $G$ with vertices $s$, $t$ and
  vertices $v^{A}_{i,d}$, $v^{B}_{j,d}$,
  for every $i \in [|A|]$, $j \in [|B|]$ and
  $0 \leq d \leq n$.
  Identify $v^{X}_{i,0}$ with $s$ and $v^{X}_{i,n}$ with $t$
  for all $i$ and $X \in \{A,B\}$.
  Note that there are $n+1$ vertices corresponding 
  to each variable in $V$.
  Now we create the edges of $G$ (and add more auxiliary vertices).
  For every $d \in [n]$,
  we will connect $v^A_{i,d-1}$ to $v^A_{i,d}$
  with a set of internally disjoint paths of length $n$,
  where each path corresponds to 
  a constraint $R(a_i,b)$ for some $b \in B$,
  and the edges of the path correspond to
  the values that $b$ may take.
  We make a similar construction for $B$.
  Formally, let us enumerate 
  the constraints in $\cC$ as $C_1, \dots, C_m$.
  For every $\ell \in [m]$,
  let $C_{\ell} = R_{\ell}(a_{i}, b_{j})$ with $i \in [|A|]$ and $j \in [|B|]$, and
  \begin{itemize}
    \item create a path $P^{A,\ell}_{{i},d}$ connecting 
      $v^A_{{i},d-1}$ and $v^A_{{i},d}$ for every $d \in [n]$ 
      with $n-1$ freshly introduced auxiliary vertices and edges 
      $e^{A,\ell}_{({i},d),({j},d')}$
      for all $d' \in [n]$, 
    \item create a path $P^{B,\ell}_{j,d}$ connecting 
      $v^B_{j,d-1}$ and $v^B_{j,d}$ for every $d \in [n]$ 
      with $n-1$ freshly introduced auxiliary vertices and edges 
      $e^{B,\ell}_{(j,d),(i,d')}$
      for all $d' \in [n]$, and
    \item add the pair of edges 
      $\{e^{A,\ell}_{(i,d),(j,d')}, e^{B,\ell}_{(j,d'),(i,d)}\}$ to $\cP'$
      for all $(d,d') \in R_\ell$.
  \end{itemize}
  After that, contract every edge that does not occur in any of $\cP$.
  Now set $k' = m$.
  Finally, define $G_1$ as the union of 
  the paths $P_{i,d}^{A,\ell}$ for all $i, d, \ell$
  and $G_2$ as the union of the paths $P_{j,d}^{B,\ell}$
  for all $j, d, \ell$.
  Clearly, our instance is split.

  Towards correctness, we first show that if $\cI$ is satisfiable,
  then $\cI'$ is a yes-instance.
  Let $\psi : V \to D$ be an assignment satisfying all constraints in $\cC$.
  Define a set of edges $Z$ in $G$ as follows.
  For every constraint $C_\ell = R_\ell(a_i, b_j)$ in $\cC$,
  let $d = \psi(a_i)$ and $d' = \psi(b_j)$.
  Add both edges in the pair
  $\{e^{A,\ell}_{(i,d),(j,d')}$, $e^{B,\ell}_{(j,d'),(i,d)}\}$ to $Z$.
  Since $\psi$ satisfies $C_\ell$,
  we have $(d,d') \in R_\ell$, so $Z$ is a union of $m$ pairs in $\cP$.
  It remains to show that $Z$ is an $st$-cut.
  We claim that for every $i \in [|A|]$,
  $Z$ intersects every path connecting $v^A_{i,d-1}$ and $v^A_{i,d}$,
  where $d = \psi(a_i)$.
  Indeed, every path $P^{A,\ell}_{i,d}$ corresponds to a constraint $\cC_\ell = R_\ell(a_i,b_j)$
  for some $j \in [|B|]$ and contains edges $e^{A,\ell}_{(i,d),(j,d')}$ for all $d' \in [n]$.
  By construction, $Z$ contains one of these edges.
  By a similar argument, $Z$ intersects every path connecting 
  $v^B_{j,d'-1}$ and $v^B_{j,d'}$ for all $j \in [|B|]$ such that $d' = \psi(b_j)$.
  Thus, no path from $s$ to $t$ remains in $G - Z$,
  and $Z$ is an $st$-cut.

  Now suppose there is an $st$-cut $Z'$ in $G'$ 
  that intersects at most $(2 - \eps)m$ pairs in $\cP$.
  We may assume without loss of generality that $Z'$ is minimal.
  We will now define an assignment $\psi' : V \to D$.
  Since $Z'$ is an $st$-cut, for every $i \in [|A|]$,
  there exists $d \in [n]$ such that
  $Z'$ intersects every path $P^{A,\ell}_{i,d}$.
  By minimality, this value $d$ is unique, and 
  it is sufficient to cut every path $P^{A,\ell}_{i,d}$ exactly once.
  Let us set $\psi'(a_i) = d$. 
  Note that the number of paths $P^{A,\ell}_{i,d}$
  is equal to the the number 
  of constraints in $\cC$ that involve $a_i$.
  We define $\psi'(b_j)$ analogously:
  for every $j \in [|B|]$, there is a unique $d' \in [n]$
  such that $Z'$ intersects every path $P^{B,\ell}_{j,d'}$.
  The total number of paths cut by $Z'$ equals to $2m$
  because every constraint is counted twice, hence $|Z'| = 2m$.
  We claim that $\psi'$ satisfies at least $\eps m$ constraints in $\cC$.
  By construction, every edge in $Z'$ is in some pair of $\cP$, 
  so $Z'$ contains both edges from at least $2m - |Z'| = \eps m$ pairs.
  Consider a pair of edges $\{e^{A,\ell}_{(i,d),(j,d')}, e^{B,\ell}_{(j,d'),(i,d)}\}$ in $\cP$.
  Recall that, by definition, $C_\ell = R_\ell(a_i,b_j)$ and $(d,d') \in R_\ell$,
  hence, $\psi'(a_i) = d$ and $\psi'(b_j) = d'$, and $\psi'$ satisfies $C_\ell$.
  In total, $\psi'$ satisfies at least $\eps m$ constraints.
  This completes the proof.
\end{proof}

\subsubsection{Non-lineal Rings}

We finally show how the non-approximability result for {\sc Split Paired Min Cut} (Theorem~\ref{thm:gap-paired-cut}) can be used for showing a similar non-approximability bound for
$\minlin{2}{R}$  when $R$ is a non-lineal ring. The proof is based on a simple but powerful tool for proving hardness results; it
underlies, for instance, the proofs of Theorem 27 in~\cite{dabrowski2023parameterized}
and Theorem 6.2 in~\cite{Dabrowski:etal:soda2023}.

\begin{theorem} \label{thm:incomparable-annihilators}
  Assuming \ETH, $\GAP{2-\eps}$-$\minlin{2}{R}$ is not in \FPT{} for any finite and commutative ring $R$ that
  is not lineal.
\end{theorem}
\begin{proof}
  We provide an fpt-reduction from $\GAP{(2-\eps)}$-\textsc{Split Paired Min Cut},
  which shows the result due to Theorem~\ref{thm:gap-paired-cut}.
  Proposition~\ref{prop:one-elem-lineal} allows us to assume that there exist elements
  $a,b \in R$ such that
  $\Ann(a)$ and $\Ann(b)$ are incomparable under set inclusion.
  Arbitrarily pick two elements $f \in \Ann(a) \setminus \Ann(b)$ and $g \in \Ann(b) \setminus \Ann(a)$,
  and note that $af=0$, $bf \neq 0$, $ag \neq 0$, and
  $bg=0$. The elements $a,b,f,g$ are now chosen
  so that they witness that $R$ does not have the magic
  square property.

  Let $(G, s, t, G_1,G_2,k, \cP)$ be an instance of \textsc{Split Paired Min Cut},
  and assume subsets $U_1, U_2 \subseteq V(G)$ form the split.
  Construct instance $(S, k)$ of \minlin{2}{R} as follows.
  Start by adding one variable to $S$ 
  for every vertex in $V(G) \setminus \{s,t\}$ with the same name.
  Create variables $s_1, s_2$ for $s \in V(G)$
  and $t_1, t_2$ for $t \in V(G)$.
  Add crisp equations $s_1=g$, $s_2=f$, $t_1=0$, $t_2=0$.
  For every edge $uv \in E(G)$ 
  not present in any pair of $\cP$,
  add a crisp equations $u = v$.
  We say that the variables introduced so far are {\em primary}.
  Finally, for every pair $p = (u_1 v_1, u_2 v_2) \in \cP$,
  we introduce two {\em auxiliary} variables $x_p$ and $y_p$,
  and create a gadget with a soft equation $x_p = y_p$
  and crisp equations
  \begin{align*}
    &au_1 = ax_p && ay_p = av_1, \\
    &bu_2 = bx_p && by_p = bv_2. 
  \end{align*}
  The construction is illustrated in Figure~\ref{fig:gadget}.

\medskip
  
  Towards proving correctness, let $Z$ be an $st$-cut in $G$ which is a union of $k$ pairs in $\cP$.
  We will show that $(S, k)$ is a yes-instance.
  Let $X \subseteq \cP$ be the set of pairs such that $Z = \bigcup X$, and
  define the set of equations $X' = \{ x_p = y_p \mid p \in X \}$.
  Note that $|X'| = |X| = k$ and $X'$ only contains soft equations.
  We argue that $S - X'$ is consistent under an assignment $\alpha$ defined as follows.
  For the primary variables, let
  \[
    \alpha(v) = 
    \begin{cases}
      g & \text{if } v \in U_1 \text{ and } s \text{ reaches } v \text{ in } G - Z, \\
      f & \text{if } v \in U_2 \text{ and } s \text{ reaches } v \text{ in } G - Z, \\
      0  & \text{otherwise}.
    \end{cases}
  \]
  Since $Z$ is an $st$-cut, we have $\alpha(t_1) = \alpha(t_2) = 0$.
  Moreover, $\alpha(u) = \alpha(v)$ for every crisp equation $u = v$ in $S$
  because $Z$ is a union of pairs in $\cP$.
  For the auxiliary variables, let
  $\alpha(x_p) = \alpha(u_1) + \alpha(u_2)$ and
  $\alpha(y_p) = \alpha(v_1) + \alpha(v_2)$.
  Note that $\alpha(u_1), \alpha(v_1) \in \{0, g\}$ and
  $\alpha(u_2), \alpha(v_2) \in \{0, f\}$, and
  it is straightforward to verify that
  $\alpha$ satisfies equations
  $au_1 = ax_p$, $av_1 = ay_p$,
  $bu_2 = bx_p$, and $bv_2 = by_p$.
  Now suppose $p = (u_1 v_1, u_2 v_2) \notin X$.
  Then edges $u_1 v_1$ and $u_2 v_2$ are present in $G - Z$,
  which implies that $\alpha(u_1) = \alpha(v_1)$ and $\alpha(u_2) = \alpha(v_2)$.
  By definition, we obtain $\alpha(x_p) = \alpha(y_p)$, so $\alpha$ satisfies $S - X'$.

\medskip

  For the other direction, suppose there exists a subset
  $Y$ of soft equations such $S - Y$ is satisfiable.
  We claim that $G$ admits an $st$-cut that intersects $|Y|$ pairs in $\cP$.
  By definition, the only soft equations in $S$ are $x_p = y_p$ for pairs $p \in \cP$.
  Define $Y' = \bigcup \{ p \mid (x_p = y_p) \in Y\}$ and 
  note that $Y'$ intersects exactly $|Y|$ pairs in $\cP$.
  We claim that $Y'$ is an $st$-cut in $G$.
  Suppose, with the aim of obtaining a contradiction, that $G_1 - Y'$ contains an $st$-path
  (the case for $G_2$ is similar).
  For every edge $u_1 v_1$ on the path, the equations in $S - Y$
  imply $au_1 = av_1$ (either via a crisp equation
  $u_1 = v_1$ or equations $au_1 = ax_p$, $x_p = y_p$, $ax_p = av_1$
  in the gadget).
  But then the equations in $S - Y$ imply
  $as_1 = at_1$, which is incompatible with the crisp equations
  $s_1 = g$ and $t_1 = 0$ since $ag \neq 0$.
\end{proof}

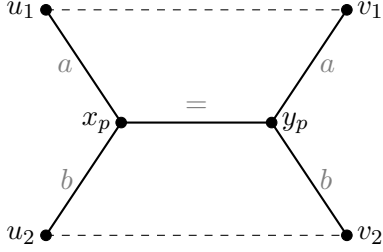
\begin{figure}
 \centering
  \begin{tikzpicture}
    \coordinate (u1) at (0,3);
    \coordinate (v1) at (4,3);
    \coordinate (u2) at (0,0);
    \coordinate (v2) at (4,0);
    \coordinate (xp) at (1,3/2);
    \coordinate (yp) at (3,3/2);

    \filldraw[black] (u1) circle (2pt) node[anchor=east]{$u_1$};
    \filldraw[black] (v1) circle (2pt) node[anchor=west]{$v_1$};
    \filldraw[black] (u2) circle (2pt) node[anchor=east]{$u_2$};
    \filldraw[black] (v2) circle (2pt) node[anchor=west]{$v_2$};
    \filldraw[black] (xp) circle (2pt) node[anchor=east]{$x_p$};
    \filldraw[black] (yp) circle (2pt) node[anchor=west]{$y_p$};

    \draw[dashed] (u1) -- (v1);
    \draw[dashed] (u2) -- (v2);
    \draw[thick] (u1) -- (xp) node[midway, left,  color=gray] {$a$};
    \draw[thick] (u2) -- (xp) node[midway, left,  color=gray] {$b$};
    \draw[thick] (xp) -- (yp) node[midway, above, color=gray] {$=$};
    \draw[thick] (yp) -- (v1) node[midway, right, color=gray] {$a$};
    \draw[thick] (yp) -- (v2) node[midway, right, color=gray] {$b$};
  \end{tikzpicture}
 \caption{System of equations obtained from a pair of edges $p = \{e_1, e_2\}$ 
 where $e_i = \{u_i, v_i\}$ in the proof of Theorem~\ref{thm:incomparable-annihilators}.
 Edges $e_1$ and $e_2$ are illustrated by dashed lines, while the equations
 are illustrated by solid lines with labels describing 
 equations between connected variables.}
 \label{fig:gadget}
\end{figure}

\section{Geometric Approach}
\label{sec:geometry}

To get a better view of Helly and lineal rings, we study them through a
geometric
lens via {\em discrete convex analysis}.
This branch of mathematics aims at a general framework for discrete optimization problems by
combining ideas from continuous optimization and combinatorial optimization~\cite{Murota:DCA}. 
From the continuous side, it may be viewed as a theory of 
convex functions $f : {\mathbb R}^d \rightarrow {\mathbb R}$
that have additional combinatorial properties. From the combinatorial side, it may be viewed as a theory of discrete functions $f : {\mathbb Z}^d \rightarrow {\mathbb R}$ or $f : {\mathbb Z}^d \rightarrow {\mathbb Z}$
that enjoy certain convexity-like properties.
To study commutative rings in a smooth way within this framework, we will focus on {\em monomial rings}, i.e. polynomial rings (with coefficients from some field) that are
factored by monomial ideals. 
Monomial ideals have a structure that makes them very useful for illustrative purposes. Despite being simple compared to polynomial ideals, they are still of fundamental importance in commutative algebra
and algebraic combinatorics~\cite{Cox:etal:IVA,Eisenbud:CA}. For instance,
monomial ideals play a crucial role in
the theory of Gröbner bases.

\subsection{Discrete Convex Analysis}
\label{sec:monomial-rings}

We begin by introducing monomial rings.
Let ${\mathbb F}$ be a field and $R={\mathbb F}[\rx_1,\dots,\rx_n]$
be the polynomial ring over ${\mathbb F}$ with
indeterminates $\rx_1,\dots,\rx_n$.
Recall that a monomial of $R$ is a product
$\rx_1^{\alpha_1} \rx_2^{\alpha_2} \dots \rx_n^{\alpha_n}$ with $\alpha_1,\dots,\alpha_n \in {\mathbb N}$.
A {\em monomial ideal} is an ideal $I \subseteq R$ that is
generated by monomials in $R$. 
For any monomial $\rx_1^{\alpha_1} \rx_2^{\alpha_2} \dots \rx_n^{\alpha_n}$, its {\em exponent vector}
is $(\alpha_1,\dots,\alpha_n)$. For any monomial ideal,
the set of all exponent vectors of all the monomials in $I$ is called the {\em exponent set} of $I$ and we denote it by $Exp(I)$. 
A {\em monomial ring} is a ring that is isomorphic to
${\mathbb F}[\rx_1,\dots,\rx_n]/I$ where ${\mathbb F}$ is a field
and $I$ is a monomial ideal.
For a monomial ring $R=\FF[\rx_1,\ldots,\rx_n]/I$ we denote $Z_E(R)=Exp(I)$ and we let $\overline{Z_E}(R)=\NN^n \setminus Z_E(R)$,
so that $\overline{Z_E}(R)$ is the set of all exponent vectors of non-zero monomials in $R$. 
For $\alpha \in \NN^n$ we let $\rx^\alpha$ denote $\rx_1^{\alpha_1} \dots \rx_n^{\alpha_n}$. 
We illustrate the concepts by using
Example 1.4.4 in \cite{HunekeS2006integral}.
Let $I=(\rx^4,\rx\ry^2,\ry^3)$
be a monomial ideal in ${\mathbb F}[\rx,\ry]$. Its
exponent set $Z_E(R)$ consists of all integer lattice points touching or within the shaded
grey area in Figure~\ref{fig:geometry-example}.

\medskip

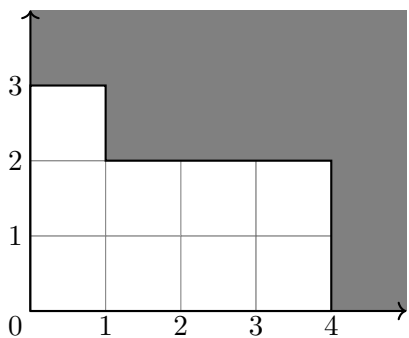
\begin{figure}[h]
\centering
\begin{tikzpicture}
    \fill[gray,opacity=0.5] (0,0) -- (5,0) -- (5,4) -- (0,4) -- cycle;

    \draw[thick,->] (0,0) -- (5,0);
    \draw[thick,->] (0,0) -- (0,4);
    
    \node at (-0.2, -0.2) {0};
    \node at (4,    -0.2) {4};
    \node at (1,    -0.2) {1};
    \node at (2,    -0.2) {2};
    \node at (3,    -0.2) {3};
    \node at (-0.2,    1) {1};
    \node at (-0.2,    2) {2};
    \node at (-0.2,    3) {3};

    \draw[gray, fill=white] (0,0) rectangle (1,1);
    \draw[gray, fill=white] (1,0) rectangle (2,1);
    \draw[gray, fill=white] (2,0) rectangle (3,1);
    \draw[gray, fill=white] (3,0) rectangle (4,1);
    
    \draw[gray, fill=white] (0,1) rectangle (1,2);
    \draw[gray, fill=white] (1,1) rectangle (2,2);
    \draw[gray, fill=white] (2,1) rectangle (3,2);
    \draw[gray, fill=white] (3,1) rectangle (4,2);
    
    \draw[gray, fill=white] (0,2) rectangle (1,3);

    \draw[thick] (0,0) -- (0,3) -- (1,3) -- (1,2) -- (4,2) -- (4,0) -- cycle;

\end{tikzpicture}
\caption{Illustration of $R=\FF[\rx,\ry]/(\rx^4,\rx\ry^2,\ry^3)$}
\label{fig:geometry-example}
\end{figure}

\medskip

If $G$ is a monomial generating set of an ideal $I$, the exponent set of $I$ consists of
all points of ${\mathbb N}^d$
that are componentwise greater than or equal to the
exponent vector of one of the exponent vectors of an element of $G$. In other
words, a monomial $m$ is in a monomial ideal $I$ if and only if $m$ is a multiple
of one of the monomial generators of $I$.
We note that elements in monomial rings have unique descriptions as sums of monomials. 

\begin{proposition} \label{prop:mr-unique-representation}
    Let $R=\FF[\rx_1,\dots,\rx_n]/I$ be a monomial ring. Then every element $r \in R$ has a unique description $r=\sum_i a_im_i$ as a weighted sum of monomials, where for each $i$, $a_i \in \FF$ is a coefficient and $m_i$ a non-zero monomial.
\end{proposition}
\begin{proof}
    Assume that an element $r \in R$ can be written in two distinct ways $r=\sum_i a_im_i= \sum_i b_im_i$ in $R$ and consider $0=\sum_i(a_i-b_i)m_i$. Then this sum is in the ideal $I$. But since $I$ is monomial, every monomial $m_i$ with a non-zero coefficient $a_i-b_i$ must be in $I$. Since by assumption the sum contains no monomials from $I$, we conclude $a_i=b_i$ for every index $i$. 
\end{proof}

We continue by introducing some basic concepts from discrete convex
analysis. An in-depth treatment of the topic can be found in the book by Murota~\cite{Murota:DCA};
our terminology largely follows Kashimura, Numata, and Takemura (KNT)~\cite{KashimuraNT13}.
To begin with, let $S$ be a nonempty set of integer points in ${\mathbb{R}}^d$. 
The {\em convex combination} of $a_1,\dots,a_m \in \reals^d$
is a point $\alpha_1 a_1 + \dots + \alpha_m a_m$ where
$\alpha_1,\dots,\alpha_m \geq 0$ and $\alpha_1+\dots+\alpha_m=1$.
The {\em convex hull} of a set $X \subseteq \reals^n$ ($cnv(X)$)
is the set of all convex combinations of points in $X$.
We say that $S \subseteq {\mathbb Z}^d$ is
$k$-{\em convex} if $S$ satisfies the following: for arbitrary
$a_1, \dots, a_{k+1} \in S$, it holds that $cnv(a_1,\dots,a_{k+1}) \cap {\mathbb Z}^d \subseteq S$.
The set $S$ is \emph{hole free} if $S=cnv(S) \cap {\mathbb{Z}}^d$.

Next, consider a separation of $S$ into two nonempty sets $A \subseteq S$ and $B=S \setminus A$. 
We say that $A$ and $B$ are \emph{separated by a hyperplane} if 
there exists an affine hyperplane $H$, where $A$ is on one side of $H$ and $B$ on the other.
This is a simplification of the more general Condition~H of KNT, but it suffices
for our purposes. 
In particular, in our applications one of the sets $A$ and $B$ will always be finite,
which allows us to assume that the separating hyperplane contains no points from $S$.
We say that $A$ and $B$ satisfy the $k$-{\em parallelogram} property
if 
$a_1,\dots,a_{k'} \in A$ and $b_1,\dots,b_{k'} \in B$ 
implies 
\[\sum_{i=1}^{k'} a_i \neq \sum_{j=1}^{k'} b_j\]
for all $1 \leq k' \leq k$. The 2-parallelogram property
is known as {\em Condition P}.
If we consider Condition P with $a_1 = a_2$, then we see that
there is no $a \in A$ such that there exists $b_1,b_2 \in B$
with $a=\frac{1}{2} b_1 + \frac{1}{2}b_2$ (and more generally, that $A$ and $B$ are 1-convex). 
Moreover, if $A$ and $B$ are separated by a hyperplane, then so are the convex hulls of $A$ and $B$, 
hence this implies that both $A$ and $B$ are hole-free, and in particular Condition~P applies to $A$ and $B$.
(See also~\cite[Lemma~2.12]{KashimuraNT13} for the more general Condition~H eluded to above.)

We can now provide the geometric interpretations of the properties of
lineal and Helly monomial rings. 
Let $R={\mathbb F}[\rx_1,\dots,\rx_n]/I$ be a monomial ring, let $Z=Z_E(R) \subseteq \NN^n$ be the
exponent set of $I$ and let $N=\overline{Z_E}(R)=\NN^n \setminus Z$ be the exponent set
of the non-zero monomials of $R$. 
Informally, we find that monomial lineal rings correspond roughly to the case
where both $N$ and $Z$ satisfy conditions of discrete convexity, and
monomial Helly rings correspond to a discrete convexity property
applying to $Z$ but not to $N$. However, the ``naive'' interpretation
of simply requiring $Z$ and/or $N$ to be hole-free is too restrictive. 

More formally, we have the following definitions.
We say that $R$ is {\em Z-hole-free}
if $Z$ is hole-free, i.e.\ $Exp(I) = cnv(Exp(I)) \cap {\mathbb Z}^n$,
and \emph{Z-1-convex} if $Z$ is 1-convex. 
Finally, we say that $R$ is \emph{fully convex} if $N$ and $Z$ are separated by a hyperplane.
We show relations and separations between these properties, illustrated in Figure~\ref{figure:link}.
In summary, we show the following.  
\begin{enumerate}
\item $R$ is lineal if and only if $(N,Z)$ satisfies Condition P.

\item The case where $R$ is fully convex is a natural special case of this.

\item While we do not have an exact characterization of monomial Helly
  rings, we sandwich it between two conditions. First, we show that if
  $R$ is Z-hole-free, then $R$ is Helly. Second, we show that if $R$ is Helly, then $Z$
  is 1-convex.  

\item Finally, we show some separation examples: We demonstrate
  monomial lineal rings $R$ (in fact, Bergen rings) where $N$ is not hole-free, and where $Z$ is
  not hole-free. This implies that not all monomial lineal rings are fully convex,
  and not all monomial Helly rings (or even monomial Bergen rings) are Z-hole-free.
\end{enumerate}

We promised in Section~\ref{sec:annihilators} to show that
$\ZZ_2[\rx,\ry] / (\rx^3,\rx\ry,\ry^3)$ is Helly.
This is very easy using the geometric approach together with the results outlined
above.
We illustrate
the exponent set $Z$ of $(\rx^3,\rx\ry,\ry^3)$ in Figure~\ref{fig:convex-hulls} (left): $Z$ consists of all integer points touching or
being inside the grey area. The dotted line illustrates the convex
hull of $Z$. We see that $Z=cnv(Z) \cap {\mathbb Z}^2$ so
$\ZZ_2[\rx,\ry] / (\rx^3,\rx\ry,\ry^3)$ is Z-hole-free and thus Helly. However, it is not lineal since $(1,1)+(1,1)=(0,2)+(2,0)$ contradicts Condition~P.
We take the opportunity to make another illustration.
Let us consider
the monomial ideal $I=(\rx^4,\rx\ry^2,\ry^3)$ of $\FF[\rx,\ry]$.
The convex hull of its exponent set $Z$ is depicted in
Figure~\ref{fig:convex-hulls} (right) and it shows that $Z \neq cnv(Z) \cap {\mathbb Z}^2$.
The ring ${\mathbb Z}_2[\rx,\ry]/I$ is thus not Z-hole-free but this does
not directly imply that the ring is non-Helly.
However, the points $(4,0)$ and $(2,2)$ are in $Z$ while $(3,1) \not\in Z$
and $(3,1)$ is in $cnv((4,0),(2,2))$. The exponent set of $I$ is thus not 1-convex
and the ring is not Helly.

The rest of the section is arranged as follows. 
To simplify the presentation, we use a support property of rings that we refer to as the \emph{no cancellations} property; see Section~\ref{sec:no-cancellation}.
Section~\ref{sec:lineal-geometry} studies monomial lineal rings, with Section~\ref{sec:fully-convex} digging deeper on fully convex rings. Section~\ref{sec:convex} contains the results on monomial Helly rings, and 
Section~\ref{sec:geometry-other} discusses connections to other notions from ring theory.

\begin{figure}
\centering
\begin{subfigure}{.5\textwidth}
  \centering
\begin{tikzpicture}
    \fill[gray,opacity=0.5] (0,0) -- (5,0) -- (5,4) -- (0,4) -- cycle;

    \draw[thick,->] (0,0) -- (5,0);
    \draw[thick,->] (0,0) -- (0,4);
    
    \node at (-0.2, -0.2) {0};
    \node at (4,    -0.2) {4};
    \node at (1,    -0.2) {1};
    \node at (2,    -0.2) {2};
    \node at (3,    -0.2) {3};
    \node at (-0.2,    1) {1};
    \node at (-0.2,    2) {2};
    \node at (-0.2,    3) {3};

    \draw[gray, fill=white] (0,0) rectangle (1,1);
    \draw[gray, fill=white] (1,0) rectangle (2,1);
    \draw[gray, fill=white] (2,0) rectangle (3,1);
    
   \draw[gray, fill=white] (0,1) rectangle (1,2);
    
    \draw[gray, fill=white] (0,2) rectangle (1,3);

    \draw[thick] (0,0) -- (0,3) -- (1,3) -- (1,1) -- (3,1) -- (3,0) -- cycle;

    \draw[dotted] (0,3) -- (1,1) -- (3,0);
    
\end{tikzpicture}
  \caption{$\FF[\rx,\ry] / (\rx^3,\rx\ry,\ry^3)$}
  \label{fig:sub1-second}
\end{subfigure}%
\begin{subfigure}{.5\textwidth}
  \centering
 \begin{tikzpicture}
    \fill[gray,opacity=0.5] (0,0) -- (5,0) -- (5,4) -- (0,4) -- cycle;

    \draw[thick,->] (0,0) -- (5,0);
    \draw[thick,->] (0,0) -- (0,4);
    
    \node at (-0.2, -0.2) {0};
    \node at (4,    -0.2) {4};
    \node at (1,    -0.2) {1};
    \node at (2,    -0.2) {2};
    \node at (3,    -0.2) {3};
    \node at (-0.2,    1) {1};
    \node at (-0.2,    2) {2};
    \node at (-0.2,    3) {3};

    \draw[gray, fill=white] (0,0) rectangle (1,1);
    \draw[gray, fill=white] (1,0) rectangle (2,1);
    \draw[gray, fill=white] (2,0) rectangle (3,1);
    \draw[gray, fill=white] (3,0) rectangle (4,1);
    
    \draw[gray, fill=white] (0,1) rectangle (1,2);
    \draw[gray, fill=white] (1,1) rectangle (2,2);
    \draw[gray, fill=white] (2,1) rectangle (3,2);
    \draw[gray, fill=white] (3,1) rectangle (4,2);
    
    \draw[gray, fill=white] (0,2) rectangle (1,3);

    \draw[dotted] (0,3) -- (1,2) -- (1,3) -- cycle;
    \draw[dotted] (1,2) -- (4,2) -- (4,0) -- cycle;

    \draw[thick] (0,0) -- (0,3) -- (1,3) -- (1,2) -- (4,2) -- (4,0) -- cycle;

 \draw[thick,->] (6,4) -- (3,1);

\node[draw] at (6.3,4.3) {{\bf !}};
   
\end{tikzpicture}
  \caption{$\FF[\rx,\ry] / (\rx^4,\rx^2\ry,\ry^3)$}
  \label{fig:sub2-second}
\end{subfigure}
\caption{Two rings and the exponent sets of $Z_E(R)$. The convex hull of $Z_E(R)$ is marked by a dotted line.}
\label{fig:convex-hulls}
\end{figure}
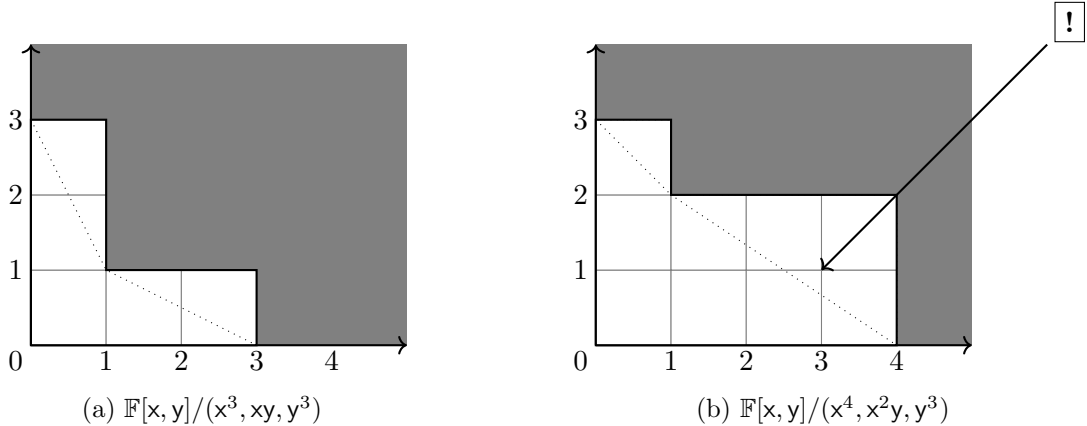

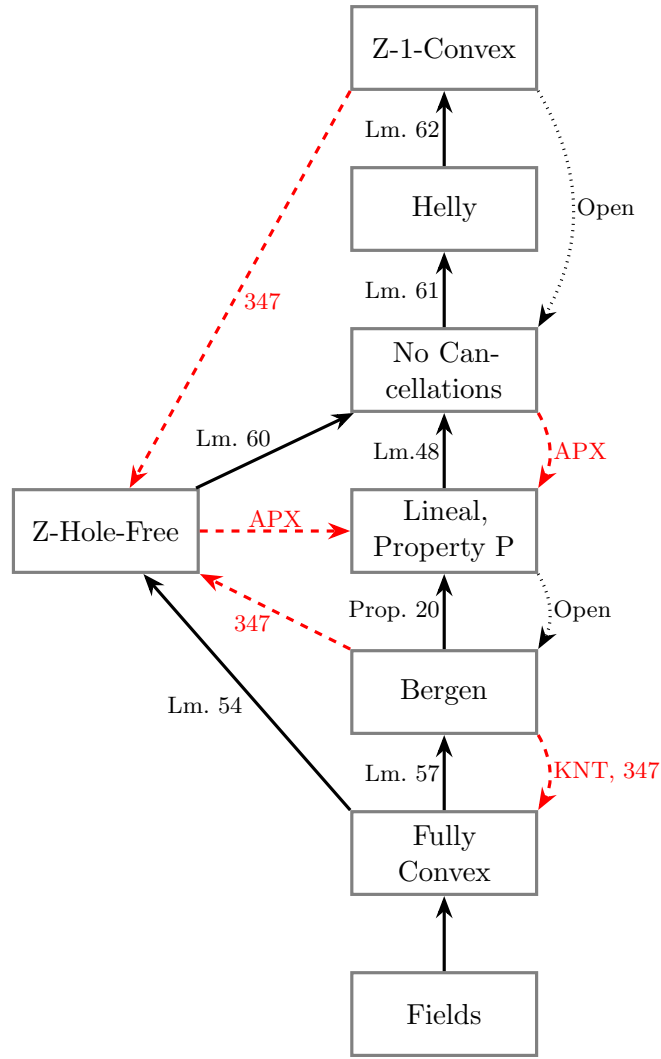
\begin{figure}
  \centering
  \input{figures/diag-monom.tikz}
  \caption{Known implications and open questions for lineal and Helly
    monomial rings. Solid black edges represent implications, red
    dashed edges represent counterexamples. 
    ``KNT'', ``347'' and ``APX'' refers to
    implications where the $R_{\rm KNT}$, $R_{347}$, and
    $R=\FF[\rx,\ry]/(\rx^3,\rx\ry,\ry^3)$ is a counterexample, respectively.
  }
  \label{figure:link}
\end{figure}

\subsection{The No Cancellations Property}
\label{sec:no-cancellation}
To support the forthcoming proofs and discussions, we consider a property of monomial rings
that we refer to as \emph{no cancellations}. 

\begin{definition}
  A monomial ring $R$ has the \emph{no cancellations property} if
  $pq=0$ for $p, q \in R$ implies that $m_pm_q=0$ for every pair of
  monomials $m_p$ in $p$ and $m_q$ in $q$.
\end{definition}

We note that not every monomial ring has the no cancellations property.
Consider the ring $R=\ZZ_2[\rx,\ry]/(\rx^2,\ry^2)$.
This has a zero product $(\rx+\ry)(\rx-\ry)=\rx^2+\rx\ry-\rx\ry+\ry^2=0$ (where naturally $+$ and $-$ are the same operation 
over characteristic 2, but the statement holds also if $\ZZ_2$ is replaced by an arbitrary field $\FF$).
Thus $R$ does not have the no cancellation property since $\rx\ry \neq 0$. 
However, this ring is non-Helly as noted above, implying that \minlin{2}{R} is \W{1}-hard to approximate.

We show that every lineal monomial ring has the no cancellations property.
In Section~\ref{sec:convex} we show that the no cancellations property implies that the ring is Helly. 

\begin{lemma} \label{lm:lineal-implies-nc}
  Every lineal monomial ring has the no cancellations property. 
\end{lemma}
\begin{proof}
  Let $R$ be a lineal, monomial ring and assume that $R$ does not have the no cancellations property. 
  First we note that for any monomial $m$ and any $r \in R$
  we have $mr=0$ if and only if $mm'=0$ for every monomial $m'$ occurring in $r$. 
  Indeed, write $r=\sum_i a_im_i$ as a sum of distinct non-zero monomials $m_i$, where $a_i>0$ is an integer coefficient for each $i$.
  Then $mr=\sum_i a_i (mm_i)$ is a representation of 0, which by Proposition~\ref{prop:mr-unique-representation} implies 
  that $mm_i=0$ for every monomial $m_i$ in $r$. 
  Now let $pq=0$ be a counterexample to the no cancellations property of $R$,
  chosen to be minimal in the sense of the cardinality of the monomial supports of $p$ and $q$.
  This implies that for every monomial $m_p$ in $p$, we have $m_pq \neq 0$. 
  Indeed, assume $m_pq=0$ and let $a$ be the coefficient of $m_p$ in $p$. 
  Let $p'=p-am_p$. Then $p'q=0$ is a smaller counterexample than $pq=0$
  (note that $p'q$ still retains all non-zero monomial products from $pq$).
  Thus we assume that $m_pq, pm_q \neq 0$ for all monomials $m_p$ in $p$ and $m_q$ in $q$. 
  We furthermore claim that $pq$ contains at least one zero monomial product.
  Indeed, let $m_p$ and $m_q$ be the earliest monomials in some monomial order (e.g. lexicographic order);
  then the monomial $m_pm_q$ is the unique earliest monomial of $pq$ in the same order, which 
  implies that it is produced only once. Under our assumption that $R$ is a monomial ring,
  over some base field $\FF$, the coefficient of $m_pm_q$ is non-zero in $pq$; hence
  $m_pm_q=0$. 
  Now we have a counterexample to lineality of $R$ in $pq=0$, $pm_q \neq 0$, $m_pq \neq 0$, $m_pm_q=0$,
  contrary to assumptions. 
\end{proof}

We also observe an equivalent characterization that illustrates how
annihilators enter the picture.

\begin{lemma} \label{lm:nc-gives-monomial-ideals}
  $R$ has the no cancellations property if and only if every
  annihilator of $R$ is a monomial ideal.
\end{lemma}
\begin{proof}
On the one hand, let $pq=0$ be a non-trivial cancellation in $R$.
Then $p \in \Ann(q)$ but $p$ is not contained in the ideal generated
by only the monomials of $\Ann(q)$ (since there by assumption is a
monomial $m_p$ in $p$ such that $m_pq \neq 0$).

On the other hand, assume that $R$ has the no cancellations property.
It suffices to consider annihilators $\Ann(r)$ for $r \in R$ since
intersections of monomial ideals are monomial. 
Let $I \subseteq \Ann(r)$ be the ideal generated by the monomials in $\Ann(r)$,
and assume that the inclusion is strict. By definition, $I$ contains precisely
all monomials $m$ such that $mr=0$. Thus there is an element
$p \in \Ann(r) \setminus I$ which contains at least one monomial $m_p$
such that $m_pr \neq 0$, yet $pr=0$. This contradicts the no cancellations
property of $R$ so $I=\Ann(r)$ and every annihilator in $R$ is a monomial ideal.
\end{proof}

\subsection{Lineal Rings}
\label{sec:lineal-geometry}

As announced, a monomial ring $R$ over $n$ variables is lineal if and only if the separation $(N,Z)$ 
of $\NN^n$ into exponents of non-zero and zero monomials meets Condition~P. 
Let us prove this fact. 

\begin{lemma}
  Let $R=\mathbb{F}[\rx_1,\dots,\rx_n]/I$ be a monomial ring.
  Let $Z=Exp(I)$ be the exponent set of $I$ and $N=\NN^n \setminus Z$.
  The following are equivalent.
  \begin{enumerate}
  \item $R$ is non-lineal
  \item There are points $a, b, c, d \in \NN^n$ such that $a+b, c+d \in Z$
    and $a+d, b+c \in N$
  \item There is a pair of lines $p_1$-$p_2$ and $q_1$-$q_2$
    with endpoints $p_1, p_2 \in Z$ and $q_1, q_2 \in N$ whose midpoints intersect.
  \item Condition P is false for $N$ and $Z$, i.e.
    there are $p_1, p_2 \in Z$ and $q_1, q_2 \in N$ such that $p_1+p_2=q_1+q_2$
  \end{enumerate}
\end{lemma}
\begin{proof}
  The second clearly implies the first since $X^a$, \dots, $X^d$ form
  a non-magic square. On the other hand, let $a, b, c, d \in R$ form a non-magic
  square so that $ab=0$, $cd=0$, $ad \neq 0$ and $bc \neq 0$. 
  We may assume that $R$ has the no cancellations property, as
  otherwise it is non-lineal. Thus, let $a', b', c', d'$ be exponent vectors
  of monomials in $a, b, c, d$ (respectively) such that $X^{a'+d'} \neq 0$
  and $X^{b'+c'} \neq 0$. Then $a'+d', b'+c' \notin Z$ and by the no
  cancellations property we have $a'+b', c' + d' \in Z$. 
  The remaining two conditions are reformulations of the second item
  and can be verified by straight-forward algebra.
\end{proof}

As noted in Section~\ref{sec:monomial-rings}, these conditions are discrete approximations of convexity.
Indeed, one special case of two crossing lines is a zero point
$p_1=p_2$ that lies on a line between non-zero points $q_1$-$q_2$, or
vice versa; more strongly, lineality implies that both $N$ and
$Z$ are 1-convex.  However, this characterization is not tight,
and in fact there are non-lineal rings where both $N$ and $Z$ are hole-free;
see Appendix~\ref{sec:counter:double-hole-free}.
We next consider the special case of fully convex rings, 
but to set the stage we first consider an alternative characterization of lineality.

\subsubsection{Threshold Labellings}
\label{sec:threshold}

We begin by providing an alternative characterization of lineality for arbitrary
finite commutative rings. Let $R$ be a finite commutative ring.
A \emph{threshold labelling} for $R$ is
a real-valued labelling $L \colon R \to [0,T]$ for a threshold $T$
such that for every $u,v \in R$, $uv=0$ if and only if
$L(u)+L(v)\geq T$. 
As a special case for monomial rings, a \emph{monomial threshold labelling} is defined as follows.
Let $M$ be the set of non-zero monomials in $R$ 
and let $L_m \colon M \to [0,T)$ for a threshold $T$
be a ``threshold labelling over $M$'', i.e. for any $a, b \in M$
we have $L_m(a)+L_m(b) \geq T$ if and only if $ab=0$.
Extend $L_m$ to $R$ via
\[
  L(u) = \min_{m \in u} L_m(u) \text{ with } L(0)=T
\]
where $m$ ranges over all monomials occurring in $u$; this is well-defined via Proposition~\ref{prop:mr-unique-representation}.
Then $L$ is a monomial threshold labelling for $R$.
We will show that $R$ is lineal if and only if it has a threshold labelling,
and that this can be assumed to be monomial for monomial rings. 

We prove that every lineal ring has a threshold labelling by using a particular graph.
A {\em zero-divisor graph} is an 
undirected graph representing the zero divisors of a commutative ring. For instance, it 
may have the elements of the ring as its vertices, and pairs of elements whose product is zero as its edges~\cite{Beck:jal1988}. We introduce a variation of this idea.
The {\em bipartite zero-divisor graph} $G$ with $V(G)=A \uplus B$ is constructed
as follows. Let $A$ and $B$ contain the elements in $R$. Join an element $r \in A$ and an element $s \in B$ if and 
only if $rs=0$. The magic square property implies that a ring is non-lineal if and only if the bipartite zero-divisor graph
contains a $2K_2$. We can connect bipartite zero-divisor graphs with
threshold labellings by utilizing {\em difference graphs}.
  A difference graph is an undirected graph $G$ with vertex labelling
  $\alpha \colon V(G) \to \reals$ and a threshold $T \in \reals$ such that
  $|\alpha(v)| < T$ for every vertex $v$, and $uv$ is an edge if and only if $|\alpha(u)-\alpha(v)| \geq T$.
  It is known that difference graphs are precisely the $2K_2$-free
  bipartite graphs~\cite[Lemma~2.2]{HammerPS90}. Hence, there is such a labelling for 
  the bipartite zero-divisor graph of $R$ if $R$ is lineal.

\begin{lemma}
  A finite commutative ring $R$ is lineal if and only if it has a
  threshold labelling. If $R$ is furthermore a monomial ring, then the threshold labelling can be assumed to be monomial.
\end{lemma}
\begin{proof}
  {\bf Forward direction.}
  Assume that $R$ is lineal and
  let $G$ be the bipartite zero-divisor graph of $R$ with $V(G)=A \uplus B$.
  We know that $G$ is $2K_2$-free so there is a function $L:A \cup B \rightarrow \reals$
  and a $T \in \reals$ such that $L(v) < T$ for every $v \in A \cup B$ and $|L(u)-L(v)| \geq T$ for every edge $uv \in E$.
  We construct a threshold labelling $L',T'$ with the aid of $L,T$.
  Assume without loss of generality that $L$ is non-negative on $A$ and non-positive on $B$
  in the bipartition of the zero-divisor graph.
Define $L_A : A \rightarrow \reals^+$ such that
$L_A(u)=L(u)$ and $L_B : B \rightarrow \reals^+$ such that $L_B(u)=-L(u)$.
Now, for any $u, v \in R$ we
  have $uv=0$ if and only if $L_A(u)+L_B(v) \geq T$ if and only if
  $L_A(v)+L_B(u) \geq T$. 
  By defining $L'(u)=L_A(u)+L_B(u)$ and setting $T'=2T$, we see that
  $[0,T']$ is the range of $L'$ and for arbitrary $u,v \in R$, $uv=0$
  if and only if $L'(u)+L'(v) \geq T'$.

\smallskip

\noindent
  {\bf Backward direction.} Let $L$ and $T$ be a threshold labelling for $R$ and arbitrary
  choose elements $a,b,c,d \in R$ such
  that
  $ab=0$ and $cd=0$. 
  Then $L(a)+L(b)+L(c)+L(d) \geq 2T$ and at least one of
  $L(a)+L(d) \geq T$ or $L(b)+L(c) \geq T$ holds,
  implying, respectively, $ad=0$ or $bc=0$. Hence, a threshold labelling implies
  the magic square property of $R$ and $R$ is lineal.

  \smallskip
  \noindent 
    {\bf Monomial rings.} Finally, we show that if $R$ is monomial then it suffices
    to consider monomial threshold labellings. Assume that $R$ is monomial and lineal.
    By Lemma~\ref{lm:lineal-implies-nc}, $R$ has the no cancellations property, 
    i.e.\ for any $r, s \in R$ we have $rs=0$ if and only if $m_rm_s=0$
    for every pair of monomials $m_r$ occurring in $r$ and $m_s$ occurring in $s$.
    Let $L$ be a threshold labelling of $R$. Then for any $r, s \in R$ we have
    \[
        rs = 0 \text{ if and only if } \min_{m_r \in r} L(m_r) + \min_{m_s \in s} L(m_s) \geq T,
    \]
    taking the empty minimization $\min_{m \in 0} L(m)$ to return $T$. 
    Then clearly we can use $L$, restricted to only monomials of $R$, 
    as the basis for a derived monomial threshold labelling $L'(r)=\min_{m \in r} L(m)$.
\end{proof}

In particular, a monomial threshold function $L$ for a monomial ring $R=\mathbb{F}[\rx_1,\dots,\rx_n]/I$,
can be viewed as a function $f \colon \integers_{\geq 0}^n \to \QQ$
where $f(a_1,\ldots,a_n)=L(\rx_1^{a_1} \dots \rx_n^{a_n})$.
In general, this function does not appear to have any strong guaranteed further properties,
but we will see next that it allows for a very clear characterization of fully convex rings. 

\subsubsection{Fully Convex Rings}
\label{sec:fully-convex}

Let $R$ be a finite, lineal, commutative ring. 
Since $R$ is local by Proposition~\ref{prop:lineal-is-local}, 
Corollary~\ref{cor:localringstandard} allows us to assume that
$R$ is defined as a quotient ring over some polynomial ring,
and if $R$ is monomial we furthermore assume a representation
$R=\FF[\rx_1,\ldots,\rx_n]/I$ where $\FF$ is a field and $I$ is a monomial ideal.
Recall that $R$ is \emph{fully convex} if $Z=Z_E(R)$ and $N=\overline{Z_E}(R)$ are separated by a hyperplane.
We provide an equivalent characterization of fully convex rings
via a special type of monomial threshold labelling. 
A \emph{linear} threshold labelling of
$R$ is a threshold labelling $L$ with some threshold $T$ such that
for any $u, v \in R$ with $u, v \neq 0$ we have $L(uv)=L(u)+L(v)$.
It is \emph{monomial} if $L$ is furthermore a monomial threshold labelling.
Note that a monomial linear threshold labelling $L$ is defined by the value $L(\rx_i)$ for every $i \in [n]$,
since it will have $L(\rx_1^{\alpha_1} \dots \rx_n^{\alpha_n}) = \sum_{i=1}^n \alpha_i L(\rx_i)$.

\begin{example} \label{ex:MTL-example}
  Consider the ring $\integers_{2}[\rx,\ry]/(\rx^2, \rx\ry, \ry^2)$. 
  Let $U=\{1,1+\rx,1+\ry,1+\rx+\ry\}$ denote the set of units and $N_0=\{\rx,\ry,\rx+\ry\}$
  the set of non-zero non-units. 
  A
  threshold labelling $L,T$ with $T=2$ can be obtained by assigning the label $0$ to
  the elements in $U$, assigning
  the label $1$ to the elements in $N_0$, and assigning the label $2$ to $0$.
  This corresponds to letting $L(\rx)=L(\ry)=1$ so this is a linear monomial
  threshold labelling.
\end{example}

We show that a monomial ring $R$ is fully convex if and only if it admits a linear monomial threshold labelling. 
To more strongly illustrate the geometric interpretation of this, we also give an equivalent geometric description. 
We need the following result, which is a common formulation of
Minkowski's hyperplane separation theorem, see Corollary~11.4.2 in \cite{Rockafellar:CA}.

\begin{theorem} \label{thm:separating-hyperplane}
Let $C$ and $D$ be two closed convex sets in ${\mathbb R}^n$
with at least one of them bounded,
and assume $C \cap D = \emptyset$. Then, there exists $0 \neq a \in R^n$
and $b \in {\mathbb R}^n$
such that
$a^Tx > b$ for all $x \in D$ and
$a^Tx < b$ for all $x \in C$.
\end{theorem}

\begin{lemma} \label{lm:characterize-fullyconvex}
    Let $R=\FF[\rx_1,\ldots,\rx_n]/I$ be a finite commutative monomial ring, let $Z=Z_E(R)$ and $N=\overline{Z_E}(R)$. The following are equivalent. 
    \begin{enumerate}
        \item $R$ is fully convex, i.e.\ $(N,Z)$ has a separating hyperplane
        \item The convex closures of $N$ and $Z$ are disjoint
        \item $R$ has a linear monomial threshold labelling
    \end{enumerate}
\end{lemma}
\begin{proof}
    First, assume that $(N,Z)$ has a separating hyperplane $H$. Recall that we assume that $H$ 
    is disjoint from all integer points. Then the convex closure of $N$ is contained in one side of the half-plane given by $H$, and the convex closure of $Z$ in the other, and the containment is strict in that both are disjoint from $H$. Thus their convex closures are disjoint. Furthermore, 
    let $H$ be defined by an equation $\sum_i a_i \rx_i = b$ for some coefficients $a_i \in \RR$.
    We note that $a_i > 0$ for each $i$, since otherwise $N$ would contain unboundedly large powers $\rx_i^d$ contradicting that $R$ is finite. Since by assumption $H$ intersects no integer points, 
    we can use $H$ for a linear monomial threshold labelling $L_m(\rx_1^{d_1} \dots \rx_n^{d_n})=\sum_i a_id_i$ with threshold $T=b$. In the other direction, by the same argument every linear monomial threshold labelling defines a separating half-plane: In particular, let $f(a)=L(\rx^a)$ where $L$ is a linear monomial threshold labelling, and let $T$ be the threshold.
    Then there is a positive value $\eps > 0$ such that for every $p \in N$ we have $f(p) \leq T-\eps$ and 
    for every $p \in Z$ we have $f(p) \geq T$. Then the half-plane defined by $f(p)=T-\eps/2$ is disjoint from integer points in $\naturals^n$. Finally, if $N$ and $Z$ have disjoint convex closures, then there is a separating hyperplane by Theorem~\ref{thm:separating-hyperplane} since $N$ is finite and $cnv(N)$ is consequently bounded.
    This hyperplane is furthermore disjoint from all points in $\naturals^n$ since $N \cup Z = \naturals^n$ and Theorem~\ref{thm:separating-hyperplane} guarantees that the separation is strict.
\end{proof}

Once again, consult Appendix~\ref{sec:counter:double-hole-free} for an illustration that the condition that the condition on convex closures cannot be strengthened to saying that $N$ and $Z$ are hole-free (and naturally disjoint).

\subsubsection{Fully Convex Implies Bergen}
\label{sec:fullyconvex-bergen}

We now show that fully convex rings are Bergen. 
Let us first show how a threshold labelling for a ring $R$ induces an equivalence relation on the elements of $R$. 
For a labelling $L$ we denote by $L(R)$ the set of labels used by
$L$, i.e. $L(R)=\SB L(v) \SM v \in R\SE$.

\begin{definition} \label{def:ntl-eqcl}
  Let $R$ a finite, lineal, commutative ring with threshold
  labelling $L$ with threshold $T$. For a label $i \in L(R)$ with $i<T$, define $i^+>i$
  as the next smallest label in $L(R)$, i.e.
  $i^+$ is the smallest value from $L(R)$ with $i^+>i$. We define an equivalence relation
  by letting $u \equiv_L v$ for $u, v \in R$ if and only if
  $L(u)=L(v)$ and $u = v \bmod {L_{L(v)^+}}$, where $L_i=\{v \in
  R \mid L(v) \geq i\}$. 
\end{definition}

\begin{example} \label{ex:equi-classes}
  Recall from Example~\ref{ex:MTL-example} that the ring $\integers_{2}[\rx,\ry]/(\rx^2, \rx\ry, \ry^2)$, 
has a monomial threshold labelling such that $L(\rx)=L(\ry)=1$ and $T=2$.
  This labelling gives us
  the equivalence relation $\equiv_L$ with equivalence classes
$\{0\}$, $\{\rx\}$, $\{\ry\}$, $\{\rx + \ry\}$, and $U$.
\end{example}

\begin{lemma} \label{lm:fully-convex-to-bergen}
  Let $R$ be a finite and commutative ring that is fully convex. Then, $R$ is a Bergen ring.
  \end{lemma}
  \begin{proof}
  Let $R=\FF[\rx_1,\ldots,\rx_n]/I$. By Lemma~\ref{lm:characterize-fullyconvex}
  $R$ has a linear monomial threshold labelling. Let $L$ be such a labelling and let $T$ be its threshold.
  Since $L$ is a monomial labelling, we have
  \[
    L(r) = \min_{m \in r} L(m)
  \]
  for all $r \in R$, where $m$ ranges over all monomials occurring in $r$.
  In particular, for any $p, q \in R$ note that $L(p+q) \geq \min(L(p), L(q))$.
  We verify that $\equiv_L$ has the matching property.

  \begin{claim} \label{claim:equiv}
    Let $u, v, a \in R$ be such that $au \neq 0$ and $av \neq 0$.
    Then, $u \equiv_L v$ if and only if $au \equiv_L av$.
    \end{claim}
    \begin{claimproof}
  On the one hand, suppose $u \equiv_L v$. Then by assumption there exists an element $m \in R$ such that
    $L(m) > L(u)=L(v)$ and $u=v+m$. Now $au=a(v+m)=av+am$ and $L(av)=L(a)+L(v) < L(a)+L(m)$ and $L(am)=\min(T,L(a)+L(m))$.
    Since $av \neq 0$ we have $L(a)+L(v) < T$ and we conclude $au \equiv_L av$.
    On the other hand, suppose that $au \equiv_L av$. This implies $L(au)=L(av)$ and that there exists $m \in R$ such that
    $L(m) > L(au)$ and $au=av+m$. Then $a(u-v)=m$ so 
    \[L(m)=L(a)+L(u-v) > L(a)+L(u)=L(a)+L(v).\] Thus $L(u-v) > L(u)=L(v)$
    and we get $u \equiv_L v$.
\end{claimproof}

We are now ready to prove that $R$ is Bergen.
We begin by constructing a suitable chain of ideals.
Consider the set $L_i=\{v \in R \; | \; L(v) \geq i\}$ for some $0 \leq i \leq T$.
If $v,w \in L_i$, then $L(v) \geq i$ and $L(w) \geq i$.
Since $L$ is a monomial labelling, we have $v \in L_i$ if and only if
$L(m) \geq i$ for every monomial $m$ in $v$, i.e.\ the set $L_i$
is a monomial ideal generated by all monomials $m$ such that $L(m) \geq i$.
Let $R=I_0 \supset I_1 \supset \dots \supset I_{\ell}=\{0\}$
be a chain of minimal length such that every $L_i$, $0 \leq i \leq T$, appears in the chain.

We continue by equipping each ideal in the chain with a equivalence relation.
Let $\equiv_i$ denote the equivalence relation $\equiv_L$ restricted to $I_i$.
Claim~\ref{claim:equiv} shows that $\equiv_i$ has the
matching property.
Arbitrarily choose $u,v \in I_i$, $1 \leq i < \ell$, 
such that $u \equiv_i v$. We consider
the element
$u-v$. 
Recall that $u \equiv_L v$ if and only if
  $L(u)=L(v)$ and $u = v \bmod {L_{L(v)^+}}$.
  It follows that $u-v \in I_{i+1}$.
  Since finally all equivalence classes take the form of $u+L_{L(u)^+}$,
  the classes of $\equiv_i$ are cosets for each $i$, and $R$ is a Bergen ring.
\end{proof}

\subsubsection{Separating Examples}
\label{sec:mtl-not-bergen}

We will now present examples of finite commutative rings that are Bergen (and thus lineal by Proposition~\ref{prop:bergenislineal})
but not fully convex, and in one case not even Z-hole-free.

The first is based on work in discrete convex analysis by Kashimura, Numata, and Takemura~\cite[Example~4.7]{KashimuraNT13}.
Define $f(x,y,z)=12x+15y+20z$ and consider the ring $R_{\rm KNT}=\mathbb{Z}_2[\rx,\ry,\rz]/(I',\rx^2\ry\rz)$
where $I'$ is the ideal generated by all monomials $\rx^a \ry^b \rz^c$ such that $f(a,b,c) > 60$.

\begin{lemma}
    $R_{\rm KNT}$ is not fully convex. In particular $N=\overline{Z_E}(R)$ is not hole-free.
\end{lemma}
\begin{proof}
    The affine hyperplane spanned by $(5,0,0)$, $(0,4,0)$, $(0,0,3)$ has the equation $f(x,y,z)=60$,
    and the point $(2,1,1)$ has value $f(2,1,1)=59$. As the former are non-zero and the latter is zero,
    there is no hyperplane that separates $N$ from $Z_E(R)$. For the same reason, $N$ is not hole-free
    since $(2,1,1) \in cnv(N)$. 
\end{proof}

Verifying that $R_{\rm KNT}$ is Bergen is straightforward but tedious. 
We defer the construction to Appendix~\ref{sec:counter:KNT}.

An example in the other direction is as follows. 
Define the ring $R_{347}$ as the monomial ring $R_{347}=\FF_2[\rx,\ry,\rz]/I$ 
where $I$ contains every monomial in the upper half-plane spanned by 
$\rx^3$, $\ry^4$, $\rz^7$, inclusive, except that $\rx\ry\rz^3$ is
non-zero. Then $Z_E(I)$ is not hole-free, and $R_{347}$ is not fully convex.

\begin{lemma}
    $R_{347}$ is not $Z$-hole-free (and in particular, not fully convex).
\end{lemma}
\begin{proof}
    The points $(3,0,0)$, $(0,4,0)$ and $(0,0,7)$ span a hyperplane $H$ with
    equation $f(x,y,z)=28x+21y+12z=84$, and the upper halfplane of $H$ is 
    contained in $cnv(Z_E(R))$, but $f(1,1,3)=85$ and $(1,1,3) \notin Z_E(R)$. 
    Thus $(1,1,3) \in cnv(Z_E(R)) \setminus Z_E(R)$, for example as
    \[
    \frac{28}{84}  \cdot (3,0,0) + \frac{21}{84} \cdot (0,4,0) + \frac{28}{84} \cdot (0,0,7) + \frac{7}{84} \cdot (0,0,8)=(1,1,3).
    \]
    Thus $R_{347}$ is not $Z$-hole-free, and not fully convex by Lemma~\ref{lm:characterize-fullyconvex}.
\end{proof}

However, again, $R_{347}$ is Bergen as shown in Appendix~\ref{sec:counter:347}.
Together, $R_{\rm KNT}$ and $R_{347}$ are complementary examples showing that 
for a monomial lineal ring $R$, we cannot assume that $Z=Z_E(R)$ or $N=\overline{Z_E}(R)$ is hole-free.

\subsection{Helly Rings}
\label{sec:convex}

We now study monomial Helly rings. Although we do not have a precise characterization of this property in terms of discrete convexity, we show via a chain of implications that every Z-hole-free ring is Helly, and every monomial Helly ring is Z-1-convex (see~Figure~\ref{figure:link}).

Let $R=\FF[\rx_1,\ldots,\rx_n]/I$ be a monomial ring. 
Recall that $R$ is {\em Z-hole-free} if $Z_E(R)$ is hole-free, i.e.\ $Exp(I) = cnv(Exp(I)) \cap {\mathbb Z}^n$.
We first show that all Z-hole-free rings have the no cancellations property. 
 
\begin{lemma} \label{lm:zhf-to-nc}
  Z-hole-free rings have the no cancellations property.  
\end{lemma}
\begin{proof}
  Let $R=\FF[\rx_1,\ldots,\rx_n]/I$ be a $Z$-hole-free monomial ring
  and consider a zero product $pq=0$. We will show that $mm'=0$ for every
  monomial $m$ in $p$ and $m'$ in $q$.
  Let $M_p$ be the set of monomials occurring in $p$ and $M_q$ the set of monomials
  occurring in $q$, and let $E_p=Exp(M_p)$ and $E_q=Exp(M_q)$. 
  Let $Q=E_p+E_q$ be their Minkowski sum, i.e. $Q$ contains all vectors $a$
  such that the full expansion of the product $p \cdot q$ produces a monomial $X^a$, 
  possibly with net coefficient zero, where we may or may not have $X^a=0$ in $R$.
  Let $P=cnv(Q)$. Then every extreme point of $P$ is contained in $Q$ and is produced in a unique way $u=v+w$
  in $E_p+E_q$.
  Indeed, let $w \in Q$ be produced as $w=u_1+v_1=u_2+v_2$ for
  $u_1, u_2 \in E_p$ and $v_1, v_2 \in E_q$ in two distinct ways.
  Let $u=(u_1+u_2)/2$ and $v=(v_1+v_2)/2$. Then $u \in cnv(E_p)$,
  $v \in cnv(E_q)$, $u+v=w$ and both $u$ and $v$ are contained in
  line segments of positive length in the respective polytope.
  Since $cnv(E_p+E_q)=cnv(E_p)+cnv(E_q)$,
  this contradicts that $w$ is an extreme point.

  But since $pq=0$ and a single contribution $(aX^u \cdot bX^v = abX^{u+v})$
  in $p \cdot q$ can never produce a coefficient of zero, for every such point $w=u+v$
  produced in only one way we must have $X^w=0$ in $R$. Now, every other point $w \in Q$ is produced as a 
  convex combination of the extreme points of $P$, so since $Z_E(R)$ is hole-free
  every point $w \in Q$ satisfies $X^w=0$. Thus for every monomial $m=X^u$ in $p$
  and $m'=X^v$ in $q$ we have $mm'=0$ and $pq=0$ is a trivial product,
  as claimed.
\end{proof}

Next, we show that the no cancellations property implies that the ring is Helly. 

\begin{lemma} \label{lm:no-cancel-helly}
  Let $R$ be a finite ring with the no cancellations property.
  Then $R$ is Helly.
\end{lemma}
\begin{proof}
Consider a one-element annihilator coset $a+\Ann(b)$ for $a, b \in R$ and let $M_b$ be
the set of monomials of $R$ that annihilate $b$. We note that the no
cancellations property implies that $\Ann(b)=(M_b)$: Indeed,
clearly $(M_b) \subseteq \Ann(b)$, so assume that there is an element
$r \in R$ such that $rb=0$ but $r \notin (M_b)$. But then $r$ contains
at least one monomial $m \notin M_b$, i.e. $mr \neq 0$,
and there is a monomial $m'$ in $b$ such that $mm' \neq 0$,
which contradicts $rb=0$ by the no cancellations property. 

Thus for any $a, b \in R$, the one-element annihilator coset
$a+\Ann(b)$ is described as a ``monomial coefficient filter'': 
$a+\Ann(b)$ contains precisely those elements $r \in R$ 
such that $r$ and $a$ have identical coefficients for every monomial
$m$ in $R$ where $m \notin M_b$.

Now consider three one-annihilator cosets $a_i+\Ann(b_i)$,
$i=1, 2, 3$ and assume that every pair of them have a non-empty
intersection. Let $m$ be an arbitrary monomial in $R$,
and let $m_i=*$ if $m \in M_{b_i}$, and otherwise let
$m_i$ be the coefficient of $m$ in $a_i$.
If there are two values $i \neq j$, $i, j \in [3]$
such that $m_i \neq m_j$ and neither equals $*$,
then $a_i+\Ann(b_i)$ and $a_j+\Ann(b_j)$ have an empty intersection. 
If this does not occur, then let $r$ be the element
where for every monomial $m$, the coefficient of $m$
is zero if $m_1=m_2=m_3=*$ and otherwise the coefficient
is the unique non-$*$ value among $m_1$, $m_2$ and $m_3$. 
Then $r \in a_i+\Ann(b_i)$ for each $i=1, 2, 3$.
\end{proof}

Finally, all monomial Helly rings are Z-1-convex.

\begin{lemma} \label{lm:helly-to-z1c}
  Let $R$ be a monomial ring and $Z=Z_E(R)$ the exponent set of its zero set.
  If $R$ is Helly, then $Z$ is 1-convex.
\end{lemma}
\begin{proof}
  Assume that $Z$ is not 1-convex, so that there are two points $p_1, p_2 \in Z$
  such that the line from $p_1$ to $p_2$ contains an integer point $q \notin Z$  
  Let $p_0=\min(p_1,p_2)$, applied coordinatewise, and by minimality
  assume that all integer points on the line from $p_1$ to $p_2$
  are outside of $Z$. Enumerate the integer points on this line as
  $p_1, q_1, \ldots, q_r, p_2$ for some $r \geq 1$ and let $d \in \ZZ^n$
  be the distance between any two adjacent points on the line.
  Split $d$ as $d=d^+-d^-$ where $d^+, d^- \in \ZZ_{\geq 0}^n$.
  Thus $p_1=p_0+(r+1)d^-$ and $p_2=p_0+(r+1)d^+$. 
  Write $m_1=X^{d^+}$ and $m_2=X^{d^-}$. Now observe
  \[
    X^{p_0}(m_1^r+m_1^{r-1}m_2+\dots+m_2^r)(m_1-m_2)= X^{p_1}-X^{p_2}=0.
  \]
  We claim that $A_1=A(X^{p_0}(m_1^r+\dots+m_2^r))$, $A_2=A(X^{p_0}m_1^r)$
  and $A_3=A(X^{p_0}m_2^r)$ form a non-distributive triple of annihilators.
  Indeed, these contain $m_1-m_2$, $m_1$, and $m_2$, respectively,
  thus $m_1 \in (A_1+A_3) \cap (A_2+A_3)$ but
  $m_1 \notin A_3 \supseteq (A_1 + A_2) \cap A_3$. 
  Thus $R$ is non-Helly.
\end{proof}

\subsection{Related Notions}
\label{sec:geometry-other}

Let us now review connections between lineal and Helly (monomial) rings and other properties studied in the literature on ring theory. 

Recall that a monomial ring $R=\FF[\rx_1,\ldots,\rx_n]/I$ is Z-hole-free if and only if $Z_E(R)=Exp(I)$ is a hole-free set (i.e. it contains all integer points in its own convex hull). 
This has an equivalent characterization in term of {\em integrality closures}.
This is an important concept in ring theory, see e.g. the book of Huneke and Swanson~\cite{HunekeS2006integral}.
Let $R$ denote a commutative ring and $I$ an ideal of $R$.
The {\em integral closure}
of $I$ (denoted by $\overline{I}$) is the set of all elements $r \in R$ that are integral over $I$, i.e.
there exist $a_i \in I^i$ such that
\[r^n+a_1r^{n-1}+a_2r^{n-2}+ \dots + a_{n-1}r+a_n=0.\]
If $I=\overline{I}$, then we say that $I$ is {\em integrally closed}.
We have the following result.

\begin{proposition}[Proposition 1.4.6 in \cite{HunekeS2006integral}]
The exponent set of the integral closure of a monomial
ideal $I$ equals all the integer lattice points in the
convex hull of the exponent set of $I$.
\end{proposition}

We see that the monomial ring ${\mathbb F}[\rx_1,\dots,\rx_n]/I$
is Z-hole-free if and only if $I$ is integrally closed.
While we are not aware of whether there is a connection between Helly rings and integrally closed ideals in general, i.e. beyond monomial rings, we note an example of a ring that is lineal but not monomial, from Huneke and Swanson~\cite[Example~1.3.3]{HunekeS2006integral}.
Let $\mathbb{F}$ be any finite field and define $R=\mathbb{F}[\rx,\ry]/(\rx^2+\ry^3,\rx\ry^3,\ry^4)$.
Then $R$ is clearly not monomial, but it is lineal; see Appendix~\ref{sec:counter:huneke-swanson}.
Moreover, the ideal $I=(\rx^2+\ry^3,\rx\ry^3,\ry^4)$ is integrally closed.

Another connection goes to the structure of the lattice of ideals of $R$. 
Since the connection is subtle, let us define our terms carefully.
Recall that a \emph{lattice} (in the sense of an algebraic structure) 
is a partially ordered set where every pair of elements
$x$ and $y$ have a unique least upper bound $x \lor y$
and greatest lower bound $x \land y$. 
A lattice $L$ over a set $S$ can be defined either via its partial order,
or by the operations representing $\lor$ (join) and $\land$ (meet).
A lattice $L$ is \emph{distributive} if
\[
  x \land (y \lor z) = (x \land y) \lor (x \land z)
\]
holds for all elements $x, y, z \in L$.
Distributivity is a common and natural structural property of lattices,
e.g.\ it is equivalent to $L$ being representable as a set family,
with join operation $\cup$ and meet operation $\cap$.

There are two distinct lattices that are commonly associated with rings.
The first, most important is the lattice of ideals. 
Let $R$ be a commutative ring. Then the set of all ideals of $R$
forms a lattice under the subset partial order. Denote this lattice $L_R$.
The join and meet operations of $L_R$ are $I \lor J = I+J$
and $I \land J = I \cap J$, respectively.
Let $L_R'$ be the sublattice of $L_R$ generated by annihilators. 

The second is the lattice of annihilators. 
That is, the set of all annihilators of $R$, partially ordered by the subset relation,
also forms a lattice, although it is distinct from $L_R'$
in that the meet and join operations are distinct -- that is,
if $I$ and $J$ are annihilator ideals, then normally $I+J$
is not an annihilator ideal (see~\cite{BGT15-annihilatorsums,BGT18-annihilatorsums,DubeT2021lattice}).

We find a close connection between $L_R'$ being distributive and $R$ being Helly.
We note the following. (In fact, the second part can be generalized to an 
equivalent characterization of Helly rings -- that is, a commutative ring $R$, 
not necessarily monomial, is Helly if and only if the distributive property
holds for all triples of one-element annihilators in $L_R'$. We omit the details.)

\begin{lemma}
  Let $R$ be a monomial ring. If $R$ has the no cancellations property,
  then the lattice $L_R'$ is distributive. On the other hand,
  if $R$ is not Helly, then $L_R'$ is non-distributive. 
\end{lemma}
\begin{proof}
  By Lemma~\ref{lm:nc-gives-monomial-ideals}, every annihilator of $R$ is monomial.
  Let $I$ and $J$ be monomial ideals in $R$, generated by
  monomials $M_I$ and $M_J$, respectively. Then it is easy to see
  that $I+J$ is the ideal generated by $M_I \cup M_J$,
  and $I \cap J$ is the ideal generated by $M_I \cap M_J$.
  Thus $L_R'$ is isomorphic to a lattice whose elements are sets of monomials
  and where the meet and join operations are $\cap$ and $\cup$.
  Then $L_R'$ is distributive.  

  On the other hand, assume that $R$ is not Helly.
By Lemma~\ref{lem:homogenised-tangle} $R$ has a tangle of cosets
$\Ann(a)$, $\Ann(b)$ and $c+\Ann(d)$ for $a, b, c, d \in R$.
We claim that $c \in \Ann(a)+\Ann(d)$,
$c \in \Ann(b)+\Ann(d)$ and
$c \notin (\Ann(a) \cap \Ann(b)) + \Ann(d)$.
Indeed, let $r \in \Ann(a) \cap (c+\Ann(d))$, which exists by assumption.
Then $r \in \Ann(a)$ and $r-c \in \Ann(d)$,
hence $c \in \Ann(a)+\Ann(d)$. By the same argument,
$c \in \Ann(b)+\Ann(d)$. 
Finally, assume that $c=p+q$ where $p \in \Ann(a) \cap \Ann(b)$
and $q \in \Ann(d)$. 
But then $p=c-q$ is contained in all three cosets of the tangle,
a contradiction.
Thus $c \in (\Ann(a) \cap \Ann(b)) + \Ann(d)$.
Thus $L_R'$ is not distributive on
the triple $\Ann(a)$, $\Ann(b)$, $\Ann(d)$.
\end{proof}

In particular, it may hold that a finite commutative ring $R$ 
is Helly if and only if the lattice $L_R'$ is distributive. 
This would be a natural direct generalization of the easily verified property
that $R$ is lineal if and only if $L_R'$ is a chain. 
We note that the lattice of annihilators is not necessarily distributive. 
Let $R=\FF[\rx,\ry,\rz]/I$ where $I$ contains $\rx\ry\rz$ and all monomials of degree at least 4. Then $R$ is Z-hole-free,
but the lattice of annihilators of $R$ is non-distributive; see Appendix~\ref{sec:counter:annihilator-lattice}.
Furthermore, $L_R$ as a whole is not distributive; only $L_R'$ is. 

As referenced above, see Birkenmeier, Ghirati, and Taherifar~\cite{BGT15-annihilatorsums,BGT18-annihilatorsums}
and Dube and Taherifar~\cite{DubeT2021lattice} for related investigations into
the lattice structure of  rings.
Furthermore, a commutative ring $R$ where the entire lattice $L_R$ is distributive
is known as \emph{arithmetical};
however, this is very restrictive, as a local arithmetical ring is precisely a chain ring~\cite{jensen1966arithmetical}.

Furthermore, the structure of $L_R'$ has intriguing connections to the algebraic approach to CSP study (see \cite{Barto:etal:polymorphisms} and the discussion in Section~\ref{sec:discussion}).
In particular, it is possible to show that if the language of equations $u=r$ and $u=rv$, $r \in R$
over a ring $R$ has a so-called \emph{majority} 
polymorphism $m \colon R^3 \to R$,
then annihilator cosets $p+\Ann(q)$ must be \emph{absorbing} for $m$, i.e.\ if 
at least two elements of $a, b, c$ are contained in such a coset
then so must $m(a,b,c)$ be. Furthermore, such absorbing sets are closed under intersection,
i.e.\ they form a semilattice. Since the intersection of (non-disjoint) annihilator cosets
is another annihilator coset, the lattice structure related to annihilators of $R$ 
appears closely connected to such algebraic conditions.

\section{Discussion}
\label{sec:discussion}

Our long-term goal is to understand the parameterized complexity
(with parameter $k$ being the number of unsatisfied equations)
of the $\minlin{r}{R}$ problem for finite commutative rings.
To this end, we have made a number of contributions.
We have proved that $\minlin{2}{R}$
is FPT-approximable within a constant  when $R$ is a Bergen ring.
We additionally proved that $\minlin{r}{R}$
is not FPT approximable within any constant (under the assumption that \FPT $\neq$ \W{1}) when $r \geq 3$ and this result
holds for all rings $R$.
We have demonstrated that there are broad classes of finite commutative rings
such that $\minlin{2}{R}$ is not FPT approximable within certain bounds:
(1) if $R$ is not Helly, then $\minlin{2}{R}$ is not FPT approximable within
any constant (unless \FPT $=$ \W{1}) and (2)
if $R$ is not lineal, then $\minlin{2}{R}$ is not FPT approximable within
$2-\eps$, $\eps > 0$ (unless the \ETH is false). We discuss a selection of possible research directions below.

We know that \minlin{2}{R} is FPT-approximable within a constant
whenever $R$ is a Bergen ring. 
Proposition~\ref{prop:sumapproximation} implies that there are finite and commutative rings $R$ 
such that \minlin{2}{R} is FPT-approximable within a constant
but $R$ is not Bergen. A simple example is the direct sum $R_1 \oplus R_2$ of two
Bergen rings $R_1$ and $R_2$: $R_1 \oplus R_2$ is not local so it is not lineal by Proposition~\ref{prop:lineal-is-local} and
consequently not Bergen by Proposition~\ref{prop:bergenislineal}.
We will now give a more interesting example.
The characterization of solution sets from Section~\ref{sec:annihilators} may suggest that lineality is a necessary condition for FPT
approximability within a constant but this example shows that this in not
true.
Our starting point is the ring
$R_{\rm APX}=\ZZ_2[\rx,\ry] / I$ 
with $I=(\rx^3,\rx\ry,\ry^3)$ that we used as an example in~Section~\ref{sec:annihilators}.
A routine verification shows that $R_{\rm APX}$ is local.
We know that $R_{\rm APX}$ is not lineal since $\Ann(\rx)$ and $\Ann(\ry)$ are incomparable sets:
$\rx \not\in \Ann(\rx)$, $\rx \in \Ann(\ry)$,
$\ry \in \Ann(\rx)$, and $\ry \not\in \Ann(\ry)$.
Hence, Proposition~\ref{prop:bergenislineal} implies that
$R_{\rm APX}$ is not Bergen. 
Theorem~\ref{thm:incomparable-annihilators} implies that
\minlin{2}{R_{\rm APX}} is not FPT-approximable within $2-\eps$
for any $\eps > 0$ so \minlin{2}{R_{\rm APX}} is not in \FPT.

We claim that $\minlin{2}{R_{\rm APX}}$ is constant factor FPT-approximable.
The algorithmic idea is inspired by {\em fiber products} (also known
as {\em pullbacks} or {\em Cartesian squares}).
One may view the approach as a way of decomposing rings
that is radically different from the Bergen-style decomposition
along a chain of ideals.
Let $A,B,C$ be three finite commutative rings
and $f : A \to C$, $g : B \to C$ be two ring homomorphisms.
Define the \emph{fiber product of $f$ and $g$} as  
\[
  \Pi(f,g) := \{ (a,b) \in A \times B : f(a) = g(b). \}
\]
Observe that if $C$ is the trivial ring $\{0\}$,
then $\Pi(f,g)$ is the direct sum of $A$ and $B$.
The ring
$R_{\rm APX} = \ZZ_2[\rx,\ry] / (\rx^3,\rx\ry,\ry^3)$
can be described as
a fiber product defined by
\begin{itemize}
  \item $A = \ZZ_2[\rx] / (\rx^3)$,
  \item $B = \ZZ_2[\ry] / (\ry^3)$,
  \item $C = \ZZ_2$, and 
  \item unit projections $f : A \to C$, $g : B \to C$ (i.e. $f(u)=1$ if and only if $u$ is a unit element).
\end{itemize}
Put another way, $f(r) = r \bmod (\rx)$ and 
$g(r) = r \bmod (\ry)$ and one can verify that
$R = \{ (a,b) \in A \times B : f(a) = f(b) \}$.
This leads to the following algorithmic idea:
view the values in $R$ as pairs $(a, b)$,
write two systems of equations (one over $A$
and another over $B$) and enforce coordination
constraints $f(a) = f(b)$.

More precisely, let $I$ be an instance of $\minlin{2}{R_{\rm APX}}$.
Note that $A$ and $B$ are Bergen (in fact, they are chain rings),
and use ideal chains $(1) \supset (\rx) \supset (\rx^2) \supset (0)$
and $(1) \supset (\ry) \supset (\ry^2) \supset (0)$, respectively.
We will only need to enforce the coordination constraints
on the first level since this is where 
the decision whether a variable will be assigned
unit or non-unit value is made.
Let us use partition with classes 
$\{0\}$, $1 + (\rx)$, $\rx + (\rx^2)$ and $\{\rx^2\}$ in $A$ and,
analogously,
$\{0\}$, $1 + (\ry)$, $\ry + (\ry^2)$ and $\{\ry^2\}$ in $B$.
Construct one graph $G^A = G(S^A)$, where $S^A$ are the equations of $S$ modulo $(\ry)$,
one graph $G^B = G(S^B)$, where $S^B$ are the equations of $S$ modulo $(\rx)$.
Identify the $s$- and $t$-vertices in $G^A$ and $G^B$.
For a variable $v \in V(S)$, we now have three vertices
$v^A_{1+(\rx)}$, $v^A_{\rx+(\rx^2)}$, $v^A_{\rx^2}$ in $G^A$ and
three vertices
$v^B_{1+(\ry)}$, $v^B_{\ry+(\ry^2)}$, $v^B_{\ry^2}$ in $G^B$.
Add undeletable coordination edges $v^A_{1+(\rx)} v^B_{1+(\ry)}$ and
note that they encode constraint ``$v^A$ is a unit if and only if $v^B$ is a unit''.
Run the remaining steps of the algorithm
(shadow removal, branching, reduction to the next ideal, etc.) unchanged.
Let $\alpha : V(S^A) \to A$ and $\beta : V(S^B) \to B$ be the assignments
produced by the algorithm.
Note that the coordination edges force $f(\alpha(v^A)) = g(\beta(v^B))$
for all $v \in V(S)$, so given $\alpha$ and $\beta$,
we can define an assignment $\phi : V(S) \to R_{\rm APX}$
as $\phi(v) = (\alpha(v^A), \beta(v^B))$.
Moreover, if both $\alpha$ and $\beta$ satisfy
an equation modulo $(\ry)$ and $(\rx)$, respectively,
then $\phi$ also satisfies the equation,
so $\cost_S(\phi) \leq \cost_{S^A}(\alpha) + \cost_{S^B}(\beta)$,
and we obtain a constant-factor approximation.

We recall from Section~\ref{sec:annihilators} that the ring
$R_{\rm APX}$ is Helly. For the moment, we have {\em no} examples of Helly rings that
are not FPT-approximable within some constant.
If ${\cal A}$ is the class of rings that are FPT approximable within some
constant, then we know (by Proposition~\ref{prop:sumapproximation}(2)) that ${\cal A}$ must be closed under
taking finite direct sums.
Helly rings are closed under finite direct sums by
Proposition~\ref{prop:helly-direct-sum}.
This motivates the following question: does ${\cal A}$ equal the set of Helly rings?
If this is indeed the case, then Theorem~\ref{thm:non-helly-hard} gives us a dichotomy under the assumption \FPT $\neq$ \W{1}: $\minlin{2}{R}$
is FPT-approximable within a constant if and only if $R$ is Helly.
The geometric approach suggests that it may be a good idea to initially restrict
the attention to monomial rings. One concrete question here is then: are all Z-hole-free rings (such as $R_{\rm APX}$) FPT-approximable within some constant?

We conclude this discussion section with a few words concerning methodology.
The universal-algebraic approach has proved to be a powerful tool in the study of the computational complexity of CSPs. 
We utilize this approach very little in this paper: the only substantial use
is in Section~\ref{sec:solving-equations-with-cosets} when proving tractability of extensions of the $\lin{2}{R}$ problem.
We note, however, that it has sometimes been used in the background: for instance,
the definition of tangles and Helly rings came out of considerations about necessary properties of majority polymorphisms --- a majority polymorphism implies constant-factor FPT approximability for Boolean {\sc MinCSP}~\cite{Bonnet:etal:esa2016}, and more generally excludes the ability to encode non-binary constraints,
such as used in Section~\ref{sec:inapprox-non-helly}.
The most obvious technical reason for us not using the algebraic approach
is that constraint languages based on linear equations over rings are not
very well-understood from a universal-algebraic point of view.
Thus, we found it easier to work directly with tools from commutative algebra and discrete convex analysis.
This is in stark contrast to Boolean languages where Post's lattice
provides a crystal-clear picture, and this explains why
the proof of the FPT dichotomy for Boolean MinCSP could use the algebraic approach
with relative ease.

One possible way of exploiting the algebraic approach is inspired by Dalmau, Krokhin, and Manokaran's~\cite{Dalmau:etal:soda2015} ideas for analysing constant-factor {\em polynomial-time} approximability of $\mincsp{\Gamma}$ over finite domains. They use the algebraic approach for
characterizing $\mincsp{\Gamma}$ for which the basic LP relaxation has ﬁnite integrality gap, and the condition they study is in terms of probability distributions based on certain polymorphisms of $\Gamma$. 
Dalmau, Krokhin, and Manokaran demonstrate that this condition sometimes can
be used for designing polynomial-time constant-factor approximation algorithms. 
The connections with constant-factor FPT approximability are unclear but they are
certainly worth studying. This is encouraged by the fact that the basic LP relaxation
(and other relaxations in the Sherali-Adams hierarchy)
has been important ingredients in many studies of CSP-based optimization problems~\cite{Dalmau:etal:jcss2018,Kolmogorov:etal:sicomp2017,Thapper:Zivny:jacm2016,Thapper:Zivny:sicomp2017}.
Dalmau, Krokhin, and Man also point out that
constant-factor approximation algorithms for MinCSP are closely related to
{\em robust approximation}. 
Robust approximation has been studied extensively in the CSP literature~\cite{Barto:Kozik:sicomp2016,dalmau2013robust,Guruswami:Zhou:tc2012,kun2012linear}.
Given an instance with at most an $\epsilon$-fraction of unsatisfiable constraints, a robust approximation algorithm seeks an assignment satisfying a $(1-g(\epsilon))$-fraction of the constraints, where $g$ satisfies $g(0) = 0$ and $\lim_{\epsilon \rightarrow 0} g(\epsilon) = 0$.  
In particular, CSP$(\Gamma)$ admits a robust algorithm with
so-called {\em linear loss}~\cite{dalmau2013robust} if and only if 
$\mincsp{\Gamma}$ has a constant-factor polynomial-time approximation algorithm.
Once again, the connections with constant-factor FPT approximability are unclear but highly interesting.

 \section*{Acknowledgements}

 We thank Karthik C.S. for pointing out an issue with an earlier version of the proof of Theorem~\ref{thm:gap-paired-cut}.
 The second and the fourth authors were partially supported by
 the Wallenberg AI, Autonomous Systems and Software Program (WASP) funded
 by the Knut and Alice Wallenberg Foundation. 
 In addition, the second author was partially supported by 
 the Swedish Research Council (VR)
 under grant 2021-04371, 
 and the fourth author was partially supported by 
 VR
 under grant 2024-00274.

\newpage

\appendix
\section{Separating Examples}

This appendix contains examples for illustrating
various phenomena. Some of these examples have been found
using computer-assisted search. The properties of these
examples can, in principle, be verified by hand but this
may be a very tedious task.

\subsection{Lineal vs. Hole-free}
\label{sec:counter:double-hole-free}

The following example shows that for a monomial ring $R$, 
$Z_E(R)$ and $\overline{Z_E}(R)$ being hole-free, separately, does not imply that $R$ is lineal. 

\begin{lemma}
  Let $R=[\rx,\ry,\rz]/I$ be the monomial ring containing all
  monomials $\rx^a\rx^b\rz^c$ such that $12a+8b+6c \leq 24$ and in
  addition contains the monomials $\rx\ry\rz$, $\rx\ry^2$ and $\rx^2\ry$,
  and where every other monomial is contained in $I$. Then $R$ is
  non-lineal but both $N=\overline{Z_E}(R)$ and $Z=Z_E(R)$ is hole-free.
\end{lemma}
\begin{proof}
  It can be verified (for example by computer support) that $N$ and
  $Z$ are both hole-free. On the other hand, we have a counterexample
  to Condition~P in $\rx^2\ry \cdot \rz^4 = \rx^2\rx \cdot \ry\rz^3$,
  where the two former terms are non-zero and the two latter are zero.
\end{proof}

\subsection{Proving the Bergen Ring Property}
\label{ssec:knt-347}

In order to prove that the rings $R_{KNT}$ and $R_{347}$ are Bergen,
we first lay some foundation. Since the sizes of these rings are too
large to work with in a pure computer proof (e.g.\ they are large
enough that an iteration over all elements becomes prohibitive,
and an iteration over all pairs impossible),
we present a way to prove the Bergen property with significantly
reduced computational demands. 

Let $R$ be a monomial ring and let $M$ be the set of monomials of $R$
with the element $0$ added.
A \emph{weak ordering} $\preceq$ is a relation that is reflexive,
transitive and strongly connected, i.e. for all $x$ and $y$
we have $x \preceq y$ or $y \preceq x$. For a weak ordering $\preceq$, 
we write $x \prec y$ if $y \not \preceq x$ and note that $x \prec y$
is equivalent to $x \preceq y$ and $y \not \preceq x$.
Define a \emph{weak monomial order} over $M$ to be a weak ordering
$\preceq$ of $M$ such that the following hold: 
\begin{enumerate}
\item $1$ is the earliest element and $0$ is the last, i.e.\ if $m \in
  M$ and $m \notin \{0,1\}$ then $1 \prec m \prec 0$
\item $\preceq$ is preserved by multiplication, i.e.
  for any $p, q, r \in M$, if $p \preceq q$
  then $pr \preceq qr$
\end{enumerate}
Note that the strict version $\prec$ may be only weakly preserved by
multiplication, i.e.\ for $p, q, r \in R$ with $p \prec q$
we can have $pr \preceq qr$ and $qr \preceq pr$, so that $pr \prec qr$
does not hold. A weak monomial order is essentially a weakened version
of the notion of an \emph{admissible monomial order} from the theory
of Gröbner bases. While not every monomial  ring with a weak monomial
order is Bergen, properties of the weak monomial order can be used to
greatly reduce the effort of verifying that a ring is Bergen. 

A weak monomial order induces an ordered partition over $M$. Denote
this partition as
\[
  M_0=\{1\} \prec \dots \prec M_\ell=\{0\}.
\]
Define a function $f \colon M \to \{0,\ldots,\ell\}$
such that for every $m \in M$ we have $m \in M_{f(m)}$.
Note that $f$ is a well-defined total function.
We refer to the set of monomials $M_i$, $i \in \{0,\ldots,\ell\}$ as a
\emph{level} of $\preceq$, and say that $i$ is a \emph{singleton level}
if $|M_i|=1$ and \emph{cycle level} otherwise. 
We observe that if we view $M$ as a multiplicative monoid, then the
equivalence relation $\equiv_f$ induced by the ordered partition is a
congruence of $M$. 

\begin{lemma} \label{supp:preserve-levels}
  Let $p, p', q, q' \in M$. If $f(p)=f(p')$ and $f(q)=f(q')$
  then $f(pq)=f(p'q')$. 
\end{lemma}
\begin{proof}
  By assumption $p \preceq p' \preceq p$, thus $pq \preceq p'q \preceq pq$
  and $f(pq)=f(p'q)$. Repeating the argument we also have $f(p'q)=f(p'q')$. 
\end{proof}

By this lemma, there exists a function
$g \colon \{0,\ldots,\ell\} \times \{0,\ldots,\ell\} \to \{0,\ldots,\ell\}$
such that for any $p, q \in R$ we have $f(pq)=g(f(p),f(q))$.
Clearly, $g$ is symmetric in that $g(i,j)=g(j,i)$.
We note further basic properties of $g$. 

\begin{lemma} \label{supp:merge-cycle-level}
  The function $g$ is non-strictly increasing in each argument. 
  Furthermore, the following hold.
  \begin{enumerate}
  \item If $i<i'$ and $j$ are such that $g(i,j)=g(i',j) < \ell$,
    or symmetrically $g(j,i)=g(j,i') < \ell$, then $g(i,j)$ is a cycle level. 
  \item If $i$ is a cycle level and $j \in \{0,\ldots,\ell\}$
    are such that $g(i,j) < \ell$, then $g(i,j)$ is a cycle level.
  \end{enumerate}
\end{lemma}
\begin{proof}  
  Let $p \in M_i$, $q \in M_{i'}$ and $r \in M_j$ for some $i, i', j$.
  If $i < i'$, then $p \preceq q$ hence $pr \preceq qr$
  so $f(pr)=g(i,j) \leq f(qr)=g(i',j)$. The other case is symmetric.
  Furthermore, if $g(i,j)=g(i',j) < \ell$ then by assumption
  $pr$ and $qr$ are distinct non-zero monomials in level $f(pq)=f(qr)$,
  which is then a cycle level. The other direction is symmetric.
  Finally, let $i$ be a cycle level, let $p, q \in M_i$
  be distinct monomials and let $m \in M_j$. 
  Then $f(pm)=f(qm)=g(i,j)$ but $pm$ and $qm$ are distinct monomials.
  Thus $g(i,j)$ is a cycle level.  
\end{proof}

The function $f$ extends to $R$ the usual way as $f(r)=\min_{m \in r} f(m)$ where $m$
ranges over all monomials occurring in $r$. We refer to $f$
as the \emph{level map} of $\preceq$. 
Similarly, $\preceq$ extends to $R$ by letting
$p \preceq_f q$ if and only if $f(p) \leq f(q)$. 
We show that these retain the properties of $f$ and $\preceq$. 

\begin{lemma} \label{supp:f-lifted}
  Let $R$ be a monomial ring with a weak monomial order $\preceq$ and
  let $f \colon R \to \{0,\ldots,\ell\}$ be the level map of $\preceq$.
  Then for any $p, q \in R$ we have $g(f(p),f(q))=f(pq)$.
  Furthermore $\preceq_f$ is a weak order over $R$ that extends $\preceq$ and
  is preserved by multiplication.
\end{lemma}
\begin{proof}
  Let $p, q \in R$, let $f(p)=i$ and $f(q)=j$,
  and let $m$ respectively $m'$ be the earliest-ordered
  monomial of $M_i$ respectively $M_j$ occurring in $p$
  respectively $q$, according to the lexicographical order.
  By assumption $f(mm')=g(i,j)$. Now either $g(i,j)=\ell$,
  in which case every monomial in the product $p \cdot q$
  is zero in $R$, $pq=0$ holds, and $f(pq)=f(0)=\ell$,
  or otherwise the monomial $mm'$ is produced only once
  in $p \cdot q$, and occurs in the result $pq$ with a non-zero
  coefficient. Thus $pq$ has a monomial $mm'$ with $f(mm')=g(i,j)$,
  and by the properties of $g$ cannot produce a monomial
  $m''$ with $f(m'') < f(mm')$. Thus $f(pq)=g(i,j)$.

  For the rest, clearly $\preceq_f$ is a weak order, which
  by construction agrees with $\preceq$ on monomials.
  Furthermore, it is preserved by multiplication:
  Let $p, q, r \in R$ and assume $p \preceq q$.
  Then by the above $f(pr)=g(f(p),f(r)) \leq 
  f(qr)=g(f(q),f(r))$, hence $pr \preceq qr$.   
\end{proof}

Next, let $\equiv$ be an equivalence relation over $R$. We say
that $\equiv$ is \emph{compatible with $\preceq$} if the following hold:
\begin{enumerate}
\item $\equiv$ respects levels, i.e. if $p \equiv q$ then $f(p)=f(q)$
\item Every equivalence class of $\equiv$ is a coset $p+I$ where $I$
  is a monomial ideal containing all monomials $m$ with $f(m) > f(p)$,
  no monomial $m$ with $f(m) < f(p)$, and not all monomials $m$ with $f(m)=f(p)$
\end{enumerate}
Note that if $i$ is a singleton level, say $M_i=\{m\}$, then the
classes of $\equiv$ within level $i$ are precisely
$am+I$ for $a \in \FF \setminus \{0\}$, where 
$I$ is the monomial ideal containing all monomials $m'$ with
$f(m')>i$.
Recall that $\equiv$ is a \emph{matching} if the following hold.
\begin{itemize}
\item $\equiv$ is a coset partition of $R$
\item $\{0\}$ is an equivalence class of $\equiv$
\item For every $p, q, r \in R$, if $p \equiv q$
  then $pr \equiv qr$
\item For every $p, q, r \in R$, if $p \not \equiv q$
  then either $pr \not \equiv qr$ or $pq=qr=0$
\end{itemize}
Note that the first two conditions are met by assumption. 
We show that in order to verify that $\equiv$ is a matching, it
suffices to check the last two conditions for elements $p, q, r \in R$
such that $f(pr)$ is a cycle level. 

\begin{lemma} \label{lm:verify-bergen}
  Let $R=\FF[\rx_1,\ldots,\rx_n]/I$ be a monomial ring,
  let $\preceq$ be a weak monomial order of $R$, and
  let $f \colon R \to \NN$ be the level map of $\preceq$.
  Let $\equiv$ be an equivalence relation compatible with $\preceq$.
  Assume that for all $p, q, r \in R$ such that $f(pr)$ is a cycle level
  we have $p \equiv q$ if and only if $pr \equiv qr$.
  Then $\equiv$ is a matching.
\end{lemma}
\begin{proof}
  By assumption, $\equiv$ is a coset partition of $R$, and has $\{0\}$
  as an equivalence class. Thus there are only two ways 
  for $\equiv$ to fail at being a matching: either there are elements
  $p, q, r \in R$ such that $p \equiv q$ but $pr \not \equiv qr$,
  or there are elements $p, q, r \in R$ such that
  $p \not \equiv q$ but $pr \equiv qr$ and $pr \neq 0$. 
  We verify these properties.

  Let $M$ be the set of monomials of $R$, plus zero, and 
  let $M_0=\{1\}$, \ldots, $M_\ell=\{\ell\}$ be the ordered partition  
  of $M$ induced by $\preceq$. 
  Let $p, r \in R$ be such that $f(pr)$ is not a cycle level and
  first assume $f(pr)=\ell$, i.e.\ $pr=0$ in $R$.
  By Lemma~\ref{supp:f-lifted}, $f$ is preserved by multiplication.
  Thus if $q \equiv p$ then $f(qr)=f(pr)=\ell$ and $qr=0$, so
  $pr \equiv qr$. On the other hand, if $p \not \equiv q$
  then either $qr \not \equiv pr$ or $pr = qr = 0$.

  We proceed with the more interesting case that $f(pr)<\ell$
  and $f(pr)$ is a singleton level. By Lemma~\ref{supp:merge-cycle-level},
  both $f(p)$ and $f(r)$ are singleton levels as well.
  Let $M_{f(p)}=\{m\}$ and $M_{f(r)}=\{m'\}$.
  Then by the restriction on $\equiv$ for singleton classes,
  we have $p=am+p'$ and $r=cm'+r'$ for some $a, c \in \FF \setminus \{0\}$,
  where $f(p') > f(p)$ and $f(r') > f(r)$.
  Thus $pr=acmm' + x'$ for some $x'$ with $f(x') > f(pr)$,
  by Lemmas~\ref{supp:f-lifted} and~\ref{supp:merge-cycle-level},
  since $f(pr)$ is a singleton level.
  Now let $q \in R$ be such that $p \equiv q$. 
  Then $q=am+q'$, $f(q') > f(q)$, and by the same arguments
  $qr=acmm' + y'$ for some $y'$ with $f(y') > f(pr)$.
  Thus $pr \equiv qr$. 

  Finally, let $q \in R$ be such that $p \not \equiv q$.
  If $f(p) \neq f(q)$, then since $f(pr)$ is not a cycle level,
  $f(pr) \neq f(qr)$ by Lemmas~\ref{supp:f-lifted} and~\ref{supp:merge-cycle-level}
  and $pr \not \equiv qr$. Otherwise we have
  $q=bm+q'$ for some $b \in \FF \setminus \{0,a\}$, $f(q') > f(q)$.
  Then $qr=cbmm' + y'$ for some $y'$ with $f(y') > f(pr)$,
  and since $a, b, c \in \FF$ are elements of a field and $a \neq b$
  we have $ac \neq bc$ and $pr \not \equiv qr$.
  Together with the assumptions made when $f(pr)$ is a cycle level,
  this shows that $\equiv$ meets all conditions of being a matching.
\end{proof}

\subsubsection{$R_{\rm KNT}$ is Bergen}
\label{sec:counter:KNT}

We recall the definition of $R_{\rm KNT}$. 
Define $f(x,y,z)=12x+15y+20z$ and let $I'$ be the ideal containing all monomials $\rx^a \ry^b \rz^c$ such that $f(a,b,c) > 60$.
Then we define the ring $R_{\rm KNT}=\mathbb{Z}_2[\rx,\ry,\rz]/(I',\rx^2\ry\rz)$.

We show that KNT is Bergen using Lemma~\ref{lm:verify-bergen} and computer support.
We define a weak order over the monomials of $R_{\rm KNT}$ as
  \[
    1 \prec \rx \prec \ry \prec \rz \prec \rx^2 \prec \rx\ry \prec \ry^2 \prec \rx\rz
    \prec \rx^3 \prec \ry\rz \prec \rz^2 \prec \rx^2\ry \prec \rx\ry^2 \prec \ry^3 \prec
    \rx^2\rz \prec \rx^4 \prec M'
  \]
where $M'$ is the set of maximal monomials. In addition, the order over $M'$ is refined as
\[
\rx\ry\rz \prec \{\rx^3\ry, \ry^2\rz, \rx\rz^2\} \prec \{\rx^2\ry^2,\ry\rz^2,\rx^3\rz,\rx\ry^3\}
    \prec \rx^5 \prec \ry^4 \prec \rz^3.
\]
The following can be verified by computer.

\begin{proposition}    
    $\preceq$ is a weak monomial order for $R_{\rm KNT}$. 
\end{proposition} 

Next, we define an equivalence relation $\equiv$ compatible with $\preceq$. As noted, $\equiv$ is determined for every singleton level (in fact, it has a single class for every singleton level since $R$ is defined over the field $\integers_2$).
For the cycle levels, we define classes as follows. 
For $\{\rx^3\ry,\ry^2\rz,\rx\rz^2\}$ we get the following classes.
\begin{enumerate}
  \item $C_{2:01*}$ containing elements $r$ with monomial $\ry^2\rz$ but
    not $\rx^3\ry$
  \item $C_{2:1*0}$ containing elements $r$ with monomial $\rx^3\ry$ but
    not $\rx\rz^2$
  \item $C_{2:*01}$ containing elements $r$ with monomial $\rx\rz^2$ but
    not $\ry^2\rz$
  \item $C_{2:111}$ containing elements $r$ with all three monomials
    $\rx^3\ry$, $\ry^2\rz$, $\rx\rz^2$
\end{enumerate}
The second cycle level is broken up similarly. 
With the same notation, indexing the monomials 
in the order $\rx^2\ry^2, \ry\rz^2, \rx^3\rz, \rx\ry^3$,
the next cycle level is broken into classes
$C_{3:10**}$,
$C_{3:00*1}$,
$C_{3:*10*}$,
$C_{3:0*10}$,
$C_{3:*111}$ and
$C_{3:1110}$.
Again, the following is verified by computer (using the property
that only few pairs of monomials produce a monomial in a cycle level).

\begin{proposition}    
    The equivalence relation $\equiv$ meets the matching condition
    at all cycle levels, i.e.\ for $p, q, r \in R_{\rm KNT}$ such
    that $pr$ is in a cycle level we have $p \equiv q$
    if and only if $pr \equiv qr$. 
\end{proposition}

Thus by Lemma~\ref{lm:verify-bergen}, $\equiv$ is a matching. 
Furthermore, since all non-maximal monomials lie in singleton levels,
for any $p, q \in R$ such that $p \equiv q$ and $p, q$ contain non-maximal monomials,
we have $p, q \prec p-q$. Thus we use the order $\prec$ to define a chain of
monomial ideals $I_i$, eliminating one non-maximal monomial at a time,
until only maximal monomials remain. For any such ideal $I_i$ with non-maximal ideals we let 
$\equiv_i$ be $\equiv$ restricted to $I_i$. 
For the level with only maximal monomials, 
the equivalence relation where all classes have cardinality one has the matching property (since multiplication by any non-units annihilates all elements). 
Thus $R_{\rm KNT}$ is Bergen. 

We remark that the cycle levels of $\preceq$ cannot be resolved into singleton levels due to the following statements.
\begin{align*}
  \ry^2 \equiv \ry^2+\rx\rz & \quad \Rightarrow \quad \ry^2\rz \equiv \ry^2\rz+\rx\rz^2 \\
  \rx^3 \equiv \rx^3+\ry\rz & \quad \Rightarrow \quad \rx^3\ry \equiv \rx^3\ry+\ry^2\rz \\
  \rz^2 \equiv \rz^2+\rx^2\ry &\quad \Rightarrow \quad \rx\rz^2 \equiv \rx\rz^2+\rx^3\ry.
\end{align*}
Since $\preceq$ needs to be preserved by multiplication, we thus force
these monomials to be equivalent in $\preceq$. 
There is a similar cycle over the second cycle level.
(The reason that $\equiv$ is still a matching is that the element $\rx^3\ry+\ry^2\rz+\rx\rz^2$ can never be produced via
multiplication by non-units, or more carefully, $C_{2:111}$ is prime.)

\subsubsection{$R_{347}$ is Bergen}
\label{sec:counter:347}

Recall that $R_{347}$ is the monomial ring $R_{347}=\ZZ_2[\rx,\ry,\rz]/I$ 
where $I$ contains every monomial in the upper half-plane spanned by 
$\rx^3$, $\ry^4$, $\rz^7$, inclusive, except that $\rx\ry\rz^3$ is
non-zero. 
That is, $R_{347}=\FF_2[\rx,\ry,\rz]/I$ where $I$ is generated by

\[\{\rx^3,\ry^4,\rz^7,\rx\ry^2\rz^2,\rx\ry^3,\rx\ry\rz^4,\rx\rz^5,
\rx^2\rz^3,\rx^2\ry^2,\rx^2\ry\rz,
\ry\rz^6,\ry^2\rz^4,\ry^3\rz^2\}.\]

This ring is Bergen despite the zero-set $Z_E(R_{347})$ not being hole-free. 

\begin{table}
  \centering
    \begin{tabular}{|r|r|r|}
\hline
      Level & Ann. & Monomials \\
      \hline
      0& 0 & 000  (unit)\\
      1& 1 & 001 \\
      2& 2 & 010 \\
      3& 3 & 002 \\
      4& 4 & 100 \\
      5& 5 & 011 \\
      6& 6 & 003 \\
      7& 7 & 101 \\
      8& 8 & 020 \\
      9& 8 & 012 \\
      10& 9 & 110 \\
      11& 10& 004 \\
      12& 10& 102 \\
      13& 11& 021 \\
      14& 11& 013 \\
      15& 12& 200 \\
      16& 12& 111 \\
      17& 13& 005 \\
      18& 13& 103 \\
      19& 14& 030 \\
      20& 14& 022 \\
      21& 14& 014, 120, 201 \\
      22& 14& 112 \\
      23& 15& 006 \\
      24& 15& 023, 031, 104, 210 \\
      25& 15& 015, 121, 202 \\
      26& 15& 113 \\
    \hline
    \end{tabular}
  \caption{Partial order levels and annihilator indices for $R_{347}$.
    In the table, a monomial $\rx^a\ry^b\rz^c$ is written as $abc$. 
  }
  \label{tab:347}
\end{table}

We use Lemma~\ref{lm:verify-bergen} to show this.
We refer to Table~\ref{tab:347} for a listing of the annihilator
levels and the levels of the weak monomial order $\preceq$ over $R_{347}$.
Monomials are represented in the table by their exponent vectors,
i.e.\ a monomial $m=\rx^a\ry^b\rz^c$ is represented by the triple $abc$ as shorthand.
Again, we use computer support to show that $\preceq$ is preserved by multiplication. 

We define an equivalence relation over $R_{347}$ compatible with $\preceq$.
Again, the classes within singleton levels are fixed, and we describe the cycle levels.

  For a cycle level $i$ containing monomials $m_1, \ldots, m_d$,
  written in this order in Table~\ref{tab:347}, the string
  $i:b_1 \cdots b_d$ represents the set of all elements $r \in R$
  such that $\ell(r)=i$, with restrictions as follows for each $j=1, \ldots, d$:
  \begin{itemize}
  \item If $b_j=0$, then $r$ does not contain $m_j$.
  \item If $b_j=1$, then $r$ contains $m_j$.
  \item If $b_j=*$ then no constraint is imposed on $m_j$ in $r$.    
  \end{itemize}
  Then $C_{i:b_1 \cdots b_d}$ is the set of all elements $r \in R$
  meeting all indicated constraints. Note that this is a coset class
  compatible with $\preceq$. In this notation, the cycle levels are refined as follows.
  \begin{enumerate}
  \item Level $21$ is broken into classes $C_{21:0*1}$, $C_{21:*10}$,
    $C_{21:1*0}$ and $C_{21:111}$
  \item Level $24$ is broken into classes
    $C_{24:**10}$, $C_{24:*10*}$, $C_{24:100*}$, $C_{24:00*1}$,
    $C_{24:0111}$, $C_{24:1011}$ and $C_{24:1111}$
  \item Level $25$ is broken into classes $C_{25:0*1}$, $C_{25:*10}$,
    $C_{25:1*0}$ and $C_{25:111}$
  \end{enumerate}
  It is easily verified that this is a partition of $R$,
  and again, with computer support (or even manually by ad-hoc arguments) 
  we can verify that it meets the precondition of Lemma~\ref{lm:verify-bergen}:
  for any $p, q, r \in R_{347}$ such that $pr$ lies in a cycle level,
  we have $p \equiv q$ if and only if $pr \equiv qr$.
  Thus by Lemma~\ref{lm:verify-bergen}, $\equiv$ is a matching.

The proof is now easily finished. 

\begin{theorem}
  The ring $R_{347}$ is Bergen. 
\end{theorem}
\begin{proof}
  As chain of ideals, we let $I_i$ be the ideal generated by all
  monomials $m$ at level $\ell(m) \geq i$, and for $i \leq 20$ we can
  simply use the equivalence relation $\equiv$ restricted to $I_i$.
  At $I_{21}$, when every remaining monomial $m$ satisfies $\ell(m) \geq 21$, 
  the only non-zero products are $p \cdot \rz$ for $p$ in level 21 or 22.
  At this point, we get a finer partition with matching structure by
  splitting levels 21 and 25 into all seven non-empty classes
  $C_{21:001}$, \ldots, $C_{21:111}$ (and similarly for level 25).
  It should be clear that this retains matching structure, and it also
  meets the condition that if $p \equiv_{21} q$ then $p-q \in I_{22}$.
  At $I_{22}$ and $I_{23}$ we can use either equivalence relation,
  and at $I_{24}$ all remaining elements are sums of maximal monomials,
  hence the only products over them are unit multiplications and
  the finest possible partition into sets of cardinality one
  has matching structure. 
\end{proof}

\subsection{Lineal and not Monomial}

\label{sec:counter:huneke-swanson}

The following is an example from Huneke and Swanson for an integrally closed ideal~\cite[Example~1.3.3]{HunekeS2006integral}.
Let $\mathbb{F}$ be any finite field and define $R=\mathbb{F}[\rx,\ry]/(\rx^2+\ry^3,\rx\ry^3,\ry^4)$.
We claim that this is lineal.

For this, first note that $I=(\rx^2+\ry^3,\rx\ry^3,\ry^4)$ contains precisely those monomials $\rx^a\ry^b$ such that $3a+2b \geq 8$. 
Hence we may write $I'=(\rx^3,\rx^2\ry,\rx\ry^3,\ry^4)$ and $I=I'+(\rx^2+\ry^3)$,
in which case $R'=\mathbb{F}[\rx,\ry]/I'$ is fully convex, by a linear monomial threshold labelling $f(a,b)=3a+2b$ with threshold $T=8$.
We note that the difference between $R$ and $R'$ is not visible under non-unit multiplication.

\begin{lemma}
    Let $r, s \in R$ such that $rs=0$ and neither of $r$ and $s$ is a unit.
    Then $rs=0$ holds also in $R'$. 
\end{lemma}
\begin{proof}
    Define $f(x,y)=3x+2y$ and assume that $r, s \in R$ are non-units such that $rs=0$. 
    Note that if $f(x,y) \geq 6$ then $\rx^x\ry^y$ is a maximal non-zero monomial, hence
    the inclusion of such a term in $r$ or $s$ makes no difference; we consider the case with no such monomials.
    Then we can write
    \[
    rs=(a\ry+b\rx+c\ry^2+d\rx\ry)(p\ry+q\rx+r\ry^2+s\rx\ry)=0=ap\ry^2 + (aq+bp)\rx\ry + (ar+cp-bq)\ry^3 + \ldots,
    \]
    where $\ldots$ has weight at least 7 in $f$. Thus $ap=0$, so $a=0$ or $p=0$, and also $aq+bp=0$.
    In the first case, the second condition requires $b=0$ or $p=0$, in the second case the second condition
    requires $a=0$ or $q=0$. In all cases, the term $(ar+cp-bq)$ has at most one non-zero component.
    Thus no ``cancellation'' is possible here, and the only possible zero products are those
    where all produced monomials occur in $I'$. 
\end{proof}

It immediately follows that $R$ cannot have elements $a, b, c, d \in R$ which violate the magic square property.

\subsection{Non-distributive Annihilator Lattice}
\label{sec:counter:annihilator-lattice}

Recall that the annihilator lattice of a ring is the collection of
annihilator ideals ordered by the subset relation. Thus it has the
operations $\Ann(I) \land \Ann(J) = \Ann(I \cup J)= \Ann(I) \cap \Ann(J)$ 
and $\Ann(I) \lor \Ann(J)$ is the unique smallest annihilator
that contains $\Ann(I) \cup \Ann(J)$. 

On the other hand, the lattice of ideals is the collection of all
ideals ordered by the subset relation, in which case it has operations
$I \land J = I \cap J$ and $I \lor J = I+J$. 
Thus, although the set of annihilators are a subset of the set of
ideals, they do not form a sublattice since normally $\Ann(I)+\Ann(J)$
will not be an annihilator ideal. 

Let $R=\FF[\rx,\ry,\rz]/((\rx,\ry,\rz)^4,\rx\ry\rz)$.
This is Z-hole-free, hence Helly, but the annihilator lattice contains
a diamond sublattice.

\begin{lemma}
  $R$ is $Z$-hole-free.
\end{lemma}
\begin{lemma}
  $Z_E(R) \subset \naturals^3$ contains all points of Hamming weight at
  least 4 and in addition the point $(1,1,1)$. Consider a convex
  combination $\sum_i \alpha_i p_i$ of points $p_i \in Z$,
  producing a point $q \in \naturals^3$. 
  Then either this ``combination'' is just $1 \cdot (1,1,1)$
  and produces $q=p$, or the Hamming weight of $q$ is at least 4.
\end{lemma}

It follows that $R$ has the no cancellations property and is Helly.
By the former, in order to capture all annihilators of $R$ it suffices
to consider $\Ann(m)$ for monomials $m$ in $R$.
These are as follows. Let $M=(\rx,\ry,\rz)$. 

\begin{lemma}
  Let $m=\rx^a\ry^b\rz^c$ be a monomial in $R$.
  Then $\Ann(m)$ contains $M^{4-a-b-c}$, and in addition
  if $\max(a,b,c) \leq 1$ then it contains the monomial
  $\rx^{1-a}\ry^{1-b}\rz^{1-c}$. It contains no other monomials.
\end{lemma}
\begin{proof}
  For any monomial $m$, viewed as a point $(a,b,c) \in \naturals^3$,
  the generators of annihilators of $m$ are simply the monomials
  $\rx^d\ry^e\rz^f$ such that $(a+d,b+e,c+f) \in Z_E(R)$.
  The claim follows easily.
\end{proof}

We can now prove the negative result.

\begin{lemma}
  The annihilators $\Ann(\rx\ry)$, $\Ann(\rx\rz)$, $\Ann(\ry\rz)$
  generate a diamond sublattice in the annihilator lattice of $R$. 
\end{lemma}
\begin{proof}
  Since each of these is multilinear of degree 2, they contain $M^2$
  and in addition respectively $\rz$, $\ry$ and $\rx$. Their pairwise
  intersections are simply $M^2$. Finally, any pair of them yields 
  the annihilator $M=\Ann(M^3)$ since it is the only annihilator that
  contains two distinct linear monomials.
\end{proof}

Since $R$ is not a chain ring, the lattice of ideals itself is also
not distributive. For example, the elements $(\rx^3)$, $(\ry^3)$ and
$(\rx^3+\ry^3)$ induce a diamond sublattice. However, naturally,
these are not annihilator ideals.

\newpage


\newcommand{\etalchar}[1]{$^{#1}$}

\end{document}

%% file: figures/z8-1.tex
\begin{tikzpicture}[node distance=0.5cm]
    \tikzstyle{no}=[draw,circle, inner sep=1pt, fill]
    \tikzstyle{ed}=[draw,color=black, line width=1pt]

    \draw
    node[no, label=left:$x_1$] (l1) {}
    node[no, below of=l1, label=left:$x_2$] (l2) {}
    node[no, below of=l2, label=left:$x_3$] (l3) {}
    node[no, below of=l3, label=left:$x_4$] (l4) {}
    node[no, below of=l4, label=left:$x_5$] (l5) {}
    node[no, below of=l5, label=left:$x_6$] (l6) {}
    node[no, below of=l6, label=left:$x_7$] (l7) {}
    ;
    \draw
    (l1) +(2cm,0cm) node[no, label=right:$y_1$] (r1) {}
    node[no, below of=r1, label=right:$y_2$] (r2) {}
    node[no, below of=r2, label=right:$y_3$] (r3) {}
    node[no, below of=r3, label=right:$y_4$] (r4) {}
    node[no, below of=r4, label=right:$y_5$] (r5) {}
    node[no, below of=r5, label=right:$y_6$] (r6) {}
    node[no, below of=r6, label=right:$y_7$] (r7) {}
    ;

    \draw
    (l1) edge[ed] (r1)
    (l2) edge[ed] (r2)
    (l3) edge[ed] (r3)
    (l4) edge[ed] (r4)
    (l5) edge[ed] (r5)
    (l6) edge[ed] (r6)
    (l7) edge[ed] (r7)
    ;
\end{tikzpicture}

%% file: figures/z8-2.tex
\begin{tikzpicture}[node distance=0.5cm]
    \tikzstyle{no}=[draw,circle, inner sep=1pt, fill]
    \tikzstyle{ed}=[draw,color=black, line width=1pt]

    \draw
    node[no, label=left:$x_1$] (l1) {}
    node[no, below of=l1, label=left:$x_2$] (l2) {}
    node[no, below of=l2, label=left:$x_3$] (l3) {}
    node[no, below of=l3, label=left:$x_4$] (l4) {}
    node[no, below of=l4, label=left:$x_5$] (l5) {}
    node[no, below of=l5, label=left:$x_6$] (l6) {}
    node[no, below of=l6, label=left:$x_7$] (l7) {}
    ;
    \draw
    (l1) +(2cm,0cm) node[no, label=right:$y_1$] (r1) {}
    node[no, below of=r1, label=right:$y_2$] (r2) {}
    node[no, below of=r2, label=right:$y_3$] (r3) {}
    node[no, below of=r3, label=right:$y_4$] (r4) {}
    node[no, below of=r4, label=right:$y_5$] (r5) {}
    node[no, below of=r5, label=right:$y_6$] (r6) {}
    node[no, below of=r6, label=right:$y_7$] (r7) {}
    ;

    \draw
    (l2) edge[ed] (r1)
    (l4) edge[ed] (r2)
    (l6) edge[ed] (r3)
    (l2) edge[ed] (r5)
    (l4) edge[ed] (r6)
    (l6) edge[ed] (r7)
    ;
\end{tikzpicture}

%% file: figures/z8-3.tex
\begin{tikzpicture}[node distance=0.5cm]
    \tikzstyle{no}=[draw,circle, inner sep=1pt, fill]
    \tikzstyle{ed}=[draw,color=black, line width=1pt]

    \draw
    node[no, label=left:$x_1$] (l1) {}
    node[no, below of=l1, label=left:$x_2$] (l2) {}
    node[no, below of=l2, label=left:$x_3$] (l3) {}
    node[no, below of=l3, label=left:$x_4$] (l4) {}
    node[no, below of=l4, label=left:$x_5$] (l5) {}
    node[no, below of=l5, label=left:$x_6$] (l6) {}
    node[no, below of=l6, label=left:$x_7$] (l7) {}
    ;
    \draw
    (l1) +(2cm,0cm) node[no, label=right:$y_1$] (r1) {}
    node[no, below of=r1, label=right:$y_2$] (r2) {}
    node[no, below of=r2, label=right:$y_3$] (r3) {}
    node[no, below of=r3, label=right:$y_4$] (r4) {}
    node[no, below of=r4, label=right:$y_5$] (r5) {}
    node[no, below of=r5, label=right:$y_6$] (r6) {}
    node[no, below of=r6, label=right:$y_7$] (r7) {}
    ;

    \draw
    (l3) edge[ed] (r1)
    (l6) edge[ed] (r2)
    (l1) edge[ed] (r3)
    (l4) edge[ed] (r4)
    (l7) edge[ed] (r5)
    (l2) edge[ed] (r6)
    (l5) edge[ed] (r7)
    ;
\end{tikzpicture}

%% file: figures/z8-4.tex
\begin{tikzpicture}[node distance=0.5cm]
    \tikzstyle{no}=[draw,circle, inner sep=1pt, fill]
    \tikzstyle{ed}=[draw,color=black, line width=1pt]

    \draw
    node[no, label=left:$x_1$] (l1) {}
    node[no, below of=l1, label=left:$x_2$] (l2) {}
    node[no, below of=l2, label=left:$x_3$] (l3) {}
    node[no, below of=l3, label=left:$x_4$] (l4) {}
    node[no, below of=l4, label=left:$x_5$] (l5) {}
    node[no, below of=l5, label=left:$x_6$] (l6) {}
    node[no, below of=l6, label=left:$x_7$] (l7) {}
    ;
    \draw
    (l1) +(2cm,0cm) node[no, label=right:$y_1$] (r1) {}
    node[no, below of=r1, label=right:$y_2$] (r2) {}
    node[no, below of=r2, label=right:$y_3$] (r3) {}
    node[no, below of=r3, label=right:$y_4$] (r4) {}
    node[no, below of=r4, label=right:$y_5$] (r5) {}
    node[no, below of=r5, label=right:$y_6$] (r6) {}
    node[no, below of=r6, label=right:$y_7$] (r7) {}
    ;

    \draw
    (l4) edge[ed] (r1)
    (l4) edge[ed] (r3)
    (l4) edge[ed] (r5)
    (l4) edge[ed] (r7)
    ;
\end{tikzpicture}

%% file: figures/diag-monom.tikz
  \begin{tikzpicture}[auto,
    > = Stealth, 
    node distance = 10mm and 20mm,
    box/.style = {draw=gray, very thick,
      minimum height=11mm, text width=22mm, 
      align=center},
    every edge/.style = {draw, ->, very thick},
    every edge quotes/.style = {font=\footnotesize, align=center, inner sep=1pt},
    counter/.style = {bend=left, color=red}
    ]
    \node (n1) [box]               {Fully Convex};
    \node (nf) [box, below=of n1]  {Fields};
    \node (n2) [box, above=of n1]  {Bergen};
    \node (n3) [box, above=of n2]  {Lineal, Property P};
    \node (n4) [box, above=of n3]  {No Cancellations};
    \node (n5) [box, above=of n4]  {Helly};
    \node (n6) [box, above=of n5]  {Z-1-Convex};
    \node (nz) [box, above left=of n2] {Z-Hole-Free};
    \draw
    (nf) edge 
    (n1)
    (n1) edge ["Lm.~\ref{lm:fully-convex-to-bergen}"] (n2)
    (n2) edge ["Prop.~\ref{prop:bergenislineal}"] (n3)
    (n3) edge ["Lm.\ref{lm:lineal-implies-nc}"] (n4)
    (n4) edge ["Lm.~\ref{lm:no-cancel-helly}"] (n5)
    (n1.north west) edge ["Lm.~\ref{lm:characterize-fullyconvex}"] (nz)
    (nz) edge ["Lm.~\ref{lm:zhf-to-nc}"] (n4)
    (n5) edge ["Lm.~\ref{lm:helly-to-z1c}"] (n6);    
    \draw[color=red,dashed,bend left]
    (n2.south east) edge ["KNT{,} 347"] (n1.north east)
    (n4.south east) edge ["APX"] (n3.north east)
    ;
    \draw[color=red,dashed]    
    (n2.north west) edge ["347"] (nz.south east)
    (n6.south west) edge ["347"] (nz)
    (nz) edge ["APX"] (n3);
    \draw[dotted, bend left]
    (n3.south east) edge ["Open"] (n2.north east)
    (n6.south east) edge ["Open"] (n4.north east);
    
  \end{tikzpicture}